\documentclass{article}
\usepackage{url}
\usepackage{graphicx}
\usepackage{amsfonts}
\usepackage{amsmath}
\usepackage{amsthm}
\usepackage{amssymb}
\usepackage{latexsym}
\usepackage{color}
\usepackage{multirow}
\usepackage[normalem]{ulem}
\usepackage{refcount}
\usepackage{hyperref}
\usepackage{cleveref}
\usepackage{thmtools,thm-restate}
\usepackage[title]{appendix}

\newcommand{\notyet}[1]{}
\newcommand{\squeezelist}{\setlength{\itemsep}{0pt}}

\newtheorem{thm}{Theorem}
\newtheorem{theorem}{{\bf Theorem}}
\newtheorem{lemma}[thm]{Lemma}
\newtheorem{proposition}[thm]{Proposition}
\newtheorem{definition}[thm]{Definition}

\newcommand\connect[1]{\texttt{connect}}
\newcommand\T{{\mathcal T}}
\newcommand\N{{\mathcal N}}
\newcommand\XE{{\mathcal E}}
\newcommand\XX{{\mathcal X}}
\newcommand\OO{{\mathcal O}}
\newcommand\hand{{\sc Hand}}
\newcommand\head{{\sc Head}}

\title{Unfolding Polycube Trees with Constant Refinement\thanks{Partial results for polycube trees (previously called \emph{orthotrees}) of degree 3 or less have appeared in~\cite{DF18}.}}
\author{
Mirela Damian\thanks{Department of Computer Science, Villanova University, Villanova, PA, \tt{mirela.damian@villanova.edu}}
\and
Robin Flatland\thanks{Department of Computer Science, Siena College, Loudonville, NY, \tt{flatland@siena.edu}}
}
\begin{document}
\thispagestyle{empty}
\date{}
\maketitle

\begin{abstract}
We show that every polycube tree can be unfolded with a $4 \times 4$ refinement of the grid faces. 
This is the first constant refinement unfolding result for polycube trees that are not required to be well-separated.  
\end{abstract}

\section{Introduction}
\label{sec:intro}
An \emph{unfolding} of a polyhedron is obtained by cutting its surface in such a way that it can be flattened
in the plane as a simple non-overlapping polygon called a \emph{net}. An \emph{edge unfolding} allows only cuts  
along the polyhedron's edges, while a \emph{general unfolding} 
allows cuts 
anywhere on the surface.
Edge cuts alone are not sufficient to
guarantee an unfolding for non-convex polyhedra~\cite{Bern-Demaine-Eppstein-Kuo-Mantler-Snoeyink-2003,BDDLOORW1998},
however it is unknown whether all non-convex
polyhedra have a general unfolding.
In contrast, all convex polyhedra
have a general unfolding~\cite[Sec.~24.1.1]{Demaine-O'Rourke-2007}, but it is 
unknown whether they all have an edge unfolding~\cite[Ch.~22]{Demaine-O'Rourke-2007}.

Prior work on unfolding algorithms for non-convex objects has focused on orthogonal polyhedra. 
This class consists of polyhedra whose edges and faces all meet at right angles.
Because not all orthogonal polyhedra have
edge unfoldings~\cite{BDDLOORW1998}, the unfolding algorithms typically
use additional non-edge cuts that  
follow one of two models.
In the \emph{grid unfolding model}, the surface is
subdivided into rectangular \emph{grid faces} by adding
edges where axis-perpendicular planes through each vertex
intersect the surface, and cuts along these added edges
are also allowed.       
In the \emph{grid refinement model}, each grid face
under the grid unfolding model is further subdivided by
an $(a \times b)$ orthogonal grid, for some positive integers $a, b \ge 1$, and cuts are also allowed
along any of these grid lines.  
%

A series of algorithms have been developed for unfolding arbitrary genus-$0$ orthogonal polyhedra, with each successive algorithm requiring less grid refinement.  
The first such algorithm~\cite{Damian-Flatland-O'Rourke-2007-epsilon} required an exponential amount of grid refinement. This was reduced to quadratic refinement in~\cite{Damian-Demaine-Flatland-2014-delta}, and then to linear in~\cite{Chang2015}.   These ideas were further extended in~\cite{Damian-Demaine-Flatland-2017-O'Rourke-genus2} to unfold arbitrary genus-$2$ orthogonal polyhedra with linear refinement.  

The only unfolding algorithms for orthogonal polyhedra that use sublinear refinement are for specialized orthogonal shape classes.  
For example, there exist algorithms for unfolding orthostacks using $1 \times 2$ refinement~\cite{BDDLOORW1998} and  Manhattan Towers using $4 \times 5$ refinement~\cite{Damian-Flatland-O'Rourke-2005-manhattan}.  There also exist unfolding algorithms for several
classes of polyhedra composed of rectangular boxes.  For example, orthotubes~\cite{BDDLOORW1998} and one layer block structures~\cite{Liou-Poon-Wei-2014-onelayer} built of unit cubes with an arbitrary number of unit holes can both be unfolded with cuts restricted to the box edges.  
%
\begin{figure}[ht]
\centering
\includegraphics[page=1,width=\linewidth]{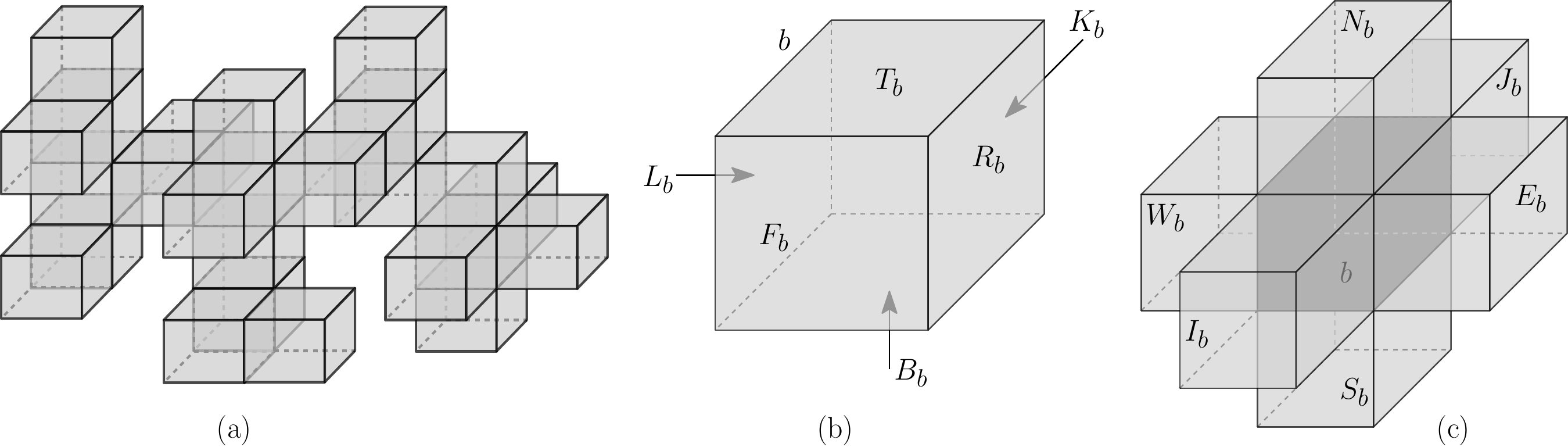}
\caption{(a) A simple polycube tree example. Notation for: (b) $b$'s faces (c) $b$'s neighbors.}
\label{fig:orthotree}
\end{figure}
%
Our focus here is on a class of orthogonal polyhedra known as polycube trees.
A \emph{polycube tree} $\OO$ is composed of axis-aligned unit cubes (boxes) glued face to face, whose surface 
is a $2$-manifold and whose dual graph $\T$ is a tree. 
(See~\autoref{fig:orthotree}a for an example.)
In the grid unfolding model, cuts are allowed along any of the 
cube edges. 
Each node in $\T$ is a box in $\OO$ and two nodes are connected by an edge 
if the corresponding boxes are adjacent  in $\OO$ (i.e., if they share a face). 
In this paper we will use the terms \emph{box} and \emph{node} interchangeably. 
The \emph{degree} of a box $b \in \OO$ is defined as the degree of its corresponding node in the dual tree $\T$. 
We select any node of degree one to be the \emph{root} of $\T$. 

In a polycube tree, each box can be classified as either a \emph{leaf}, a \emph{connector}, or a \emph{junction}.
A leaf is a box of degree one; a connector is a box of degree two whose two adjacent boxes
are attached on opposite faces; all other boxes are junctions.  

Because polycube trees are orthogonal polyhedra, they can be unfolded using the general algorithm in~\cite{Chang2015}
with linear refinement. 
Algorithms for unfolding polycube trees using less than linear refinement
have been limited to polycube trees that are
\emph{well-separated}, meaning that no two junction boxes are adjacent. 
In~\cite{Damian-Flatland-Meijer-O'Rourke-2005-orthotrees}, the authors provide an algorithm for grid unfolding well-separated polycube trees.  
Recent work in~\cite{Ho-Chang-Yen-2017-orthographs} shows that the related class of well-separated orthographs (which allow arbitrary genus)
can be unfolded with a $2 \times 1$ refinement.  

In this paper we provide an algorithm for unfolding all polycube trees 
using a $4 \times 4$ refinement of the cube faces. 
For each box $b$ in $\T$, the algorithm unfolds $b$ and the boxes in the subtree rooted at $b$ recursively. 
Intuitively, the algorithm unfolds surface pieces of $b$ along a carefully constructed
path. When the path reaches a child box of $b$, the child is recursively unfolded and then the path continues on
$b$ again to the next child (if there is one).  The unfolding of $b$ and its subtree is contained within a rectangular region having
two staircase-like bites taken out of it.
This is the first sublinear refinement unfolding result for the class of all polycube trees,
regardless of whether they are well-separated or not.

\section{Terminology}

For any box $b \in \OO$, 
$R_b$ and $L_b$ are the \emph{right} and \emph{left} faces of $b$ (orthogonal to the $x$-axis);
$F_b$ and $K_b$ are the \emph{front} and \emph{back} faces of $b$ (orthogonal to the $z$-axis); and
$T_b$ and $B_b$ are the \emph{top} and \emph{bottom} faces of $b$ (orthogonal to the $y$-axis).
See~\autoref{fig:orthotree}b.
We use a different notation for boxes adjacent to $b$, to clearly distinguish them from faces: 
$E_b$ and $W_b$ are the \emph{east} and \emph{west} neighbors of $b$ (adjacent to
$R_b$ and $L_b$, resp.);
$N_b$ and $S_b$ are the \emph{north} and \emph{south} neighbors of $b$ (adjacent to
$T_b$ and $B_b$, resp.); and
$I_b$ and $J_b$ are the \emph{front} and \emph{back} neighbors of $b$ (adjacent to
$F_b$ and $K_b$, resp.). See~\autoref{fig:orthotree}c. 
We omit the subscript whenever the box $b$ is clear from the context. 
We use combined notations to refer to the east neighbor of $N$ as $NE$, the back neighbor of $NE$ as $NEJ$, and so on.

%
\begin{figure*}[t]
\centering
\includegraphics[width=\linewidth]{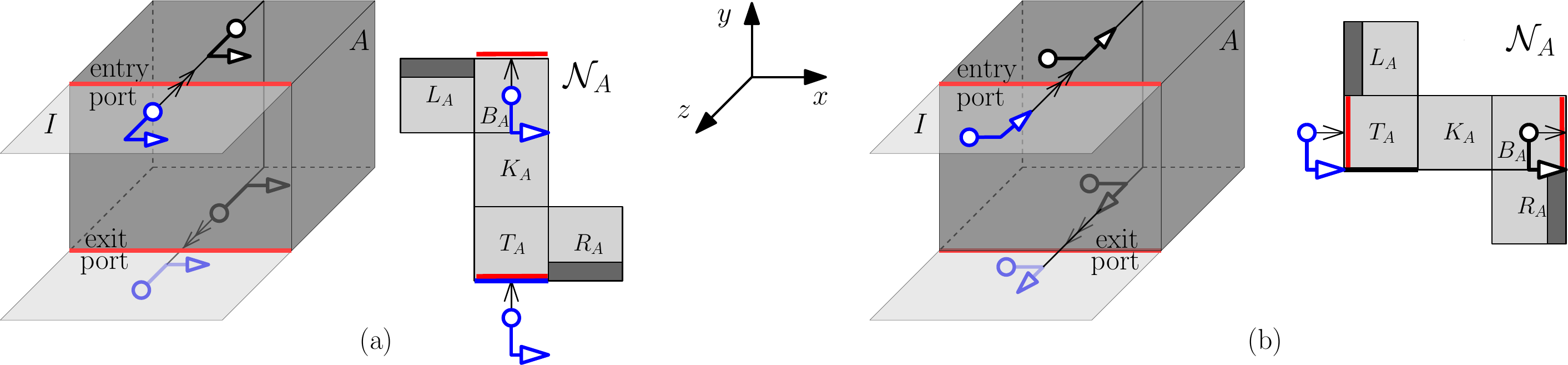}
\caption{(a) \head-first and (b) \hand-first unfolding of leaf box; dark-shaded pieces can be removed without disconnecting the nets.}
\label{fig:degree1}
\end{figure*}
%

If a face of a box $b \in \OO$  is also a face of $\OO$, we call it an \emph{open face}; otherwise, we call it a \emph{closed face}. 
On the closed face shared by $b$ with its parent box in $\T$, we identify a pair of opposite edges, one called the \emph{entry port} and the other called the \emph{exit port} (shown in red and labeled in \autoref{fig:degree1}).  The unfolding of $b$ is determined by an \emph{unfolding path} that starts on $b$ near $b$'s entry port, recursively 
visits all boxes in the subtree $\T_b \subseteq \T$ rooted at $b$, and ends on $b$ near $b$'s exit port. 
We denote by $\N_b$ the unfolding net produced by a recursive unfolding of $b$. For simplicity, we will sometimes omit the word ``recursive'' when referring to a recursive unfolding of a box $b$ and simply call it an \emph{unfolding of} $b$, with the understanding that all boxes in $\T_b$ get unfolded during the process. 

To make it easier to visualize the unfolding path, we use an $L$-shaped guide (or simply $L$-guide) with two orthogonal pointers, namely a \emph{\hand\ pointer} and a \emph{\head\ pointer}. (See~\autoref{fig:degree1}, where the  \head\ and the \hand\ pointers are represented by the circle and the arrow, respectively.) 
With very few exceptions, the unfolding path extends in the direction of one of the two pointers. Whenever the unfolding path follows the direction of the \hand, we say that it extends \hand\emph{-first}; 
otherwise, it extends \head\emph{-first}.  
%
Surface pieces traversed in the direction of the \hand\  (\head) will flatten out horizontally (vertically) in the plane.

As a simple example, consider the unfolding of a leaf box $A$ from~\autoref{fig:degree1}a. The $L$-guide is shown positioned on top of $A$'s parent box $I$ at the entry port. 
The unfolding path extends \head-first across the top, back, and bottom faces of $A$, and ends on the bottom of $A$ at the exit port.   
The resulting unfolding 
net $\N_A$ consists of $A$'s open faces $T_A$, $K_A$, $B_A$, $L_A$, and $R_A$. In all unfolding
illustrations, the outer surface of $\OO$ is shown. 
When describing and illustrating the unfolding of a box $A$, we will assume without loss of generality that the box is in \emph{standard position} (as in \autoref{fig:degree1}a), with its parent $I_A$ attached to its front face $F_A$ and its entry (exit) port on the
top (bottom) edge of $F_A$.

%
\begin{figure}[ht]
\centering
\includegraphics[page=2,width=\linewidth]{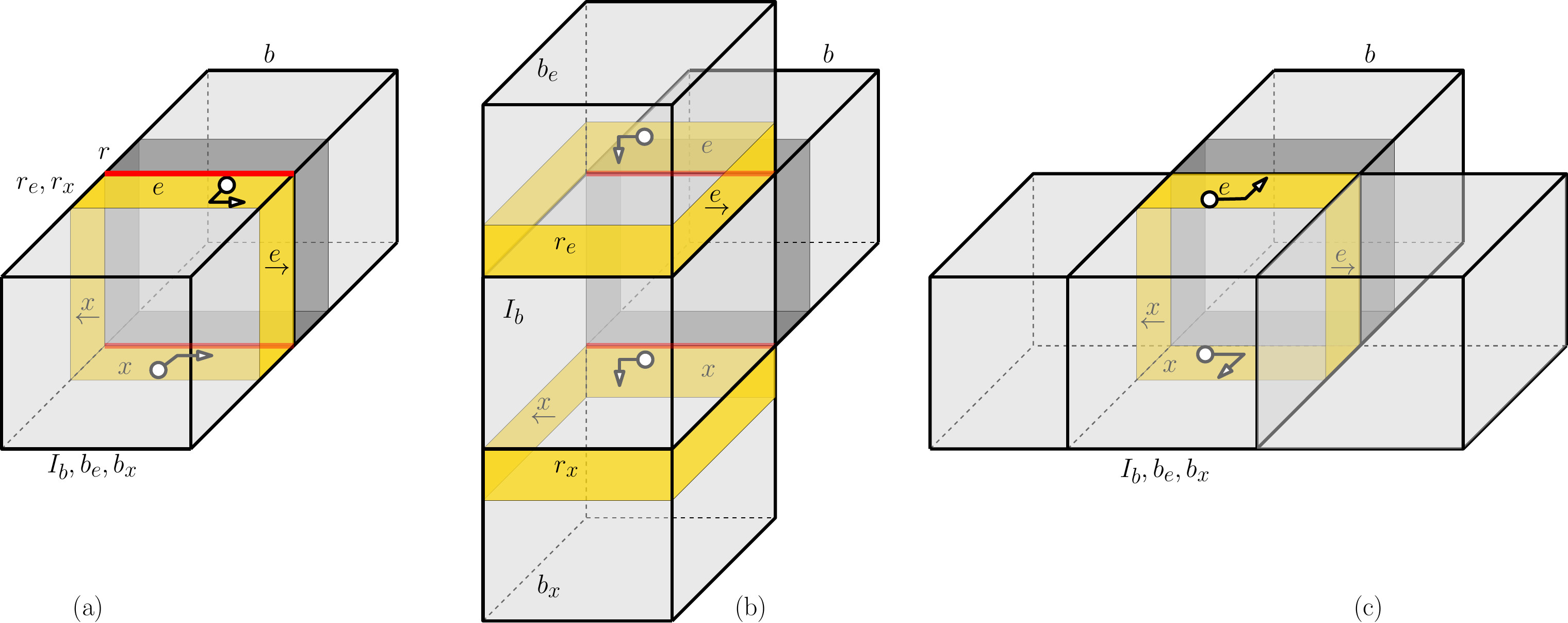}
\caption{ Box $b$ in standard position with parent $I_b$ and ring $r$ (a) entry and exit boxes $b_e$, $b_x$ coincide with parent $I_b$; entry ring $r_e$ coincides with exit ring $r_x$; entry ring face $e \in r_e$ is the top and its successor $\xrightarrow{e} \in r_e$ is the right face of $r_e$;  exit ring face $x \in r_x$ is the bottom and its predecessor $\xleftarrow{x} \in r_x$ is the left face of $r_x$
(b) entry box $b_e$ with entry ring $r_e$ and entry face $e$ lies north of $I_b$; 
exit box $b_x$ with exit ring $r_x$ and exit face $x$ lie south of $I_b$; $\xrightarrow{e}$ is the successor of $e$ on the entry ring $r_e$, and $\xleftarrow{x}$ is the predecessor of $x$ on the exit ring $r_x$ (c) $\xrightarrow{e}$ and $\xleftarrow{x}$ are closed ($e$ and $x$ are always open, by definition).}
\label{fig:rings}
\end{figure}
%
The \emph{ring} $r$ of a box $b$ includes  all the points on the surface of $b$ (not necessarily on the 
surface of $\OO$) that are within distance $1/4$ of the closed face shared with $b$'s parent.
Thus, $r$ consists of four $1/4 \times 1$ rectangular pieces (which we call \emph{ring faces}) connected in a cycle.
(See~\autoref{fig:rings}a, where $r$ is the shaded band on $b$'s surface wrapping around $b$'s front face; box $b$ is shown in standard position, so its parent $I_b$ attaches to $b$'s front face.)
The \emph{entry box} $b_e$ of $b$ is the box containing the open face in $\T \setminus \T_b$ adjacent to $b$'s entry port.
Note that $b_e$ may be $b$'s parent (as in~\autoref{fig:rings}a),
but this is not necessary (see~\autoref{fig:rings}b, where $b_e$ is the box on top of $b$'s parent $I_b$). 

The \emph{entry ring} $r_e$ of $b$ includes all points of $b_e$ that are 
within distance $1/4$ of the closed face of $b_e$ adjacent to $b$'s entry port. 
(Refer to~\autoref{fig:rings}.) 
The face $e$ of $r_e$ adjacent to $b$'s entry port is the \emph{entry ring face}.
Similarly, the \emph{exit box} $b_x$ of $b$ is the box containing the open face in $\T \setminus \T_b$ adjacent to $b$'s exit port.
Note that $b_x$ may be $b$'s parent (as in~\autoref{fig:rings}a), 
but this is not necessary (see~\autoref{fig:rings}b, where $b_x$ is the box south of $b$'s parent $I_b$). 
The \emph{exit ring} $r_x$ of $b$ includes all points of $b_x$ that are 
within distance $1/4$ of the closed face of $b_x$ adjacent to $b$'s exit port. 
The face $x$ of $r_x$ adjacent to $b$'s exit port is the \emph{exit ring face}.
Note that both $e$ and $x$ are \emph{open} ring faces (by definition).  When unclear from context, we will use subscripts (i.e., $e_b$ and $x_b$) to specify the entry and exit faces of a particular box $b$.

In a \head-first unfolding of a box $b$, the $L$-guide begins on the entry ring face $e$ with the \head\ 
pointing toward the entry port, and it ends on the exit ring face $x$ with the \head\ pointing away from the exit port; the \hand\ has the same orientation at the start and end of the unfolding. (See Figures~\ref{fig:degree1}a,~\ref{fig:rings}a.) 
Similarly, in a \hand-first unfolding, the $L$-guide begins on the entry ring face $e$ with the \hand\ pointing toward the entry port, and it ends on the exit ring face $x$ with the \hand\ pointing away from the exit port; the \head\ has the same orientation at the start and end of the unfolding. (See Figures~\ref{fig:degree1}b,~\ref{fig:rings}b.)
In standard position, the \hand\ in a \head-first unfolding will point either east or west.  If it points east (west) we say that the unfolding is a \hand-east (west), \head-first unfolding.  Similarly, in a \hand-first unfolding, the \head\ will either point east or west. If it points east (west), we say the unfolding is a \head-east (west), \hand-first unfolding. 

In a \head-first (\hand-first) unfolding of $b$ with entry ring face $e$, $\xrightarrow{e}$ is the 
ring face of $r_e$ encountered immediately after $e$ when cycling around $r_e$ in the direction pointed to by the \hand~(\head) of the $L$-guide as positioned on $e$ at the start of $b$'s unfolding.   Similarly, 
in a \head-first (\hand-first) unfolding of $b$ with exit ring face $x$, $\xleftarrow{x}$ is the ring face of $r_x$ encountered just before $x$ when cycling around $r_x$ in the direction pointed to by the \hand~(\head) of the $L$-guide as positioned on $x$  at the end of $b$'s unfolding path.~\autoref{fig:rings}  shows
$\xrightarrow{e}$ and $\xleftarrow{x}$ labeled.  Note that, although $e$ and $x$ are
open ring faces by definition, $\xrightarrow{e}$ and $\xleftarrow{x}$ may be closed 
(see~\autoref{fig:rings}c for an example).

\section{Inductive Regions}
\noindent
Let $b \in \T$ be an box to be unfolded recursively.

%

\begin{definition}
\label{def:region}
\emph{
A \emph{\head-first inductive region} for $b$ is a rectangle at least three units wide and three units tall, with two staircase bites taken out of the lower left and upper right corners, as shown in~\autoref{fig:regions}a.
The \emph{entry (exit)} port of the inductive region is the lower left (upper right) horizontal segment that lies strictly inside the bounding box of the region. If $b$ is not a leaf, 
the unit cells labeled $\XE_b$ and $\XX_b$ in~\autoref{fig:regions}a are conditionally included in the inductive region as follows: 
\begin{itemize}
\item If the successor $\xrightarrow{e}$ of the entry ring face $e$ is closed, 
then $\XE_b$ is included as part of the inductive region, otherwise, $\XE_b$ is not part of the inductive region. 
In the latter case, we refer to the unit segment right of the entry port as the \emph{entry port extension}.
\item If the predecessor $\xleftarrow{x}$ of the exit ring face $x$ is closed, then 
$\XX_b$ is included as part of the inductive region, 
otherwise, $\XX_b$ is not part of the inductive region. 
In the latter case, we refer to the unit segment left of the exit port as the \emph{exit port extension}.
\end{itemize}
}
\end{definition}


See~\autoref{fig:netconnections} for a few examples. A \head-first unfolding of $b$ produces a net $\N_b$ that fits within the \head-first inductive region and whose entry (exit) port aligns to the left (right) with the entry (exit) port of the inductive region.

%
\begin{figure}[ht]
\centering
\includegraphics[page=3,width=\linewidth]{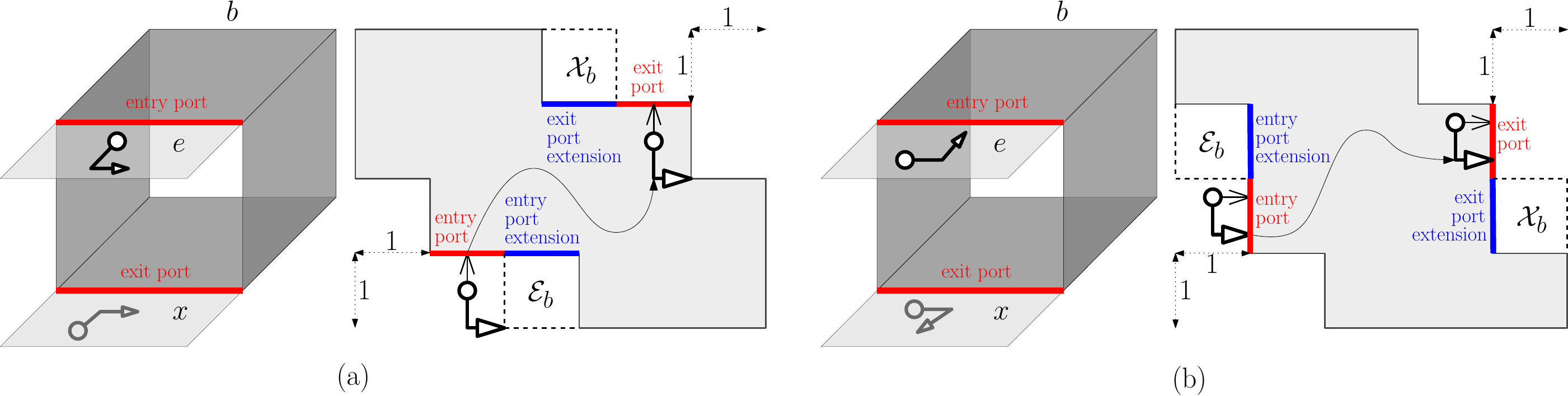}
\caption{Inductive region for (a) \head-first unfolding (b) \hand-first unfolding. 
}
\label{fig:regions}
\end{figure}
%

A \emph{\hand-first inductive region} for $b$ is an orthogonally convex polygon shaped as in~\autoref{fig:regions}b. 
Its shape is isometric to that of a \head-first inductive region, and one can be obtained from the other through a clockwise $90^{\circ}$-rotation, followed by a vertical reflection. The unit cells $\XE_b$ and $\XX_b$ in~\autoref{fig:regions}b are conditionally included in the inductive region according to the rules stated in~\autoref{def:region}. 

\begin{lemma}
Let $b$ be an arbitrary box in $\OO$, and let $d$ be the box corresponding to $b$ in a horizontal reflection of $\OO$. Let $\N_d$ be the unfolding net produced by a \head-first unfolding of $d$. If $\N_{d}$ is rotated counterclockwise by $90^{\circ}$ and then reflected horizontally, then the result is a \hand-first unfolding of $b$. 
\label{lem:hand}
\end{lemma}
\begin{proof}
First note that, when applied to the L-guide, the combined ($90^\circ$-rotation, reflection) transformation switches the \head\ and \hand\ positions. This implies that the successor $\xrightarrow{e}$ of $d$'s entry ring face is the same before and after the combined ($90^\circ$-rotation, reflection) transformation, because it extends in the direction of the \hand~(\head) in a \head-first (\hand-first) unfolding. Similarly, the predecessor $\xleftarrow{x}$ of $d$'s exit ring face is the same before and after the transformation. Thus the rules from~\autoref{def:region} for including $\XE_d$ and $\XX_d$ in the inductive region for $d$ refer to the same ring faces before and after the transformation. These together show that, when applied to the unfolding net, this transformation turns a \head-first recursive unfolding of $d$ into a \hand-first recursive unfolding of $b$.
\end{proof}

\medskip
Lemma~\ref{lem:hand} enables us to focus the rest of the paper on \head-first unfoldings only, with the understanding that the results transfer to \hand-first unfoldings. 

\section{Net Connections}
We now discuss the type of connections that each \head-first unfolding 
net $\N_b$ associated with a box $b$ must provide to ensure that it connects to the rest of $\T$'s unfolding. 
To do so, we need a few more definitions.

\begin{figure*}[ht]
\centering
\includegraphics[page=4,width=\linewidth]{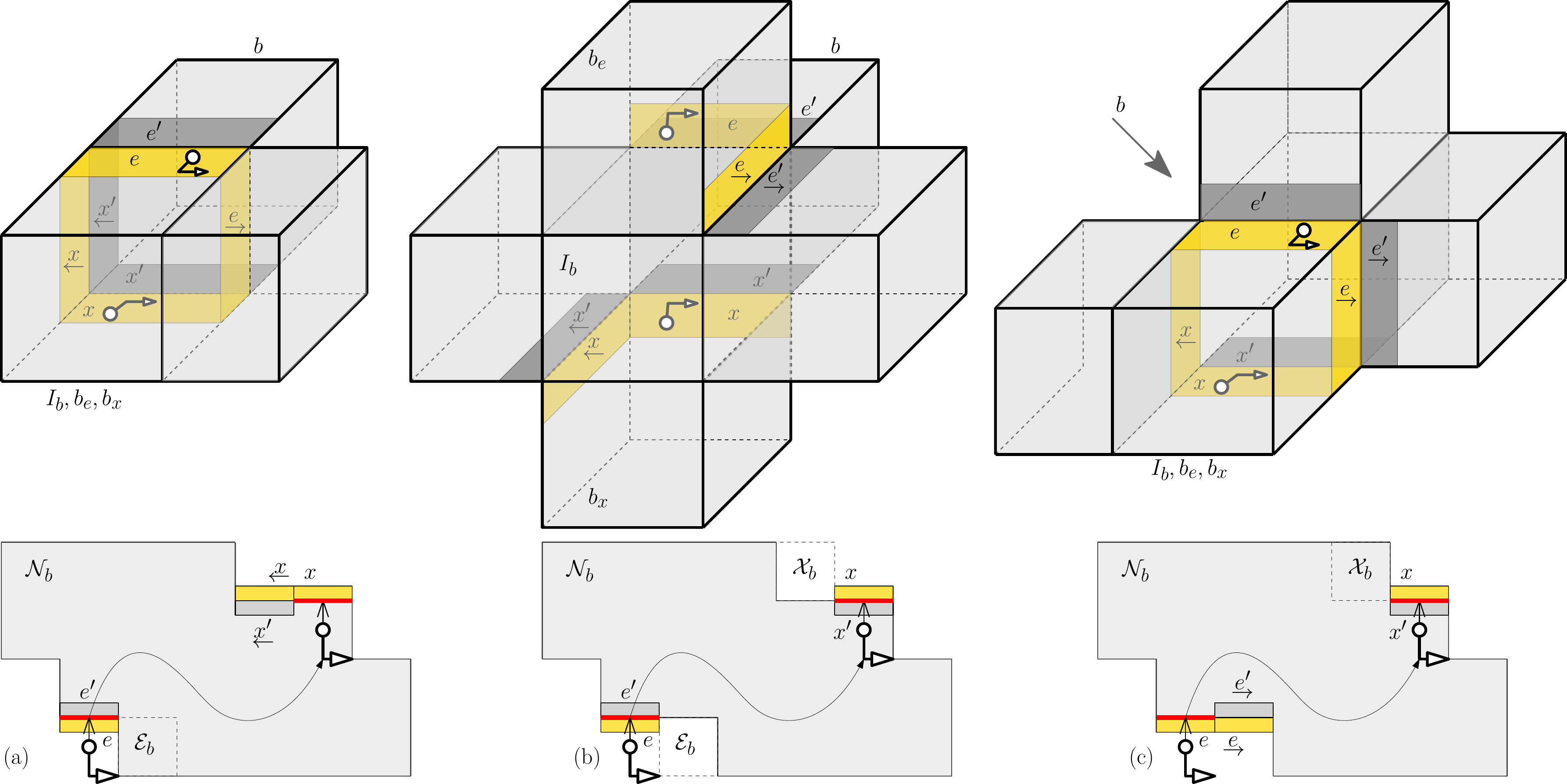}
\caption{
Net connections.
(a) Type-1 entry connection, because $\xrightarrow{e}$ is closed (so $\XE_b$ is part of the inductive region); type-2 exit connection, because $\xleftarrow{x}$ is open and adjacent to $\T_b$ (type-1 exit connection would also be allowed here) 
(b) Type-1 entry and exit connections, because $\xrightarrow{e}$ and $\xleftarrow{x}$ are non-adjacent to $\T_b$; they are both open, so $\XE_b$ and $\XX_b$ are not part of the inductive region
(c) Type-2 entry connection, because $\xrightarrow{e}$ is open and adjacent to $\T_b$ (type-1 entry connection would also be allowed here);  type-1 exit connection, because $\xleftarrow{x}$ is closed (so $\XX_b$ is part of the inductive region). Note that the strips $e$, $x$, $\xrightarrow{e}$ and $\xleftarrow{x}$ highlighted along the nets do not necessarily attach to $\N_b$; they are included here for the purpose of illustrating the definitions.
}
\label{fig:netconnections}
\end{figure*}
%

Let $e'$ ($x'$) be the open ring face of $\T_b$ that is adjacent to $e$ ($x$) along the 
entry (exit) port.
If $\xrightarrow{e}$ ($\xleftarrow{x}$) is open, let
$\xrightarrow{e'}$ ($\xleftarrow{x'}$) be the open ring face 
adjacent to it along its side of unit length (see \autoref{fig:netconnections}). Note that, although 
$e$ and $\xrightarrow{e}$ are ring faces from the same box by definition, ring faces
$e'$ and $\xrightarrow{e'}$ may be from different boxes (as in \autoref{fig:netconnections}b,c), and similarly for 
$x'$ and $\xleftarrow{x'}$. 
Although these definitions may seem a bit intricate at this point, they will greatly simplify the description of our approach. 

If $b$ is not the root of $\T$, to ensure that $b$'s net connects to the rest of $\T$'s unfolding, it must provide type-1 or type-2 connection pieces placed along the boundary inside its inductive region.  These connections are defined as follows:
\begin{itemize}
\squeezelist
\item A \emph{type-$1$ entry connection} consists of the ring face $e'$ placed alongside the entry port. (See~\autoref{fig:netconnections}(a,b) for examples.)  
\item A \emph{type-$1$ exit connection} consists of the ring face $x'$ placed alongside the exit port. (See~\autoref{fig:netconnections}(b,c) for examples.)
\item A \emph{type-$2$ entry connection} is used when the ring face $\xrightarrow{e}$ is open and adjacent to $T_b$, and  
consists of the ring face $\xrightarrow{e'}$ placed alongside the entry port extension. 
(See~\autoref{fig:netconnections}c for an example.)
\item A \emph{type-$2$ exit connection} is used when the ring face $\xleftarrow{x}$ is open and adjacent to $\T_b$, and consists of the ring face $\xleftarrow{x'}$ placed alongside the exit port extension. 
(See~\autoref{fig:netconnections}a for an example.)
\end{itemize}
The unfolding of $b$ begins (ends) on the type-$1$ or type-$2$ entry (exit) connection of $b$'s net. As we
will show, the existence of these connections is enough to guarantee that $b$'s net connects to the rest of $\T$'s unfolding.  In most cases, the connection to the rest of $\T$'s unfolding will be made along the port or port extension side of $b$'s type-$1$
or type-$2$ connection.  In some 
cases though, the connection will be made along the left (right) side of $b$'s type-$1$ entry (exit) connection.

\section{Unfolding Invariants} 
\label{sec:invariants}
We will make use of the following invariants tied to a recursive unfolding of a box $b \in \T$ other than the root box: 
\begin{enumerate}
\item[]
\begin{itemize}
\item [(I1)] 
The recursive 
unfolding of $b$ produces an unfolding net $\N_b$ that fits within the 
inductive region and includes all open faces of $\T_b$, with cuts restricted to a 
$4 \times 4$ refinement of the box faces.

\item [(I2)]
The unfolding net $\N_b$ provides the following entry and exit connections (see~\autoref{fig:netconnections}):
\begin{itemize}
\squeezelist
\item [(a)] 
If $\xrightarrow{e}$ is open and adjacent to a face in $\T_b$, then $\N_b$ provides either a type-1 or type-2 entry connection. Otherwise, $\N_b$ provides a type-1 entry connection. 
\item[(b)] 
If $\xleftarrow{x}$ is open and adjacent to a face in $\T_b$, then $\N_b$ provides either a type-1 or type-2 exit connection. Otherwise, $\N_b$ provides a type-1 exit connection. 
\end{itemize} 

\item [(I3)] 
Open faces of $b$'s ring that are not used in $N_b$'s entry and exit connections can be removed from $\N_b$ without disconnecting $\N_b$. 
\end{itemize}
\end{enumerate}
Invariant (I3) is sometimes employed in gluing two nets together, particularly in cases where the exit port of a box $b$ does not align with the entry port of the box $b'$ next visited by the unfolding path. In such cases, the unfolding algorithm may use ring pieces of $b$ and $b'$ identified by (I3) to form a bridge between their corresponding nets. Referring forward to~\autoref{fig:NEWSdegree5}~for example, the strips $R_N$ and $T_E$ are removed from their respective nets $\N_N$ and $\N_E$ and used to connect the two nets. Similarly, the strips $B_W$ and $L_S$ are removed from $\N_W$ and $\N_S$ and used to connect $\N_W$ and $\N_S$.

The following proposition follows immediately from the definition of the invariants (I1)-(I3) above. 
\begin{proposition}
\label{prop:ih}
If a net $\N_b$ satisfies invariants (I1)-(I3), then the net obtained after a $180^\circ$-rotation of $\N_b$ also satisfies invariants (I1)-(I3) (with entry and exit switching roles). 
\end{proposition}

\begin{lemma}
Let $\xi$ be the unfolding path and $\N_b$ the unfolding net produced by a recursive unfolding of $b$. Let $\overleftarrow{\xi}$ be the unfolding path traversed in reverse, starting at the exit port of $\N_b$ and ending at the entry port of $N$, with the \head\ and \hand\ pointing in opposite direction. If $\N_b$ satisfies the invariants (I1)-(I3), then the unfolding net induced by $\overleftarrow{\xi}$ also satisfies invariants (I1)-(I3). 
\label{lem:ih-symmetry}
\end{lemma}
\begin{proof}
The unfolding net $\overleftarrow{\N_b}$ induced by $\overleftarrow{\xi}$ is a diagonal flip ($180^\circ$-rotation) of $\N_b$. This along with~\autoref{prop:ih} implies that $\overleftarrow{\N_b}$ satisfies invariants (I1)-(I3). 
\end{proof}

\section{Main Result}
This section introduces our main result, which uses~\autoref{thm:main1} below. We note here that~\autoref{thm:main1} makes references to upcoming lemmas, which are organized into separate sections for clarity and ease of reference. So the main role of~\autoref{thm:main1} is to organize all unfolding cases into a structure that outlines the proof technique detailed in subsequent sections.

\begin{theorem} 
\label{thm:main1}
Any box $A \in \T$ other than the root satisfies invariants (I1)-(I3) (listed in Section~\ref{sec:invariants}).
\end{theorem}
\begin{proof}
The proof is by strong induction on the height $h$ of $\T_A$. The base case corresponds to $h = 0$ (i.e, $A$ is a leaf).

Consider the unfolding of leaf box $A$ depicted in~\autoref{fig:degree1}a: starting at $A$'s entry port, the unfolding path simply moves \head-first until it reaches $A$'s exit port. 
We now show that, when laid flat in the plane, the open faces of $A$ form a net $\N_A$ that satisfies invariants (I1)-(I3).  First note that the net $\N_{A}$ from~\autoref{fig:degree1}a fits within the inductive region and includes all open faces of $A$, therefore invariant (I1) is satisfied. 
To check (I2), observe that $\N_A$ provides type-1 entry and exit connections, since
$e' \in T_A$ and $x' \in B_A$ are positioned alongside the entry
and exit ports. To check (I3), observe that the open ring faces of $A$ not
used in $A$'s entry or exit connections are the dark-shaded pieces from~\autoref{fig:degree1}a, 
and their removal does not disconnect $\N_A$.
Thus $\N_A$ also satisfies all three invariants.

The inductive hypothesis states that the theorem holds for any dual subtree of height $h$ or less. To prove the inductive step, we consider a dual subtree $\T_A$ of height $h+1$, and prove that the theorem holds for the root $A$ of $\T_A$.

First note that, because $A$ is not the root of $\T$, $A$ has a parent in $\T$. Also, since the height of $\T_{A}$ is at least 1, $A$ has at least one child in $\T_A$. By the inductive hypothesis, each child of $A$ satisfies invariants (I1)-(I3). 
We discuss five cases, depending on the degree of $A$.
\begin{enumerate}
\item $A$ is of degree 2: this case is settled by~\autoref{thm:degree2}. 
\item $A$ is of degree 3: this case is settled by~\autoref{thm:degree3}. 
\item $A$ is of degree 4: this case is settled by~\autoref{thm:degree4}. 
\item $A$ is of degree 5: this case is settled by~\autoref{thm:degree5}. 
\item $A$ is of degree 6: this case is settled by~\autoref{thm:degree6}. 
\end{enumerate}
Having exhausted all cases, we conclude the result of this theorem. 
\end{proof}

\begin{theorem} \emph{{\bf [Main result.]}} 
Any polycube tree $\OO$ can be unfolded into a net using a $4 \times 4$ refinement. 
\label{thm:main}
\end{theorem}
\begin{proof}
%
\begin{figure}[ht]
\centering
\includegraphics[width=0.6\linewidth]{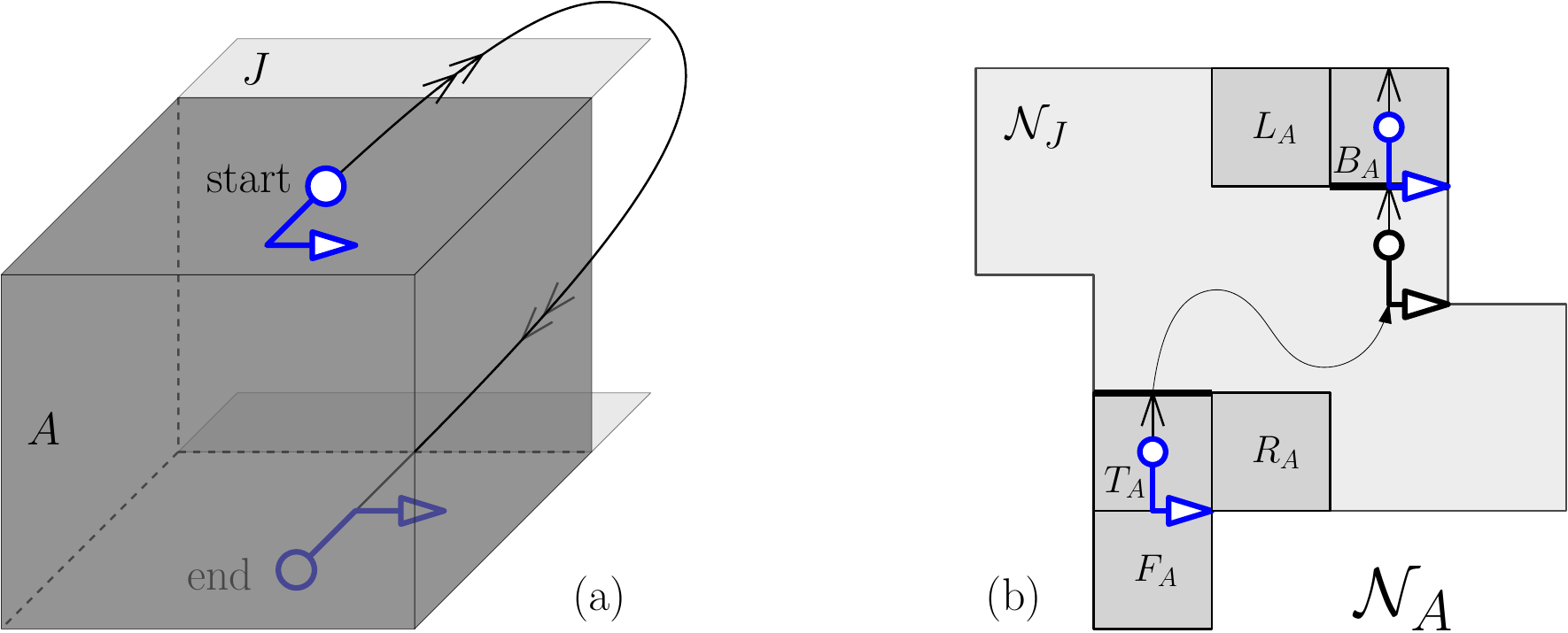}
\caption{Unfolding of root $A$ with back child $J$ (a) unfolding path (b) unfolding net $\N_A$.}
\label{fig:root}
\end{figure}
%
Let $\T$ be the dual tree of $\OO$ and let $A \in \T$ be the root of $\T$ (by definition, $A$ is a node of degree one in $\T$). Assume 
without loss of generality that $A$ has a back child $J$ (if this is not the case, reorient $\OO$ to make this assumption hold).
A recursive unfolding of $A$ is depicted in~\autoref{fig:root}a: starting \head-first on the top face of $A$, the unfolding path recursively visits $J$ and returns to the bottom face of $A$. The resulting net takes the shape depicted in~\autoref{fig:root}b. 

By~\autoref{thm:main1}, $J$ satisfies invariants (I1)-(I3), so its net $\N_J$ takes the shape depicted in~\autoref{fig:root}b. 
Notice that $e_J \in T_A$ and $x_J \in B_A$. Since $\xrightarrow{e_J} \in R_A$ and $\xleftarrow{x_J} \in L_A$ are both open,
the unit squares $\XE_J$ and $\XX_J$ (occupied in~\autoref{fig:root}b by $R_A$ and $L_A$, respectively) do not belong to the inductive region for $J$. 
Furthermore, since $\xrightarrow{e_J}$ and $\xleftarrow{x_J}$ are adjacent to $\T_J$,
invariant (I2) applied to $J$ tells us that $\N_J$ provides either type-1 or type-2 entry and exit connections. 
If of type-1, the entry (exit) connection attaches to $T_A$ ($B_A$); otherwise, it attaches to $R_A$ ($L_A$). In either case, the surface piece $\N_A$ depicted in~\autoref{fig:root}b is connected. Invariant (I1) applied to $J$ tells us that $\N_J$ is a net that includes all open faces in the subtree $\T_J$ rooted at $J$ and uses a $4 \times 4$ refinement. This along with the fact that the open faces of $A$ attach to $\N_J$ without overlap settles this theorem. 
\end{proof}

\medskip
\noindent
The need for a $4 \times 4$ refinement will become clear later in Section~\ref{sec:degree3}, where we discuss a case that requires a $4$-refinement along one dimension (depicted in~\autoref{fig:NJdegree3}a).

\section{Unfolding Algorithm}
\label{sec:alg}
Our unfolding algorithm uses an unfolding path that begins on the top face of the root box of $\T$, recursively visits all nodes in the subtree rooted at the (unique) child of the root box, and ends on the bottom face of the root box (as depicted in~\autoref{fig:root}). The result is a net that includes all open faces of $\OO$ (as established by ~\autoref{thm:main}).

\medskip
This section is dedicated to proving the five results referenced by~\autoref{thm:main1}. The unfolding algorithm is implicit in the proofs of these results. Here we provide a complete discussion for boxes of degrees 1, 2 and 6.  For boxes of degree 3, 4, and 5, we select only a few representative cases that exemplify our main ideas. The reader can refer to the appendix 
for the remaining cases, which are very similar. (We do not include all cases here in order to avoid repetitiveness and improve the flow and clarity of our techniques.) 


\subsection{Unfolding Degree-2 Nodes}
\label{sec:degree2}
In this section we describe the recursive unfolding of a box $A \in \T$ of degree $2$, and show that it satisfies the invariants (I1)-(I3) listed in Section~\ref{sec:invariants}. 

\begin{theorem}
\label{thm:degree2}
Let $A \in \T$ be a degree-2 box. If $A$'s child satisfies invariants (I1)-(I3), then 
$A$ satisfies invariants (I1)-(I3).
\end{theorem}
\begin{proof}
Our analysis is split into four different cases, depending on the position of $A$'s child (note that $A$'s parent contributes one unit to $A$'s degree):
\begin{itemize}
\item[]
\begin{itemize}
\item[]
\begin{enumerate}
\squeezelist
\item[Case 2.1] $E$ is a child of $A$. This case is settled by~\autoref{lem:Edegree2}. 

\item[Case 2.2] $W$ is a child of $A$. This case is settled by~\autoref{lem:Wdegree2}.
\item[Case 2.3] $J$ is a child of $A$. This case is settled by~\autoref{lem:Jdegree2}.
\item[Case 2.4]  $N$ is a child of $A$. This case is settled by~\autoref{lem:Ndegree2}. 
\end{enumerate}  
\end{itemize}
\end{itemize}
The case where $S$ is a child of $A$ is a vertical reflection of Case 2.4. 
\end{proof}

%
\begin{figure}[ht]
\centering
\includegraphics[page=1,width=\linewidth]{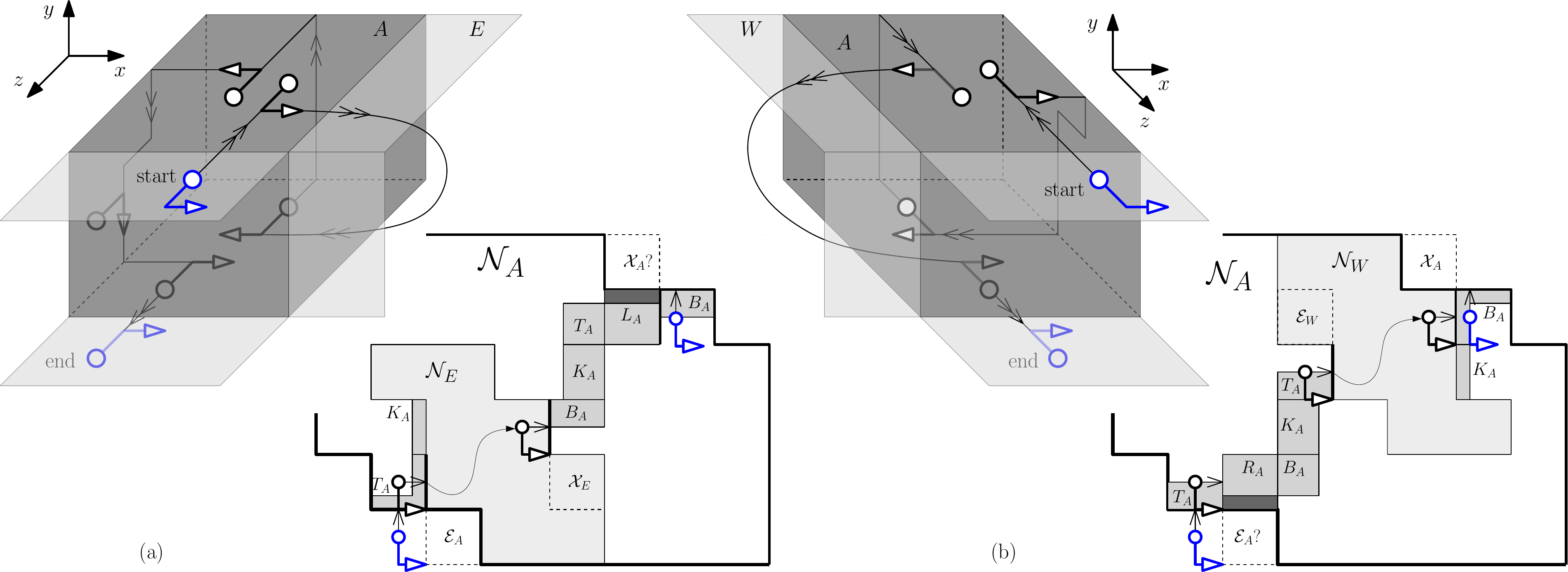}
\caption{Unfolding of degree-2 box $A$ with parent $I$ and child (a) $E$ (b) $W$.}
\label{fig:EWdegree2}
\end{figure}
%
\begin{lemma}
Let $A \in \T$ be a degree-2 node with parent $I$ and child $E$ \emph{(Case 2.1)}. 
If $E$ satisfies invariants (I1)-(I3), then $A$ satisfies invariants (I1)-(I3).
\label{lem:Edegree2}
\end{lemma}
\begin{proof}
\autoref{lem:hand} allows us to restrict our attention to \head-first unfoldings of $A$. The unfolding for this case is depicted in~\autoref{fig:EWdegree2}a: starting at $A$'s entry port, the unfolding path moves \head-first
to $T_A$, then proceeds \hand-first to recursively unfold $E$; 
from $E$'s exit ring face  on $B_A$, it proceeds \head-first up $K_A$ to
$T_A$; from $T_A$, it proceeds \hand-first down $L_A$ to $B_A$, and then
moves \head-first on $B_A$ to $A$'s exit port.  
We now show that, when visited in this order and laid flat in the plane, 
the open faces in $\T_A$ form a net $\N_A$ that satisfies invariants (I1)-(I3).

First note that the net $\N_A$ in~\autoref{fig:EWdegree2}a provides type-1 entry and exit connections, 
since $e'_A \in T_A$ and $x'_A \in B_A$ are positioned alongside its entry and exit ports. This shows that $\N_A$ satisfies invariant (I2).
Also note that (I3) is satisfied, because the only open ring face of $A$ not used in $\N_A$'s entry or exit connections is the piece of $L_A$ dark-shaded in \autoref{fig:EWdegree2}a (located below $\N_A$'s exit port extension), which can be removed from $\N_A$ without disconnecting $\N_A$.

It remains to show that $\N_A$ satisfies invariant (I1). 
We begin with the following set of observations showing that 
the net $\N_E$ produced by the recursive unfolding of $E$ connects to the pieces
of $T_A$, $K_A$, and $B_A$ placed alongside its boundary:
\begin{itemize}
\item Observe first that the entry (exit) port in the recursive unfolding of $E$  is the top (bottom) edge of $R_A$. With this entry (exit) port, $E$'s entry (exit) ring face $e_E$ ($x_E$) is on $T_A$ ($B_A$) and its successor (predecessor) $\xrightarrow{e_E}$ ($\xleftarrow{x_E}$) is on $K_A$ ($F_A$).
\item Since $\xrightarrow{e_E} \in K_A$ is open, the unit square $\XE_E$ (occupied by $\xrightarrow{e_E}$ in~\autoref{fig:EWdegree2}a) is not part of the inductive region for $E$.
Since $\xrightarrow{e_E}$ is also adjacent to $\T_E$, invariant (I2) applied to $E$ tells us 
that $\N_E$ provides either a type-1 or type-2 entry connection.  If $\N_E$ provides a type-1 entry connection, then $e_E'$ is located alongside its entry port, and it connects (by definition) to $e_E \in T_A$  located on the other side of its entry port  (see \autoref{fig:EWdegree2}a); if $\N_E$ provides a type-2 connection, then $\xrightarrow{e'_E}$ is located alongside its entry port extension, and it connects (by definition) to $\xrightarrow{e_E} \in K_A$ located  on the other side of its entry port extension. 
\item Since $\xleftarrow{x_E} \in F_A$ is closed, $\XX_E$ is part of $E$'s inductive region and
the invariant (I2) applied to $E$ tells us 
that $\N_E$  provides a type-1 exit connection.  This means that $x_E'$ is located alongside $\N_E$'s
exit port, and it connects (by definition) to the piece
of $x_E \in B_A$ located on the other side of its exit port (see \autoref{fig:EWdegree2}a).
\end{itemize}

Because invariant (I1) tells us that $\N_E$ is connected and because the pieces of $A$ placed alongside $\N_E$ connect to $\N_E$'s entry and exit connections, we can conclude that $\N_A$ is connected.  By invariant (I1) applied to $E$, the net $\N_E$
includes all open faces in $\T_E$ using a $4 \times 4$ refinement. This along with the fact that $\N_A$ includes $T_A, L_A, B_A,$ and $K_A$ (which are $A$'s open faces) using a $4 \times 4$ refinement shows that $\N_A$ includes all open faces of $T_A$ using a $4\times4$ refinement.   Finally, $\N_A$ fits within $A$'s inductive region as illustrated in \autoref{fig:EWdegree2}a, noting that no part of $\N_A$ lies within the cells marked $\XE_A$ and $\XX_A$ (which renders a discussion of whether or not these cells are part of its inductive region unnecessary). 
Thus we can conclude that $\N_A$ satisfies invariant (I1).
\end{proof}

%
%

\begin{lemma}
Let $A \in \T$ be a degree-2 node with parent $I$ and child $W$ \emph{(Case 2.2)}.
If $W$ satisfies invariants (I1)-(I3), then $A$ satisfies invariants (I1)-(I3).
\label{lem:Wdegree2}
\end{lemma}
\begin{proof}
The unfolding for this case is depicted in~\autoref{fig:EWdegree2}b.  
Note that this unfolding path can be obtained by rotating the path from~\autoref{fig:EWdegree2}a by $180^\circ$.
This along with Lemmas 4 and 5 
implies that the net $\N_A$ from~\autoref{fig:EWdegree2}b satisfies invariants (I1)-(I3).  
%
\end{proof}

%
\begin{figure}[ht]
\centering
\includegraphics[width=0.9\linewidth]{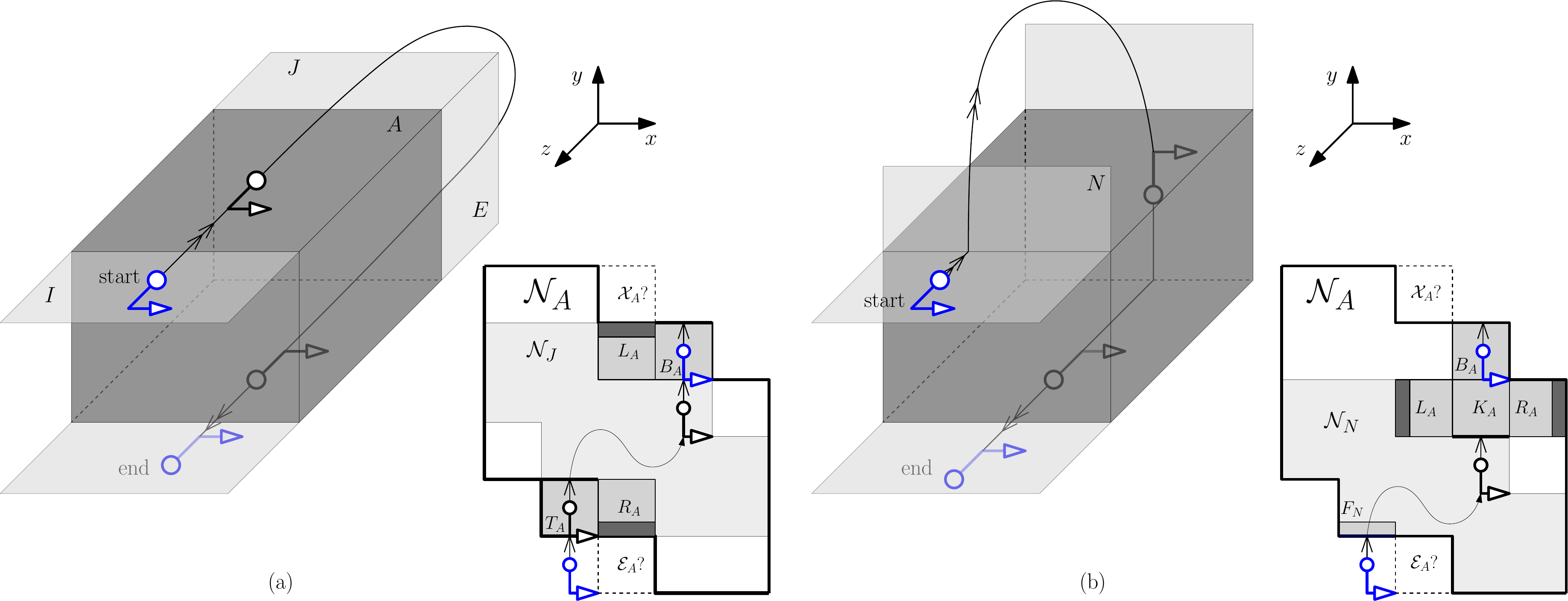}
\caption{Unfolding of degree-2 box $A$ with parent $I$ and child (a) $J$ (b) $N$. 
}
\label{fig:JNdegree2}
\end{figure}
%

\begin{lemma}
Let $A \in \T$ be a degree-2 node with parent $I$ and child $J$ \emph{(Case 2.3)}.
If $J$ satisfies invariants (I1)-(I3), then $A$ satisfies invariants (I1)-(I3).
\label{lem:Jdegree2}
\end{lemma}
\begin{proof}
Consider the unfolding depicted in~\autoref{fig:JNdegree2}a, and notice its similarity with the unfolding of the root box from~\autoref{fig:root}. We show that the unfolding $\N_A$ from~\autoref{fig:JNdegree2}a satisfies invariants (I1)-(I3). 

Note that $\N_A$ provides a type-1 entry connection ($e'_A \in T_A$) and a type-1 exit connection ($x'_A \in B_A$), and therefore it satisfies invariant (I2). Since $\xrightarrow{e_J} \in R_A$ ($\xleftarrow{x_J} \in L_A$) is open,
$\XE_J$ ($\XX_J$) is not part of $J$'s inductive region. Furthermore, since $\xrightarrow{e_J}$ ($\xleftarrow{x_J}$)
is adjacent to $\T_J$, invariant (I2) applied to $J$ tells us that $\N_J$ provides a type-1 or type-2 entry (exit) connection, which attaches to $T_A$ or $R_A$ ($B_A$ or $L_A$). Thus the net $\N_A$ is connected. 

By invariant (I1), $\N_J$ covers all open faces in $\T_J$ using a 
$4 \times 4$ refinement. Since $\N_A$ includes the open faces of $A$ without any refinement, 
we conclude that $\N_A$ includes all open faces of $\T_A$. Noting that $\N_A$ fits within $A$'s inductive region
(and doesn't use the cells marked $\XE_A$ and $\XX_A$), we conclude that $\N_A$ satisfies invariant (I1).
Finally, the open ring faces of $A$ not used in its entry and exit connections (dark-shaded in~\autoref{fig:JNdegree2}a) can be removed from $\N_A$ without disconnecting $\N_A$, therefore $\N_A$ satisfies invariant (I3).
\end{proof}

\begin{lemma}
Let $A \in \T$ be a degree-2 node with parent $I$ and child $N$ \emph{(Case 2.4)}.
If $N$ satisfies invariants (I1)-(I3), then $A$ satisfies invariants (I1)-(I3).
\label{lem:Ndegree2}
\end{lemma}
\begin{proof}
Consider the unfolding depicted in~\autoref{fig:JNdegree2}b. Note that $\xrightarrow{e_A} = \xrightarrow{e_N} \in R_I$ is not adjacent to $\T_N$, therefore $\N_N$ will provide a type-1 entry connection (by (I2) applied to $N$), which is also a type-1 entry connection for $A$ (because $e'_N = e'_A \in F_N$).  Note that $\N_A$ also provides a type-1 exit connection $x'_A \in B_A$, therefore $\N_A$ satisfies invariant (I2). Since $\xleftarrow{x_N} \in L_A$ is open, the unit square $\XX_N$ (occupied by $L_A$ in~\autoref{fig:JNdegree2}b) does not belong to the inductive region for $N$. Furthermore, since $\xleftarrow{x_N}$ is adjacent to $\T_N$, invariant (I2) applied to $N$ tells us that $\N_N$ provides a type-1 or type-2 exit connection, 
which attaches to $K_A$ or $L_A$ (located along the exit port and exit port extension). 
Thus the net $\N_A$ is connected.  
Arguments similar to those in \autoref{lem:Jdegree2} complete the proof that $\N_A$ 
satisfies (I1) and show that it satisfies invariant (I3).
%
%
\end{proof}


\subsection{Unfolding Degree-3 Nodes}
\label{sec:degree3}
In this section we describe the recursive unfolding of a box $A \in \T$ of degree $3$, and show that it satisfies 
the invariants (I1)-(I3) listed in Section~\ref{sec:invariants}. 

\begin{theorem}
\label{thm:degree3}
Let $A \in \T$ be a degree-3 box. If $A$'s children satisfy invariants (I1)-(I3), then 
$A$ satisfies invariants (I1)-(I3).
\end{theorem}
\begin{proof}
Our analysis is split into five different cases, depending on the position of $A$'s children:
\begin{itemize}
\item[]
\begin{itemize}
\item[]
\begin{enumerate}
\squeezelist
\item[Case 3.1] $E$ and $J$ are children of $A$. The case where $W$ and $J$ are children of $A$ is a horizontal reflection of this case, with the unfolding path traversed in reverse. 
%
\item[Case 3.2] $N$ and $J$ are children of $A$. The case where $S$ and $J$ are children of $A$ is a vertical reflection of this case, with the unfolding path traversed in the reverse. 

\item[Case 3.3] $E$ and $W$ are children of $A$.
\item[Case 3.4]  $N$ and $S$ are children of $A$.
\item[Case 3.5] $N$ and $E$ are children of $A$. This is the same as the case where $S$ and $W$ are children of $A$, rotated by $180^\circ$ about the z-axis (so the unfolding path is the same, but traversed in reverse). 

\item[Case 3.6] $N$ and $W$ are children of $A$. This is the same as the case where $S$ and $E$ are children of $A$, rotated by $180^\circ$ about the z-axis (so the unfolding path is the same, but traversed in reverse). 
\end{enumerate}  
\end{itemize}
\end{itemize}
The rest of this section is devoted to a detailed analysis of Cases 3.1 and 3.2. Case 3.2 in particular is special because it 
requires $4 \times 4$ refinement. 
Cases 3.3 through 3.5, while employing different unfolding paths, use similar arguments in their correctness proofs and are detailed in Appendix~\ref{sec:degree3-appendix}. 
Note that the ability to "traverse in reverse" in some of the cases listed above follows from~\autoref{lem:ih-symmetry}.
\end{proof}


%
\begin{figure}[ht]
\centering
\includegraphics[width=\linewidth]{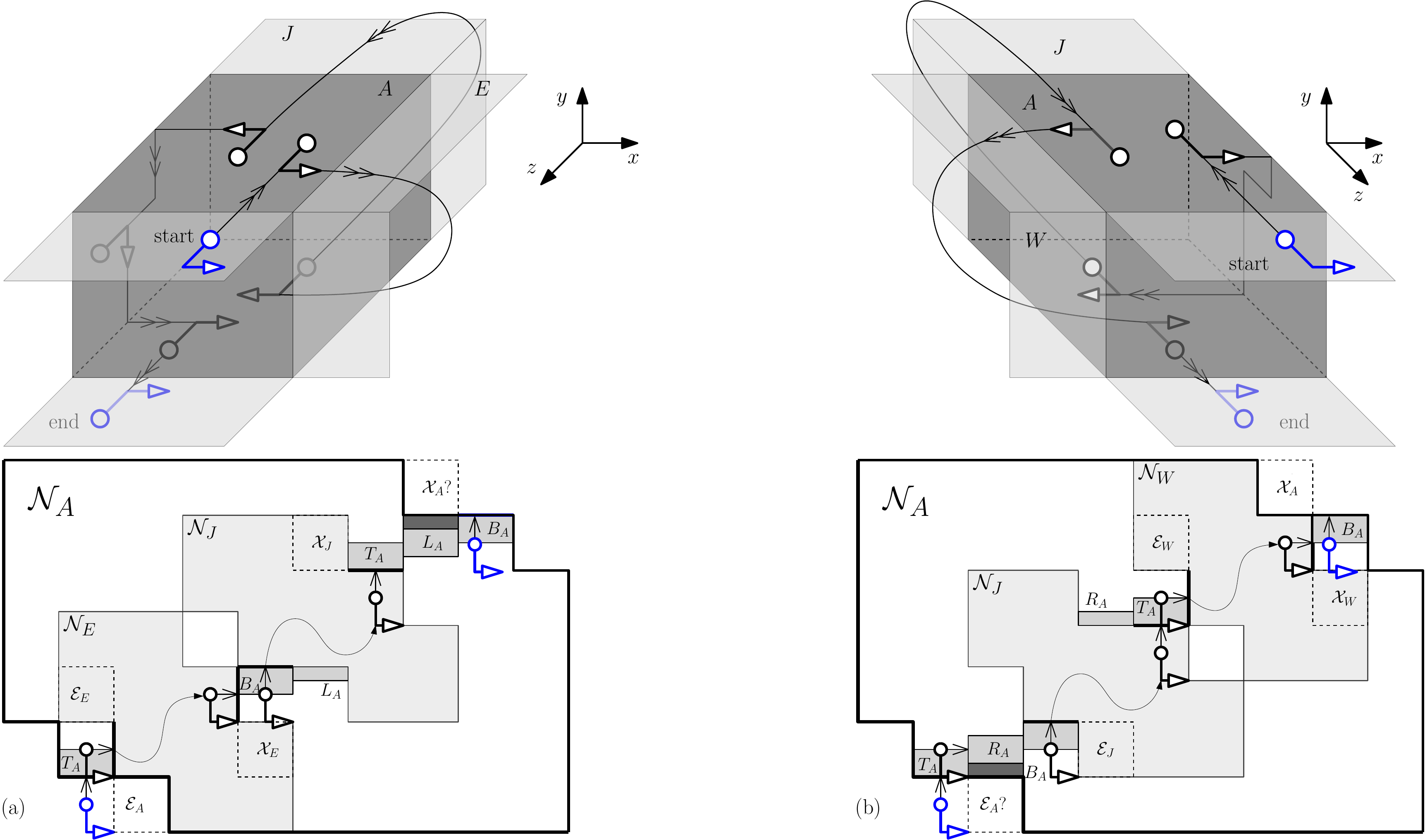}
\caption{Unfolding of degree-3 box $A$ with children $J$ and (a) $E$ (b) $W$.}
\label{fig:EJ-WJdegree3}
\end{figure}
%

\begin{lemma}
Let $A \in \T$ be a degree-3 node with parent $I$ and children $E$ and $J$ \emph{(Case 3.1)}.
If $A$'s children satisfy invariants (I1)-(I3), then $A$ satisfies invariants (I1)-(I3).
\label{lem:EJdegree3}
\end{lemma}
\begin{proof}
The unfolding for this case is depicted in~\autoref{fig:EJ-WJdegree3}a. Observe that it is a generalization of the
degree-2 unfolding from~\autoref{fig:EWdegree2}a, where the unfolded face $K_A$ is replaced by the recursive unfolding of child $J$. 
Since the two unfoldings and the proofs of their correctness are very similar, we only point out the differences here:
\begin{itemize}
%
\item Because the ring face $\xrightarrow{e_E} \in K_A$ is closed, $\XE_E$ is 
part of $E$'s inductive region. By invariant (I2) applied to $E$, $\N_E$ provides a type-1 entry connection, which connects to $e_E \in T_A$.
\item Observe that the entry (exit) port for $J$ is the bottom (top) edge of $F_J$
and so the entry (exit) ring face $e_J$ ($x_J$) is part of $B_A$ ($T_A$). Because $\xrightarrow{e_J} \in L_A$ is open, the unit square $\XE_J$ (occupied by $\xrightarrow{e_J}$ in~\autoref{fig:EJ-WJdegree3}a) is
not part of $J$'s inductive region. Furthermore, since $\xrightarrow{e_J}$ is adjacent to $\T_J$, invariant (I2) applied to $J$ tells us that $\N_J$ provides a type-1 or type-2 entry connection: if type-1, then it  connects  to 
the piece $e_J \in B_A$; if type-2, then it connects to $\xrightarrow{e_J} \in L_A$. 
\item Because the ring face $\xleftarrow{x_J} \in R_A$ is closed, $\XX_J$ is 
part of $E$'s inductive region. By invariant (I2) applied to $J$, $\N_J$ provides a type-1 exit connection, which connects to $x_J \in T_A$. 
%
\end{itemize}
These differences combined with arguments similar to those in \autoref{lem:Edegree2} show
that $\N_A$ satisfies invariants (I1)-(I3).
\end{proof}

The case where $W$ and $J$ are children of $A$ shown in~\autoref{fig:EJ-WJdegree3}b is the reverse
of the case shown in~\autoref{fig:EJ-WJdegree3}a. This along with~\autoref{lem:ih-symmetry} implies that the net
$\N_A$ from~\autoref{fig:EJ-WJdegree3}b also satisfies invariants (I1)-(I3).

%
\begin{figure*}[htbp] 
\centering
\includegraphics[width=\linewidth]{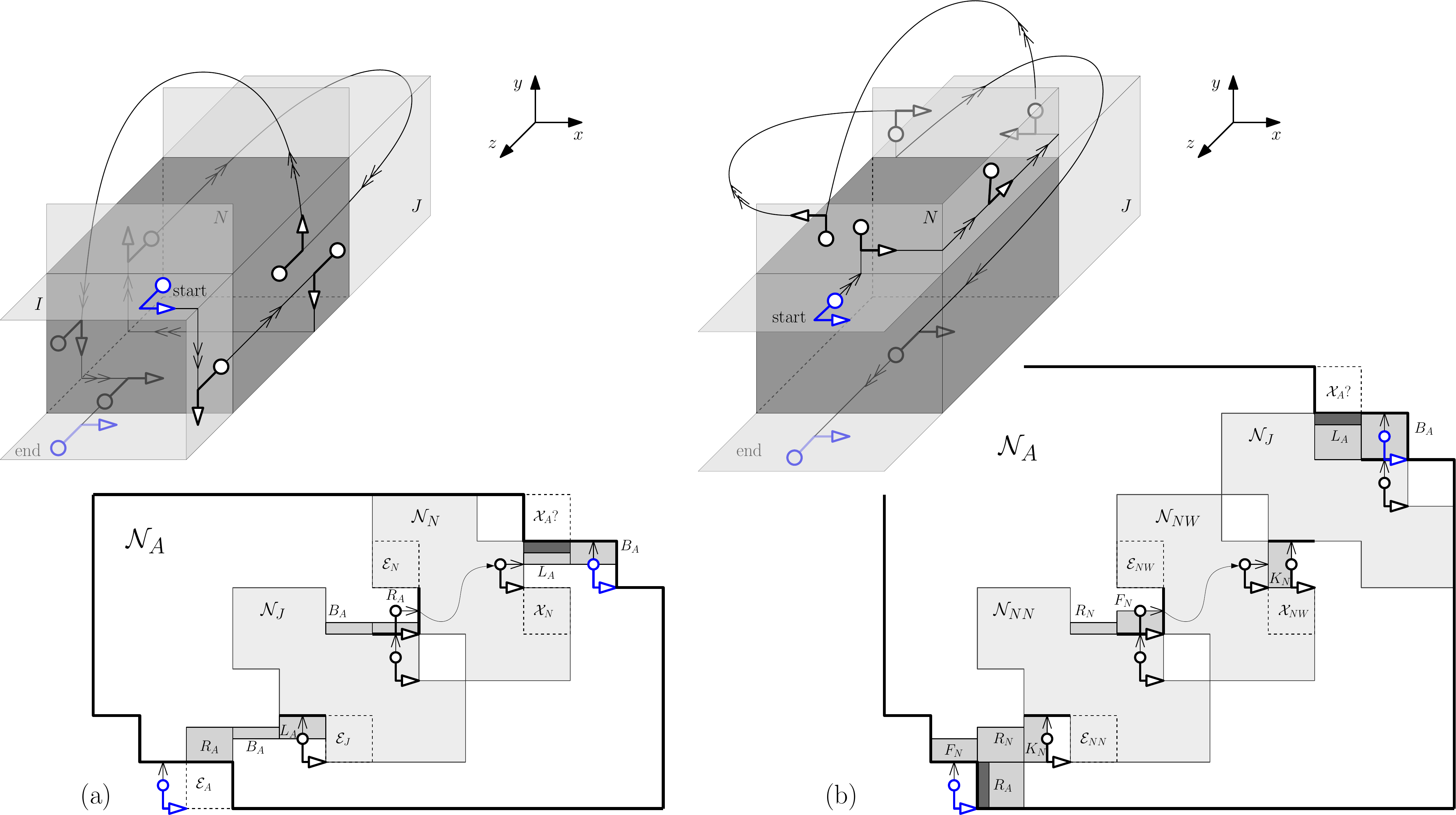}
\caption{Unfolding of degree-3 box $A$ with parent $I$ and children $N$ and $J$ (a) $R_I$ open. Note that $A$ requires a $4$-refinement along the $z$-dimension to be able to generate the strips of $B_A$ and $L_A$ shown here.} (b) $R_I$ closed (so $R_N$ open). 
\label{fig:NJdegree3}
\end{figure*}
%
\begin{lemma}
\label{lem:NJdegree3}
Let $A \in \T$ be a degree-3 node with parent $I$ and children $N$ and $J$ \emph{(Case 3.2)}.
If $A$'s children satisfy invariants (I1)-(I3), then $A$ satisfies invariants (I1)-(I3).
\end{lemma}
\begin{proof}
We discuss two situations, depending on whether $R_I$ is open or closed. Assume first that $R_I$ is open, and 
consider the unfolding depicted in~\autoref{fig:NJdegree3}a. Note that $\xrightarrow{e_A} \in R_I$ is open and adjacent 
to $\T_A$, and $\N_A$ provides a type-2 entry connection $\xrightarrow{e'_A} \in R_A$. Also note that 
$\N_A$ provides a type-1 exit connection $x'_A \in B_A$. These together show that $\N_A$ satisfies invariant (I2). The following observations support our claim that $\N_A$ satisfies invariant (I1):
\begin{itemize}
\item The entry and exit ring faces for $J$ are $e_J \in L_A$  and $x_J \in R_A$. Since $\xrightarrow{e_J} \in T_A$ is closed, $\N_J$ provides a type-1 entry connection, which attaches to $e_J \in L_A$. Since $\xleftarrow{x_J} \in B_A$ is open and adjacent to $\T_J$, 
the unit square $\XX_J$ (occupied by $\xleftarrow{x_J}$ in~\autoref{fig:NJdegree3}a) does not belong to the inductive region for $J$, and
$\N_J$ may provide a type-1 or type-2 exit connection: if type-1, it attaches to the ring face $x_J \in R_A$ placed alongside its exit port; if type-2, it connects to the ring face $\xrightarrow{x_J} \in B_A$ placed alongside its exit port extension. 
\item The entry and exit ring faces for $N$ are $e_N \in R_A$ and $x_N \in L_A$.
Note that $\N_N$ provides type-1 entry and exit connections (since $\xrightarrow{e_N} \in F_A$ and $\xleftarrow{x_N} \in K_A$ are both closed), which attach to the pieces of the entry and exit ring faces placed alongside its entry and exit ports.
\end{itemize}
Finally, note that the only open ring face of $A$ not involved in $A$'s entry and exit connections is the dark-shaded piece of $L_A$ from~\autoref{fig:NJdegree3}a, whose removal does not disconnect $\N_A$. Thus $\N_A$ satisfied (I3) as well.

\medskip
\noindent
Assume now that $R_I$ is closed, and consider the unfolding depicted in~\autoref{fig:NJdegree3}b. Note that $\N_A$ provides type-1 entry and exit connections $e'_A \in F_N$ and $x'_A \in B_A$, therefore it satisfies invariant (I2). 
The following observations support our claim that $\N_A$ satisfies invariant (I1):
\begin{itemize}
\item The entry and exit ring faces for $NN$, $NW$ and $J$ are as follows: $e_{NN} \in K_N$ and $x_{NN} \in F_N$; 
$e_{NW} \in F_N$ and $x_{NW} \in K_N$; 
and $e_{J} \in K_N$ and $x_{J} \in B_A$.
\item $\N_{NN}$, $\N_{NW}$ and $\N_J$ provide type-1 entry connections. This is because $\xrightarrow{e_{NN}} \in L_N$ is closed, $\xrightarrow{e_{NW}} \in B_N$ is closed, and $\xrightarrow{e_J} \in R_N$ is not adjacent to $\T_J$. 
\item Since $\xleftarrow{x_{NN}} \in R_N$ is open, the unit square $\XX_{NN}$ (occupied by $\xleftarrow{x_{NN}}$ in~\autoref{fig:NJdegree3}b) does not belong to the inductive region for $NN$. Similarly, since $\xleftarrow{x_J} \in L_A$ is open, 
the unit square $\XX_J$ (occupied by $L_A$ in~\autoref{fig:NJdegree3}b) does not belong to the inductive region for $J$. 
\item Since $\xleftarrow{x_{NW}} \in T_N$ is closed, $\N_{NW}$ provides a type-1 exit connection.
\item Since $\xrightarrow{e_A} \in R_I$ is closed, the unit square $\XE_A$ (occupied by $R_A$ in~\autoref{fig:NJdegree3}b) belongs
to the inductive region for $A$.
\end{itemize}
Finally, note that the removal of the open ring faces of $A$ not involved in $A$'s entry and exit connections (shown dark-shaded in~\autoref{fig:NJdegree3}b) does not disconnect $\N_A$. Thus $\N_A$ satisfies (I3) as well.
\end{proof}

\medskip
\noindent
As a side note, the unfolding from~\autoref{fig:NJdegree3}a is the first unfolding example that requires a $4$-refinement along one dimension of the grid: one $1/4 \times 1$ strip of $B_A$ is needed to transition from $R_A$ to $L_A$; one $1/4 \times 1$ strip of $B_A$ is needed alongside $\N_J$'s exit port extension, to connect to the type-2 connection that $\N_J$ may provide; and one $1/2 \times 1$ strip of $B_A$ is needed alongside $\N_A$'s exit port, so that it remains connected to the piece of $L_A$ to its left, once the dark-shaded ring face that lies on $L_A$ has been removed.


\subsection{Unfolding Degree-4 Nodes}
\label{sec:degree4}
In this section we describe the recursive unfolding of a box $A \in \T$ of degree $4$, and show that it 
the invariants (I1)-(I3) listed in Section~\ref{sec:invariants}. 

\begin{theorem}
\label{thm:degree4}
Let $A \in \T$ be a degree-4 box. If $A$'s children satisfy invariants (I1)-(I3), then 
$A$ satisfies invariants (I1)-(I3).
\end{theorem}
\begin{proof}
Our analysis is split into seven different cases, depending on the position of $A$'s children:
\begin{itemize}
\item[]
\begin{itemize}
\item[]
\begin{enumerate}
\squeezelist
\item[Case 4.1] $J$, $E$ and $W$ are children of $A$. 
\item[Case 4.2] $N$, $E$ and $W$ are children of $A$. The case where $S$, $E$ and $W$ are children of $A$ is a vertical reflection of this case. 
\item[Case 4.3] $N$, $E$ and $J$ are children of $A$. The case where $S$, $W$ and $J$ are children of $A$ is a $180^\circ$-rotation about the $z$-axis of this case. 

\item[Case 4.4] $N$, $W$ and $J$ are children of $A$. The case where $S$, $E$ and $J$ are children of $A$
is a $180^\circ$-rotation about the $z$-axis of this case.

\item[Case 4.5]  $N$, $E$ and $S$ are children of $A$. 

\item[Case 4.6]  $N$, $W$ and $S$ are children of $A$. 
\item[Case 4.7] $N$, $J$ and $S$ are children of $A$. 
\end{enumerate}  
\end{itemize}
\end{itemize}
It can be verified that this is an exhaustive list of all possible cases for a degree-$4$ node. 
Case 4.1 is settled by~\autoref{lem:EJWdegree4}. 
Cases 4.2 through 4.7, while employing different unfolding paths, use similar arguments in their correctness proofs and are detailed in Appendix~\ref{sec:degree4-appendix}. 
\end{proof}


\begin{lemma}
\label{lem:EJWdegree4}
Let $A \in \T$ be a degree-$4$ node with parent $I$ and children $J$, $E$ and $W$ \emph{(Case 4.1)}. 
If $A$'s children satisfy invariants (I1)-(I3), then $A$ satisfies invariants (I1)-(I3).
\end{lemma}
\begin{proof}
%
\begin{figure}[ht]
\centering
\includegraphics[page=1,width=.8\textwidth]{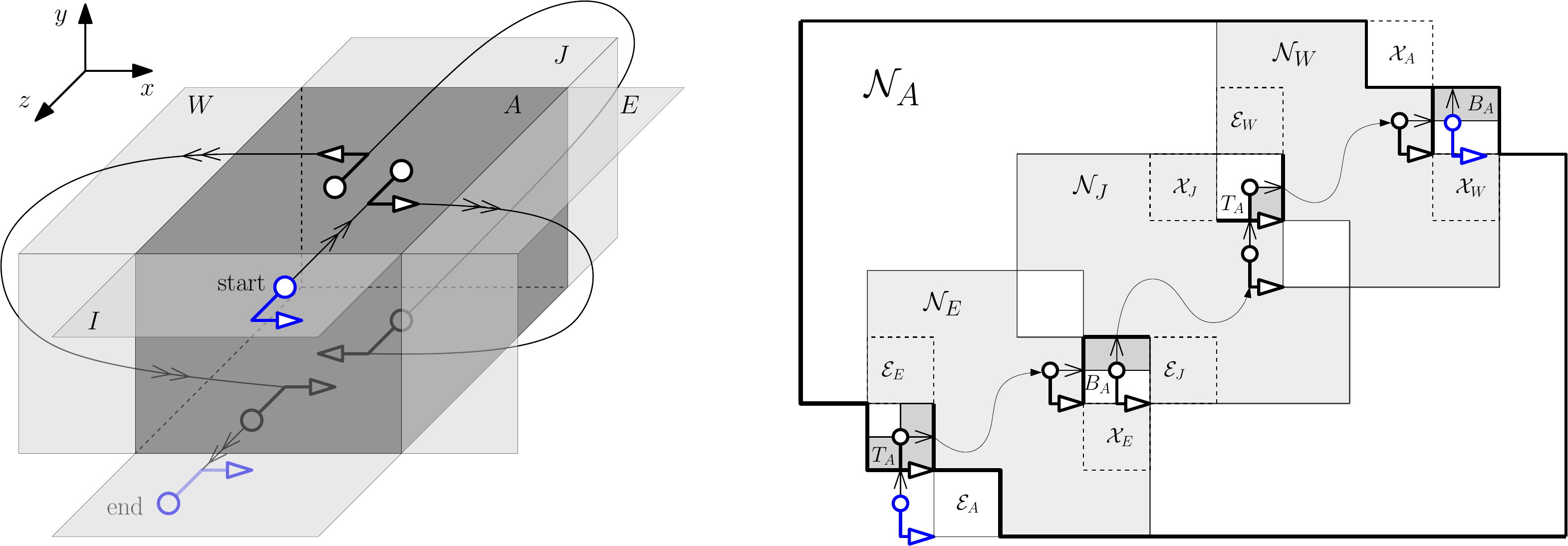}
\caption{Unfolding of degree-$4$ box $A$ with children $J$, $E$ and $W$.}
\label{fig:EJWdegree4}
\end{figure}
%
%
Consider the unfolding depicted in~\autoref{fig:EJWdegree4}, and note that 
it is a generalization of the
degree-3 unfolding from~\autoref{fig:EJ-WJdegree3}a, where the unfolded face $L_A$ is replaced by the recursive unfolding of child $W$. Since the two unfoldings and their proofs are very similar, we only point out the differences here: 
\begin{itemize}
\item Because $\xrightarrow{e_J} \in L_A$ is closed, $\XE_J$ is part of $J$'s 
inductive region and $\N_J$ provides a type-1 entry connection which connects to  
 $e_J \in B_A$. 
\item Observe that the entry (exit) port for $W$ is the top (bottom) edge of $R_W$
and so the entry (exit) ring face $e_W$ ($x_W$) is part of $T_A$ ($B_A$). Because $\xrightarrow{e_W} \in F_A$ ($\xleftarrow{x_W} \in K_A$) is closed, $\XE_W$
($\XX_W$) is part of $W$'s
inductive region, and  
 $\N_W$ provides a type-1 entry (exit) connection which connects to the 
piece of $e_W \in T_A$ ($x_W \in B_A$) placed along $\N_W$'s entry (exit) port. 
\end{itemize}
These differences combined with arguments similar to those in \autoref{lem:EJdegree3} show
that $\N_A$  satisfies invariants (I1) and (I2).
Finally note that (I3) is trivially satisfied, because all of $A$'s open ring faces  
are used in its entry and exit connections. We therefore conclude that $\N_A$ from \autoref{fig:EJWdegree4} satisfies invariants (I1)-(I3).
\end{proof}


\subsection{Unfolding Degree-5 Nodes}
\label{sec:degree5}
In this section we describe the recursive unfolding of a box $A \in \T$ of degree $5$, and show that it satisfies the invariants (I1)-(I3) listed in Section~\ref{sec:invariants}. 

\begin{theorem}
Let $A \in \T$ be a degree-5 box. If $A$'s children satisfy invariants (I1)-(I3), then 
$A$ satisfies invariants (I1)-(I3). 
\label{thm:degree5}
\end{theorem}
\begin{proof}
Our analysis is split into four different cases, depending on the position of $A$'s children:
\begin{itemize}
\item[]
\begin{itemize}
\item[]
\begin{enumerate}
\squeezelist
\item[Case 5.1] $J$ is not a child of $A$ (so $N$, $E$, $W$ and $S$ are children of $A$). 
\item[Case 5.2] $W$ is not a child of $A$ (so $N$, $E$, $J$ and $S$ are children of $A$). 

\item[Case 5.3] $E$ is not a child of $A$ (so $N$, $W$, $J$ and $S$ are children of $A$).
\item[Case 5.4]  $N$ is not a child of $A$ (so $E$, $W$, $J$ and $S$ are children of $A$). The case when $S$ is not a child of $A$ is a vertical reflection of this case.
\end{enumerate}  
\end{itemize}
\end{itemize}
It can be verified that this is an exhaustive list of all possible cases for a degree-$5$ node. Case 5.1 is settled by~\autoref{lem:degree5}. Cases 5.2 through 5.4, while employing different unfolding paths, use similar arguments in their correctness proofs and are detailed in Appendix~\ref{sec:degree5-appendix}. 
\end{proof}

\medskip

\medskip
Before getting into details on Case 5.1, we introduce a preliminary lemma that will simplify our analysis.

\begin{lemma}
\label{lem:degree5-connector}
Let $A \in \T$ be a degree-5 node with parent $I$ and children $N$, $E$, $W$ and $S$. Then either $N$ and $S$ are both non-junction boxes, or else $E$ and $W$ are both non-junction boxes. 
\end{lemma}
\begin{proof}
Assume to the contrary that at least one box in each pair ($N$, $S$) and ($E$, $W$) -- say, $N$ and $E$ -- is a junction (the argument for any choice of junctions is the same). This implies that $N$ has a back 
neighbor (because any other neighbor position that would render $N$ a junction would also render a loop in $\T$), 
and similarly for $E$. Note however that $NJ$ and $EJ$ meet at an edge, therefore $NJ$ must have either a south or an east neighbor (because $\OO$ is homeomorphic to a sphere). However, each of these cases renders a cycle in $\T$, a contradiction. 
\end{proof}

\begin{lemma}
\label{lem:degree5}
Let $A \in \T$ be a degree-$5$ node with parent $I$ and children $N$, $E$, $W$ and $S$ \emph{(Case $5.1$)}. 
If $A$'s children satisfy invariants (I1)-(I3), then $A$ satisfies invariants (I1)-(I3).
\end{lemma}
\begin{proof}
%
\begin{figure}[ht]
\centering
\includegraphics[page=1,width=\textwidth]{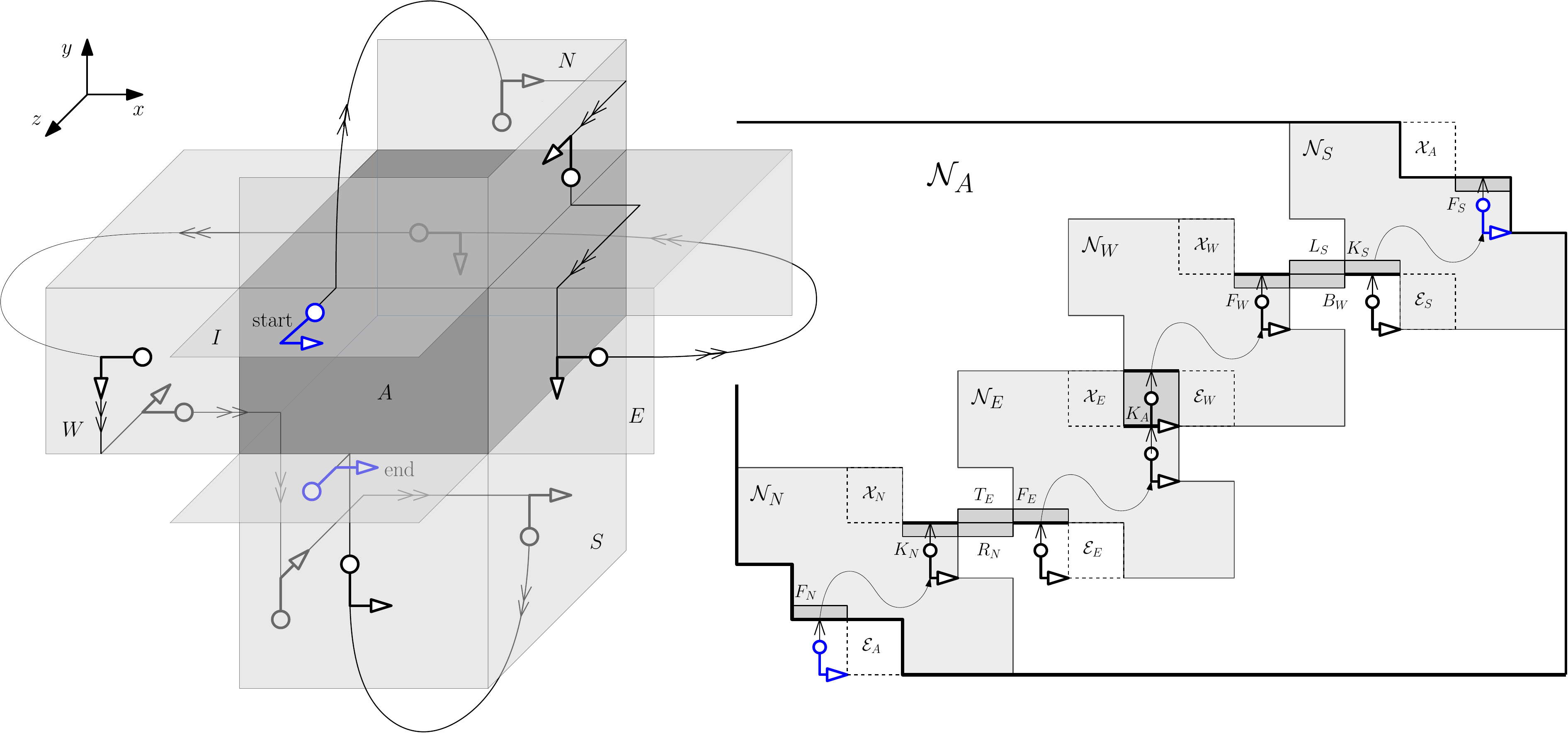}
\caption{Unfolding of degree-5 box $A$ with children $N$, $E$, $W$ and $S$ ($N$ and $S$ are non-junctions).}
\label{fig:NEWSdegree5}
\end{figure}
%
By~\autoref{lem:degree5-connector}, either $N$ and $S$ are both non-junctions, or $E$ and $W$ are both non-junctions.  
Assume first that $N$ and $S$ are both non-junctions and consider the unfolding depicted in~\autoref{fig:NEWSdegree5}:
starting at $A$'s entry port, the unfolding path proceeds \head-first to recursively unfold $N$; 
upon reaching $N$'s exit port on $K_N$, it moves \hand-first to $R_N$, \head-first to $T_E$, 
 \hand-first to $F_E$, then proceeds \head-first to recursively unfold $E$ and $W$;  
upon reaching $W$'s exit port on $F_W$, it moves \hand-first to $S_W$, \head-first to $L_S$, 
 \hand-first to $K_S$, then proceeds \head-first to recursively unfold $S$, ending at $A$'s exit port.
(Note that both $K_N$ and $K_S$ are open, since $N$ and $S$ are non-junctions.)  
We now show that, when visited in this order and laid flat in the plane,
the open faces in $\T_A$ form a net $\N_A$ that satisfies invariants (I1)-(I3). 

We start by showing that $\N_A$ that satisfies invariant (I2).
Note that $\xrightarrow{e_A} = \xrightarrow{e_N} \in R_I$ is open but not adjacent to $\T_N$, therefore $\N_N$ will provide a type-1 entry connection (by (I2) applied to $N$), which is also a type-1 entry 
connection for $A$ (because $e'_N = e'_A \in F_N$). 
Similarly, $\xleftarrow{x_A} = \xleftarrow{x_S} \in L_I$ is open but not adjacent to $\T_S$, 
therefore $S$ will provide a type-1 exit connection (by (I2) applied to $S$), which is also 
a type-1 exit connection for $A$ (because $x'_S = x'_A \in F_S$). This shows that $\N_A$ satisfies invariant (I2). Also note that (I3) is trivially satisfied, because $A$ has no open ring faces. 

It remains to show that $\N_A$ satisfies invariant (I1).
We begin by showing that $\N_A$ is connected: 
\begin{itemize}
\item
Observe that the exit port for $N$ is the top edge of $K_A$, and so $N$'s exit
ring face $x_N$ is on $K_A$. Its successor $\xleftarrow{x_{N}}$ is therefore on $L_A$ and is closed. 
 Invariant (I2) a applied to $N$ tells us that 
$\N_N$ provides a type-1 exit connection $x'_N \in K_N$ alongside its exit port. 
\item When the unfolding path reaches $N$'s exit port, it
deviates from prior unfoldings in that it doesn't move onto $x_N \in K_A$.  Instead it stays on $N$ and moves \hand-first
across $K_N$ to $R_N$ (which is open
because, if there were a box $NE$ adjacent to it, then boxes $NE, E, A, N$ would form a cycle).
Therefore, a new technique described here is used to connect $\N_N$ to the rest of $\N_A$. 
Note that the ring face
of $N$ located along the bottom of $R_N$ is adjacent to $x'_N \in K_N$.
In addition, this ring face
is not used as an entry or exit connection in $\N_N$ (because $\N_N$  has
type-1 entry/exit connections), so by invariant (I3) applied to $N$, 
it can be relocated outside of $\N_N$ without disconnecting $\N_N$.  
We relocate it to the right of $\N_N$'s exit port, where it connects to 
$\N_N$'s type-1 exit connection $x'_N  \in K_N$, as shown in \autoref{fig:NEWSdegree5}.  
This relocated piece of $R_N$ serves as a bridge to the unfolding of the next box $E$. 
\item Next we turn to $\N_E$. The recursive unfolding applied to $E$ 
uses the front edge of $R_A$ for its entry port and the back edge of $R_A$ for its
exit port. With this unfolding, $e'_E \in F_E$, $e_E \in R_I$, and while 
$\xrightarrow{e_E}  \in B_I$ is open, it is not adjacent to $\T_E$. 
Therefore the invariant (I2) applied to $E$ tells us that it provides a type-1 entry connection.  
Similarly, $x_E \in K_A$ and $\xleftarrow{x_E} \in T_A$ is closed. Thus
$\N_E$ also provides a type-1 exit connection.
The ring face of $E$ located along the left edge of $T_E$
is not used as an entry or exit connection for $E$ and so by invariant (I3) (applied to $E$),
it can be relocated outside of $\N_E$ without disconnecting it.  In the
unfolding in \autoref{fig:NEWSdegree5}, it is relocated to the left of $\N_E$'s entry port.
This relocated piece of $T_E$ serves as a bridge to the unfolding of the previous box $N$.
Thus the two relocated ring faces (one a piece of $R_N$ taken from $\N_N$ and the other a piece of $T_E$ taken from $\N_E$) form a bridge 
between the exit connection $x'_N \in K_N$ of $\N_N$ and the entry connection $e'_E \in F_E$
of $\N_E$. 
Finally, $\N_E$'s type-1 exit connection $x'_E$ connects to $x_E \in K_A$ shown unfolded alongside
$\N_E$'s exit port.
\item Similar arguments hold for $\N_W$. Note that the entry (exit) port for $W$ is
the back (front) edge of $L_A$. Also note that $\xrightarrow{e_W} \in B_A$ is closed and $\xleftarrow{x_W} \in T_I$ is open but not adjacent to $\T_W$, therefore $\N_W$ provides type-1 entry and exit connections.  Its entry connection attaches to $e_W \in K_A$
and its exit connection attaches to the ring face of $W$ located along the right edge of $B_W$, which has been
relocated right of the exit port of $\N_W$. 
\item Similar arguments hold for $\N_S$. Note that the entry (exit) port for $S$ is
the back (front) edge of $B_A$. Also note that  
$\xrightarrow{e_S} \in R_A$ is closed, therefore   
$\N_S$ provides a type-1 entry connection $e'_S \in K_S$. Its entry connection
attaches to the ring face of $S$ located along the top edge of $L_S$,
which has been relocated left of the entry port of $\N_S$.   
\end{itemize}
We conclude that $\N_A$ is connected. By invariant (I1), $\N_N$, $\N_E$, $\N_W$ and $\N_{S}$ include all open faces in $\T_N$, $\T_E$, $\T_W$ and $\T_{S}$ respectively, using a $4 \times 4$ refinement. 
Observe that the net $\N_A$ from~\autoref{fig:NEWSdegree5} also includes the open face $K_A$ of $A$ without any refinement. 
This shows that $\N_A$ includes all open faces in $\T_A$ using a $4 \times 4$ refinement.
Finally, $\N_A$ fits within $A$'s inductive region (as illustrated in \autoref{fig:NEWSdegree5}), noting that 
it does not utilize $\XE_A$ or $\XX_A$. 
We therefore conclude that $\N_A$ satisfies invariant (I1).

%
\begin{figure}[ht]
\centering
\includegraphics[page=1,width=0.7\textwidth]{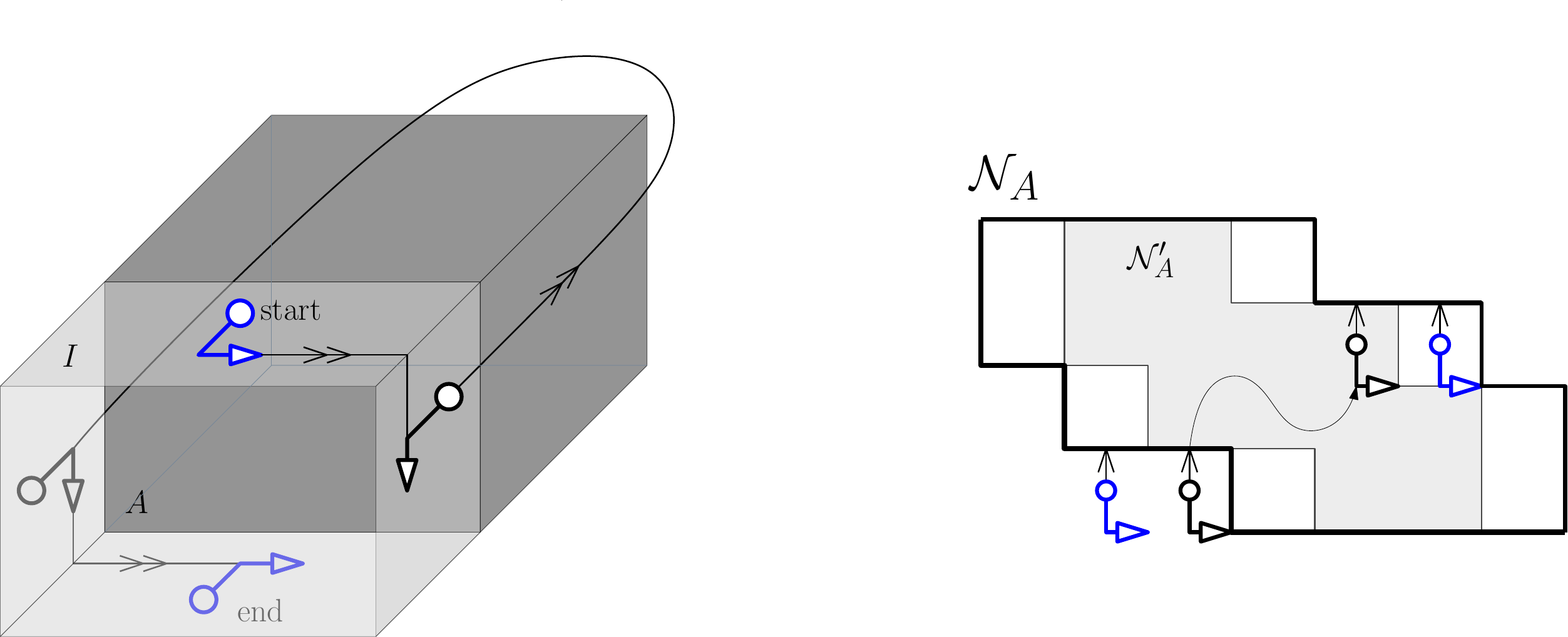}
\caption{Unfolding of box $A$ with non-junction parent $I$.}
\label{fig:reduce}
\end{figure}
%
The case where $E$ are $W$ are non-junctions can be reduced to the case where $N$ are $S$ are non-junctions using the method depicted in~\autoref{fig:reduce}: from the entry port, the unfolding proceeds \hand-first to $R_I$ (note that $I$ is a non-junction in our context, so both $T_I$ and $R_I$ are open), then follows the path from~\autoref{fig:NEWSdegree5} (imagine the box from~\autoref{fig:NEWSdegree5} rotated clockwise by $90^\circ$, so that its entry guide  aligns with the guide on $R_I$ from~\autoref{fig:reduce}). Then the net labeled $\N'_A$ in~\autoref{fig:reduce} is identical to the net from~\autoref{fig:NEWSdegree5}. From the exit port of $\N'_A$ on $L_I$, the unfolding proceeds \hand-first to the exit port of $\N_A$ on $B_I$. 

We have already established that $\N'_A$ satisfies invariants (I1)-(I3). Now note that the net $\N'_A$ from~\autoref{fig:NEWSdegree5} provides type-1 entry and exit connections, which implies that the net $\N_A$ from~\autoref{fig:reduce} provides type-2 entry and exit connections. These together with the fact that $\xrightarrow{e_A} \in R_I$ and $\xleftarrow{x_A} \in L_I$ are open and adjacent to $\T_A$, imply that $\N_A$ satisfies invariants (I1)-(I3). 
\end{proof}


\subsection{Unfolding Degree-6 Nodes}
\label{sec:degree6}
The following observation follows immediately from the tree structure of $\T$. 

\begin{proposition}
\label{prop:degree6}
Every neighbor of a  degree-6 node in $\T$ is a connector or a leaf.
\end{proposition}
We now show that invariants (I1)-(I3) hold for any degree-6 box. 
%
\begin{figure}[ht]
\centering
\includegraphics[width=\linewidth]{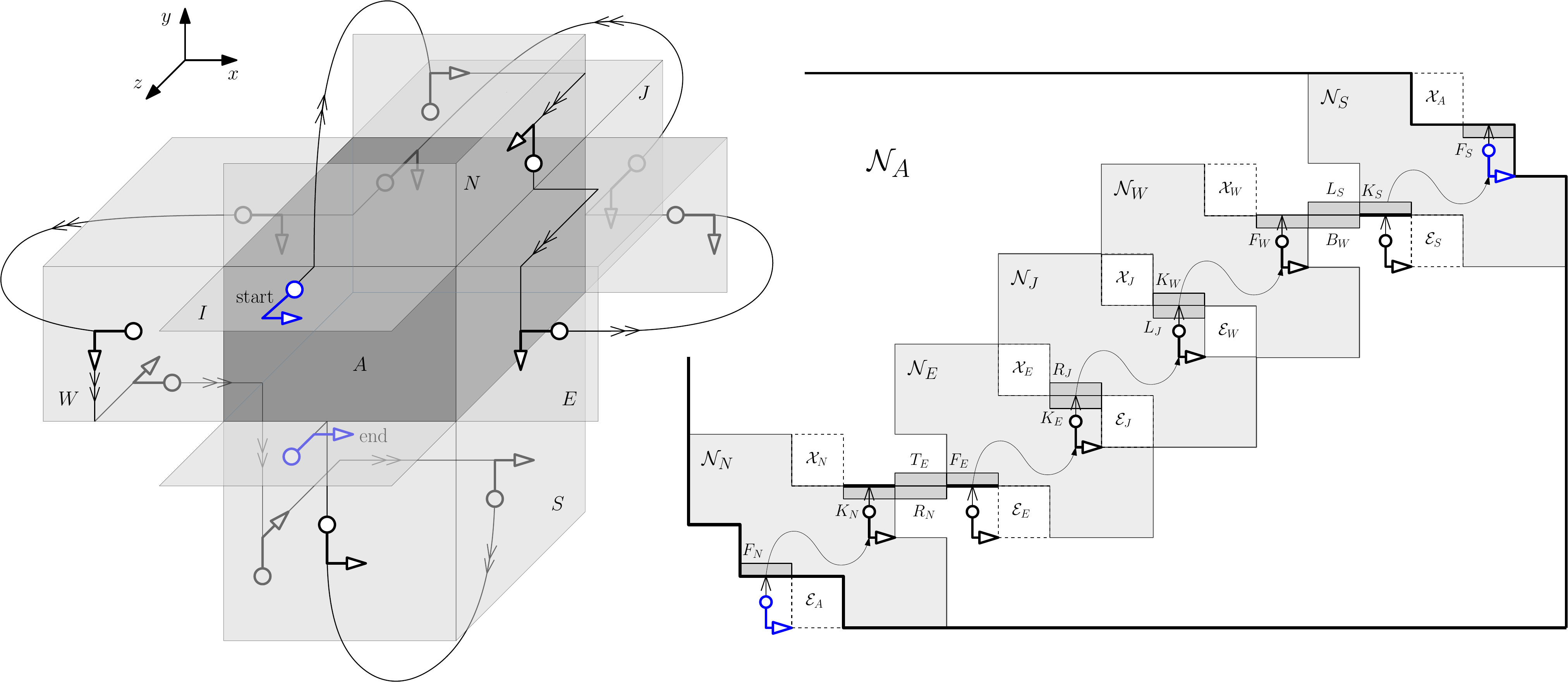}
\caption{Unfolding of degree-6 box $A$ (a) unfolding path (b) unfolding net $\N_A$.}
\label{fig:degree6}
\end{figure}
%

\begin{theorem}
\label{thm:degree6}
Let $A \in \T$ be a degree-6 box. If $A$'s children satisfy invariants (I1)-(I3), then 
$A$ satisfies invariants (I1)-(I3). 
\end{theorem}
\begin{proof}
Consider the unfolding depicted in~\autoref{fig:degree6}.
Observe that it is a generalization of the
degree-5 unfolding from~\autoref{fig:NEWSdegree5}, where the unfolded face $K_A$ is replaced by the recursive unfolding of child $J$. This generalization is possible because 
$N$ and $S$ are non-junctions by~\autoref{prop:degree6}. Since the two unfoldings and their proofs are very similar, we only point out the differences here. 

We first note that all children of $A$ provide type-1 entry and exit connectors, since they are all leaves or connector boxes by \autoref{prop:degree6}, and the unfoldings for these types of boxes use only type-1 connectors. In particular,
this means that the type-1 exit connector $x'_E \in K_E$ of $\N_E$ connects to the type-1
entry connector $e'_J \in R_J$ of $\N_J$, as shown in \autoref{fig:degree6}. 
It also means that the type-1 exit connector $x'_J \in L_J$ of $\N_J$ connects to the type-1
entry connector $e'_W \in K_W$ of $\N_W$, also shown in \autoref{fig:degree6}.
Thus $\N_A$ is connected.
 
Applying arguments similar to those in \autoref{lem:degree5} and noting that $\N_J$ includes all open faces in
$\T_J$ with $4 \times 4$ refinement (by invariant (I1) applied to $J$), we conclude that the net $\N_A$ from \autoref{fig:EJWdegree4} satisfies invariants (I1)-(I3).
 \end{proof}

\section{Complete Unfolding Example}
\label{sec:example1}
%
\begin{figure}[htp]
\centering
\includegraphics[page=1,width=0.9\linewidth]{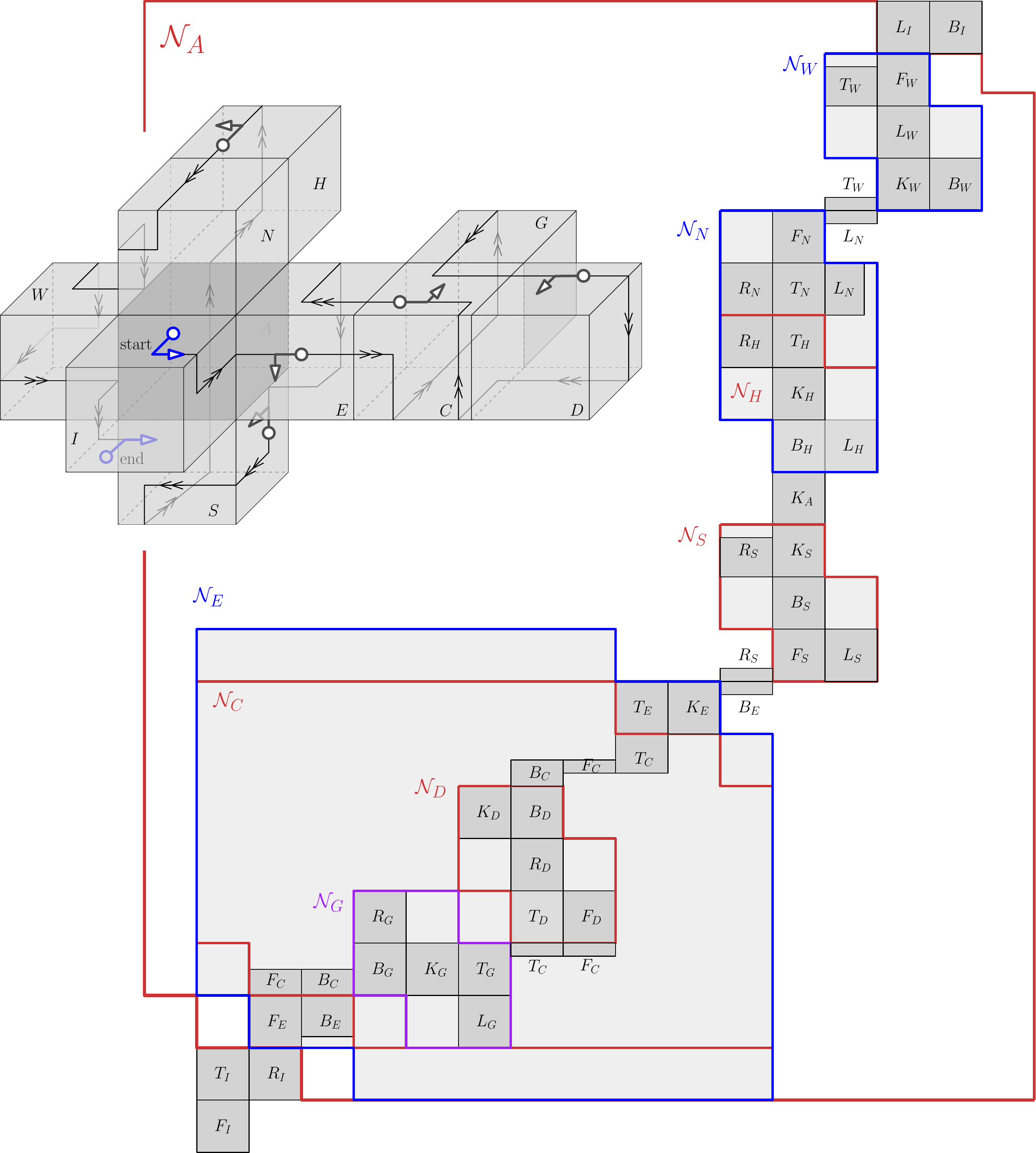}
\caption{Complete unfolding example of a polycube tree (with root $I$). Here $C$ is the box that requires a $4$-refinement along one grid dimension.}
\label{fig:example1}
\end{figure}
%

\autoref{fig:example1} illustrates a complete unfolding example for a polycube tree composed of ten boxes. The root $I$ of the the unfolding tree $\T$ is a degree-1 box with back child $A$, which is unfolded recursively. Observe that $A$ is a degree-5 box with non-junction children $E$ and $W$, therefore its unfolding follows the pattern from~\autoref{fig:reduce} (which employs the unfolding from~\autoref{fig:NEWSdegree5} in constructing $\N'_A$). In the following we classify the nodes in $\T$ based on their degree and orientation, and map them to the unfolding patterns discussed in earlier sections. To be able to do so, we view each node in $\T$ in standard position (with parent attached to the front face and entry and exit ports on top and bottom edges of the front face, respectively):
\begin{itemize}
\item The east child $E$ of $A$ is a degree-2 box with back child $C$, so its unfolding follows the pattern from~\autoref{fig:JNdegree2}a. 
\item $C$ is a degree-3 box with back child $D$ and south child $G$, so its unfolding follows the pattern from~\autoref{fig:NJdegree3}a, traversed in reverse. 
\item $N$ is a degree-2 box with north child $H$, so its unfolding follows the pattern from~\autoref{fig:JNdegree2}b. 
\item $D$, $S$, $H$ and $W$ are leaves that employ the \head-first unfolding pattern from~\autoref{fig:degree1}a.
\item $G$ is a leaf that uses the \hand-first unfolding pattern from~\autoref{fig:degree1}b.
\end{itemize}
The result is the net depicted in~\autoref{fig:example1}, with the subnets marked and appropriately labeled.

\section{Conclusion}
We show that every polycube tree can be unfolded with a $4 \times 4$ refinement of the grid faces. This is 
the first result on unfolding arbitrary polycube trees using a constant refinement of the grid. 
It is open whether all polycube trees can be grid-unfolded without any refinements. 


\newcommand{\etalchar}[1]{$^{#1}$}

\newpage
\begin{appendix}

\section{Unfolding Degree-3 Nodes (Remaining Cases)}
\label{sec:degree3-appendix}
This and subsequent appendices discuss unfoldings for cases not included in the main body of the paper. 
We illustrate the unfolding path and the resulting unfolding net for each case scenario, then present a digest of the correctness proof that focuses on the specifics of each case. When combined with arguments similar to the ones used in the main part of the paper, each proof digest yields a complete correctness proof. This way we avoid repetition and improve the readability flow. 

\medskip
\noindent
In this section we discuss the unfoldings for cases 3.3 through 3.6 listed in Section~\ref{sec:degree3}.

%
\begin{figure*}[htbp]
\centering
\includegraphics[width=\linewidth]{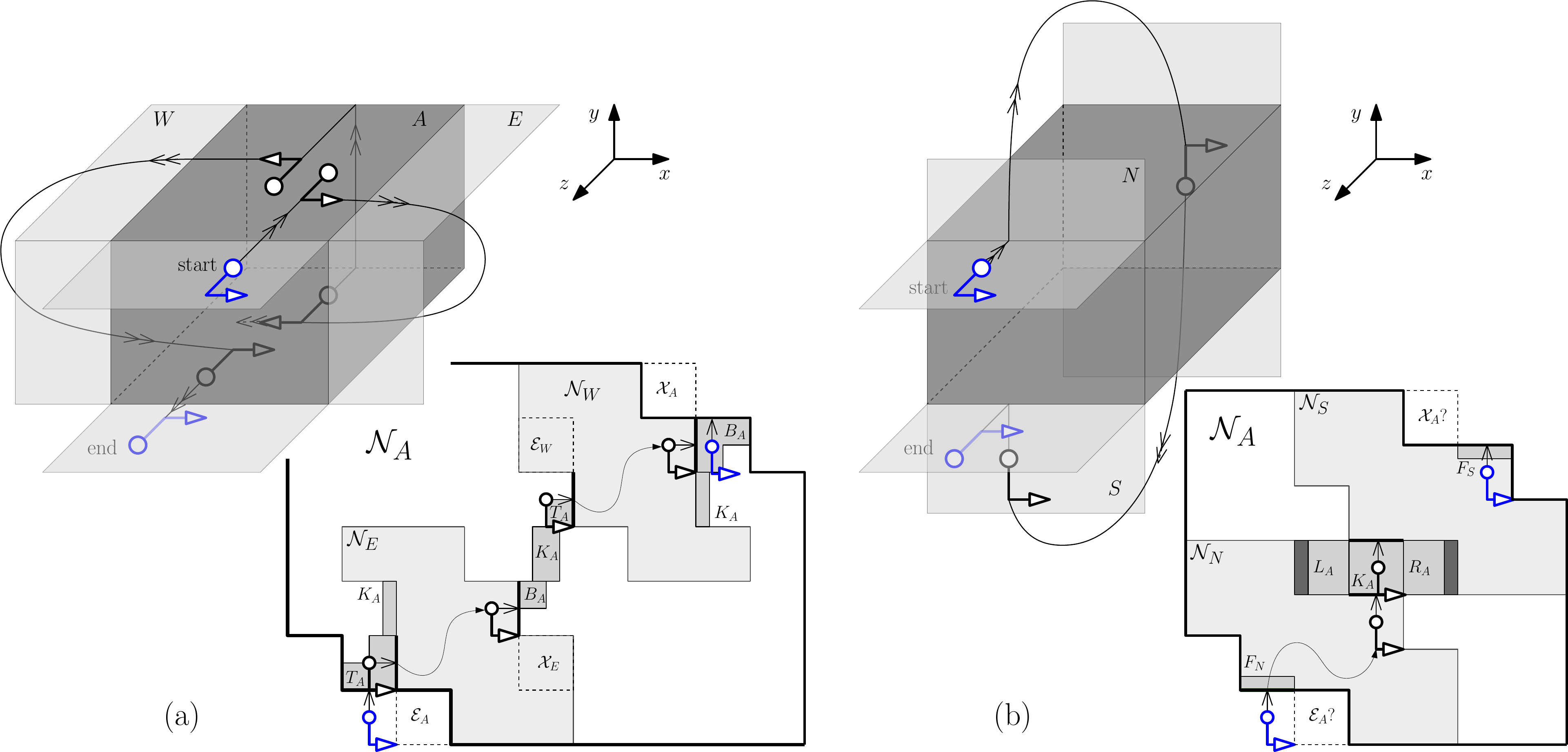}
\caption{Unfolding degree-3 box $A$ with parent $I$ (a) children $E$ and $W$ (b) children $N$ and $S$.}
\label{fig:EW-NSdegree3}
\end{figure*}
%
\begin{lemma}
\label{lem:EWdegree3}
Let $A \in \T$ be a degree-3 node with parent $I$ and children $E$ and $W$ \emph{(Case 3.3)}.
If $A$'s children satisfy invariants (I1)-(I3), then $A$ satisfies invariants (I1)-(I3).
\end{lemma}
\begin{proof}
Consider the unfolding from~\autoref{fig:EW-NSdegree3}a, and notice that this unfolding is a degenerate case of the 
unfolding from~\autoref{fig:EJWdegree4}, where the recursive unfolding of the child $J$ is replaced by the face $K_A$. 
Since the two unfoldings and their proofs of correctness are very similar, we only
point out the differences here:
\begin{itemize}
\item Since $\xrightarrow{e_E} \in K_A$ is open and adjacent to $\T_E$, the unit square $\XE_E$ (occupied by $\xrightarrow{e_E}$ 
in~\autoref{fig:EW-NSdegree3}a) does not belong to the inductive region for $E$ and $\N_E$ may provide a type-1 or a type-2 entry connection: if type-1, it connects to the ring face $e_E \in T_A$ placed alongside its entry port (as in the general case from~\autoref{fig:EJWdegree4}); if type-2, it connects to the ring face $\xrightarrow{e_E} \in K_A$ placed alongside its entry port extension. 
\item  Similarly, since $\xleftarrow{x_W} \in K_A$ is open and adjacent to $\T_W$, the unit square $\XX_W$ (occupied by $\xleftarrow{x_W}$ in~\autoref{fig:EW-NSdegree3}a) does not belong to the inductive region for $W$ and $\N_W$ may provide a type-1 or a type-2 exit connection: if type-1, it connects to the ring face $x_W \in B_A$ placed alongside its exit port (as in the general case from~\autoref{fig:EJWdegree4}); if type-2, it connects to the ring face $\xleftarrow{x_W} \in K_A$ placed alongside its exit port extension. 
\end{itemize}
These changes are reflected in~\autoref{fig:EW-NSdegree3}a. Arguments similar to the ones used in the proof of~\autoref{lem:EJWdegree4} show that the net $\N_A$ from~\autoref{fig:EW-NSdegree3}a satisfies invariants (I1)-(I3). 
\end{proof}

\begin{lemma}
\label{lem:NSdegree3}
Let $A \in \T$ be a degree-3 node with parent $I$ and children $N$ and $S$ \emph{(Case 3.4)}.
If $A$'s children satisfy invariants (I1)-(I3), then $A$ satisfies invariants (I1)-(I3).
\end{lemma}
\begin{proof}
Consider the unfolding from~\autoref{fig:EW-NSdegree3}b, and notice that it is a generalization of the 
unfolding from~\autoref{fig:JNdegree2}b, where the unfolded face $B_A$ is replaced by the recursive unfolding of $S$. 
Since the two unfoldings and their proofs of correctness are very similar, we only
point out the differences here:
\begin{itemize}
\item The entry and exit ring faces for $S$ are $e_S \in K_A$  and $x_S \in B_I$, respectively.
\item Since $\xrightarrow{e_S} \in R_A$ is open and adjacent to $\T_S$, the unit square $\XE_S$ (occupied by $R_A$ in~\autoref{fig:EW-NSdegree3}b) does not belong to the inductive region for $S$ and $\N_S$ may provide a type-1 or a type-2 entry connection: if type-1, it connects to the ring face $e_S \in K_A$ placed alongside its entry port; if type-2, it connects to the ring face $\xrightarrow{e_S} \in R_A$ placed alongside its entry port extension. 
\item Since $\xleftarrow{x_S} \in L_I$ is not adjacent to $\T_S$, $\N_S$ will provide a type-1 exit connection, which is also a type-1 exit connection for $\N_A$ (because $x_S = x_A$). 
\end{itemize}
These observations, along with the arguments used in the proof of~\autoref{lem:Ndegree2}, show that the unfolding $\N_A$ from~\autoref{fig:EW-NSdegree3}b satisfies invariants (I1)-(I3).
\end{proof}

\begin{lemma}
\label{lem:NEdegree3}
Let $A \in \T$ be a degree-3 node with parent $I$ and children $N$ and $E$ \emph{(Case 3.5)}.
If $A$'s children satisfy invariants (I1)-(I3), then $A$ satisfies invariants (I1)-(I3).
\end{lemma}
\begin{proof}
%
\begin{figure*}[htpb]
\centering
\includegraphics[page=1,width=\linewidth]{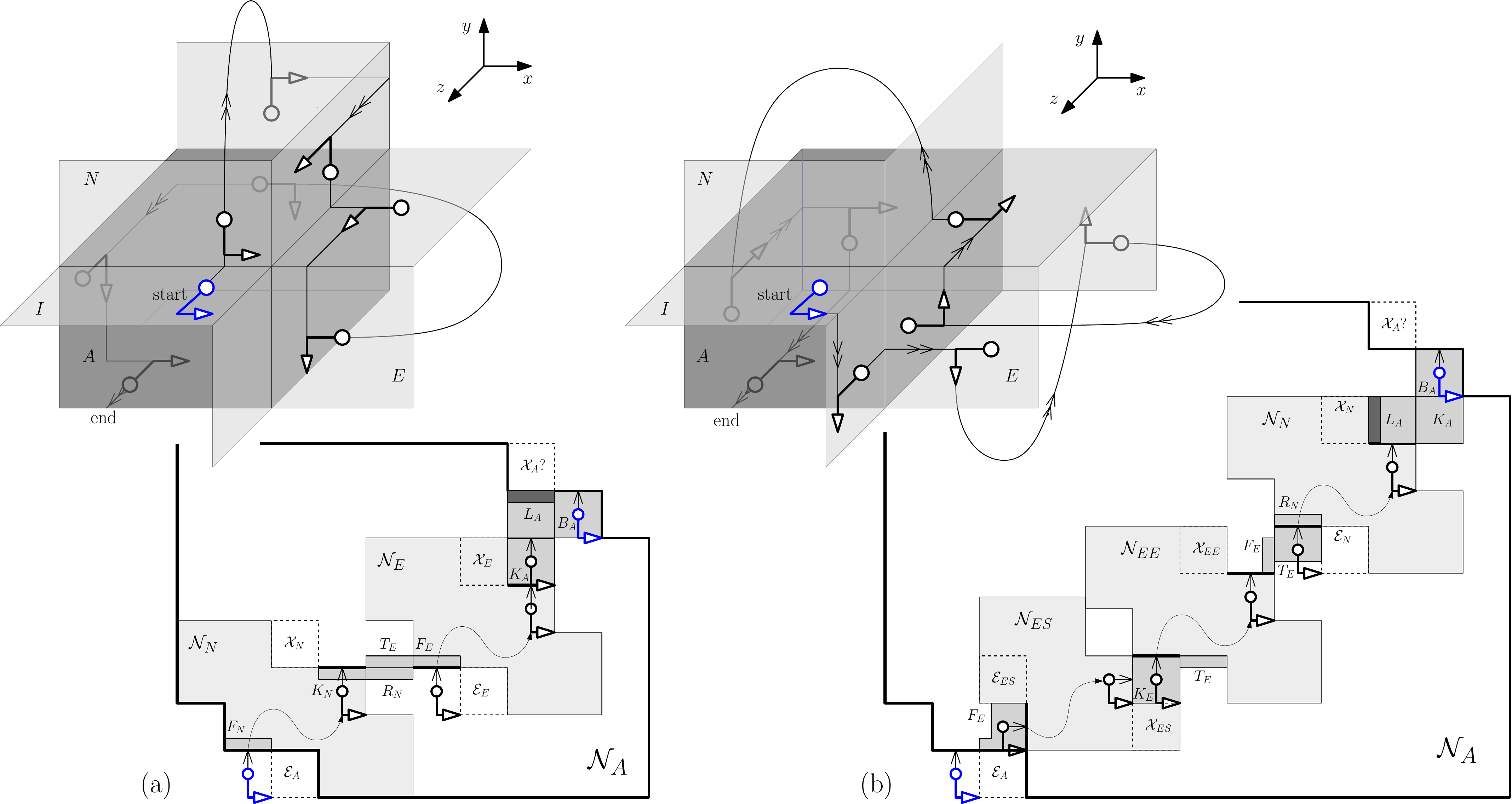}
\caption{\hand-east unfolding of degree-3 box $A$ with parent $I$ and children $N$ and $E$ (a) $K_N$ open 
(b) $K_N$ closed (so $K_E$ open); unfolding shown for case when $ES$ and $EE$ exist. }
\label{fig:NEdegree3a}
\end{figure*}

We discuss two different scenarios,  depending on whether $K_N$ is open or closed. Assume first that 
$K_N$ is open, and consider the \hand-east unfolding depicted in~\autoref{fig:NEdegree3a}a. Notice that this unfolding follows a path very similar to the one from~\autoref{fig:NEWSdegree5} (which depicts the case where $A$ has two additional children $W$ and $S$), so in a way this case can be viewed as a degenerate case of the one from~\autoref{fig:NEWSdegree5}. The only difference is that, in~\autoref{fig:NEdegree3a}a, once the unfolding path reaches the back face $K_A$, it continues \head-first to $L_A$ and then \hand-first to the exit port of $A$. Note that the resulting net $N_A$ provides a type-1 exit connection $x'_A \in B_A$, and the ring face $\xleftarrow{x'_A} \in L_A$ (dark-shaded in~\autoref{fig:NEdegree3a}a) can be removed from $\N_A$ without disconnecting $\N_A$. These observations, combined with the arguments used in the proof of~\autoref{lem:degree5}, show that $\N_A$ satisfies invariants (I1)-(I3). 

Assume now that $K_N$ is closed (note that in this case $K_E$ is open), and consider the \hand-east unfolding depicted in~\autoref{fig:NEdegree3a}b, which handles the more general case where $ES$ and $EE$ exist (handling cases when one or both of these boxes are missing requires only minor modifications). Note that $\xrightarrow{e_A} \in R_I$ is open and adjacent to $\T_A$, and 
$\N_A$ provides a type-2 entry connection $\xrightarrow{e'_A} \in F_E$. Also note that $\N_A$ provides a type-1 exit connection $\xrightarrow{x'_A} \in B_E$. These together show that $\N_A$ satisfies invariant (I2).
The following observations support our claim that $\N_A$ satisfies invariant (I1):
\begin{itemize}
\item The entry and exit ring faces for $ES$, $EE$ and $N$ are as follows: $e_{ES} \in F_E$ and $x_{ES} \in K_E$; 
$e_{EE} \in K_E$ and $x_{EE} \in F_E$; and $e_{N} \in T_E$ and $x_{N} \in L_A$.
\item $\N_{ES}$ and $\N_N$ provide type-1 entry connections. This is because 
$\xrightarrow{e_{ES}} \in R_E$ is closed, and $\xrightarrow{e_{N}} \in K_E$ is not adjacent to $\T_E$. 
\item Since $\xrightarrow{e_{EE}} \in T_E$ is open, 
the unit square $\XE_{EE}$ (occupied by $\xrightarrow{e_{EE}}$ in~\autoref{fig:NEdegree3a}b) does not belong to 
the inductive region for $EE$.
%
\item $\N_{ES}$, $N_{EE}$ and $\N_N$ provide type-1 exit connections. This is because 
$\xleftarrow{x_{ES}} \in L_E$, $\xleftarrow{x_{EE}} \in B_E$ and $\xleftarrow{x_{N}} \in F_A$ are all closed. 
\end{itemize}
Regarding invariant (I3), note that the open ring face of $A$ located on $L_A$ 
(dark-shaded in~\autoref{fig:NEdegree3a}b) can be removed from $\N_A$ without disconnecting $\N_A$, therefore $\N_A$ satisfies (I3).

\end{proof}

\begin{lemma}
\label{lem:NWdegree3}
Let $A \in \T$ be a degree-3 node with parent $I$ and children $N$ and $W$ \emph{(Case 3.6)}.
If $A$'s children satisfy invariants (I1)-(I3), then $A$ satisfies invariants (I1)-(I3).
\end{lemma}
\begin{proof}

\medskip
\noindent
This case is slightly more complex and spans four different case scenarios:
\begin{enumerate}
\item $B_I$ open
\item $T_N$ open
\item $B_I$ closed, $T_N$ closed and $K_N$ open
\item $B_I$ closed, $T_N$ closed and $K_N$ closed
\end{enumerate}

%
\begin{figure*}[htp]
\centering
\includegraphics[page=1,width=\linewidth]{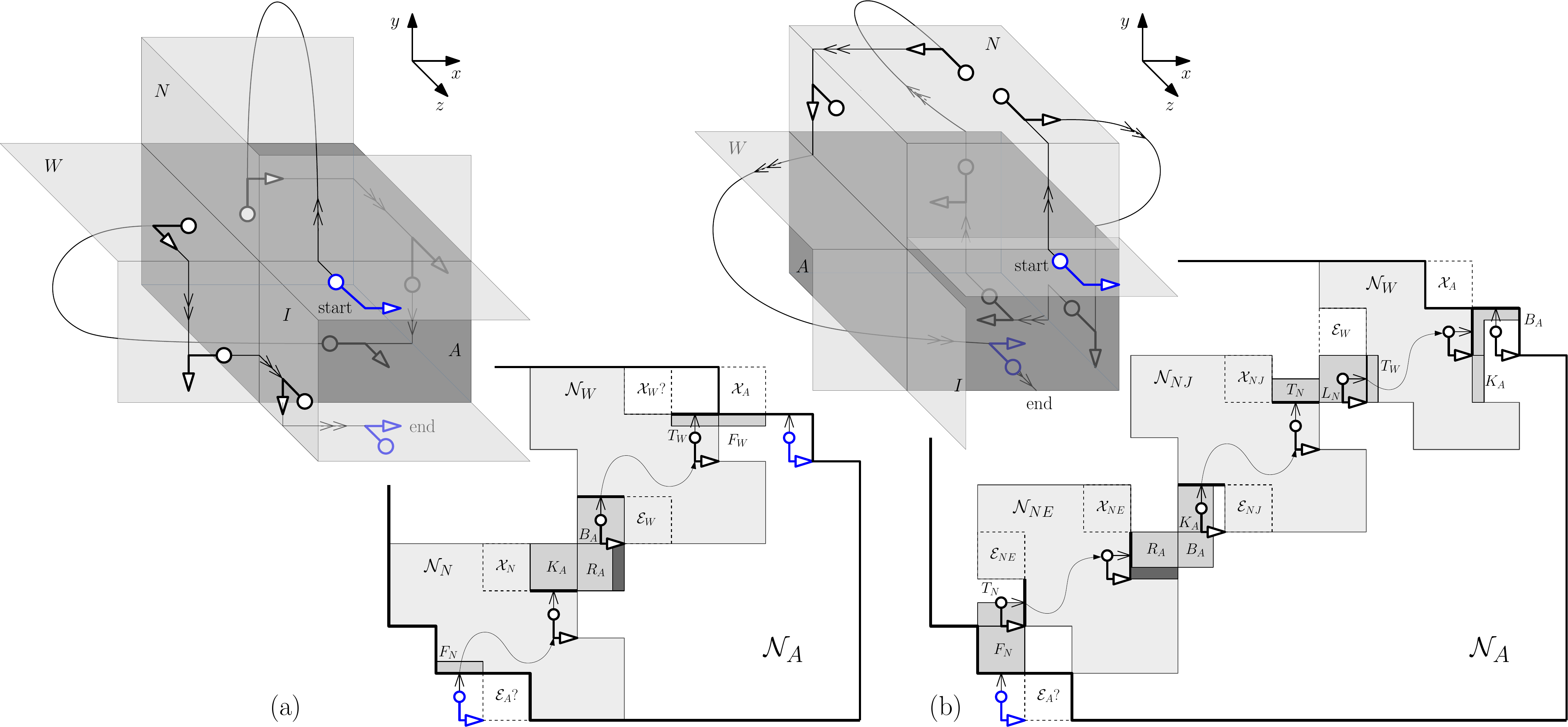}
\caption{Unfolding of degree-3 box $A$ with parent $I$ and children $N$ and $W$
(a) $B_I$ open  
(b) $T_N$ open; unfolding shown for the case when $NE$ and $NJ$ exist. 
}
\label{fig:NEdegree3b}
\end{figure*}
%

\noindent
{\bf Case 1: $B_I$ open.} Consider the unfolding depicted in~\autoref{fig:NEdegree3b}a, and identify the following entry and exit ring faces for $N$ and $W$: $e_N \in T_I$ and $x_N \in K_A$; and $e_W \in B_A$ and $x_W \in L_N$. 
Note that $\xrightarrow{e_N} \in R_I$ is not adjacent to $\T_N$, therefore $\N_N$ provides a type-1 entry connection $\xrightarrow{e'_N} \in F_N$, which is also a type-1 entry connection for $\N_A$ (since $e_A = e_N$). Also note that $\xleftarrow{x_A} \in L_I$ is open and adjacent to $\T_A$, and $\N_A$ provides a type-2 exit connection $\xleftarrow{x'_A} \in F_W$. These together show that $\N_A$ satisfies invariant (I2). Turning to (I1), note that $\N_N$ and $\N_W$ provide type-1 entry and exit connections. This is because 
$\xleftarrow{x_N} \in L_A$ is closed, 
$\xrightarrow{e_W} \in F_A$ is closed, 
and $\xleftarrow{x_{W}} \in K_N$ is not adjacent to $\T_W$.
These together imply that $\N_A$ satisfies invariant (I1).
Finally, note that the ring face of $A$ located on $R_A$
(dark-shaded in~\autoref{fig:NEdegree3b}a) can be removed without disconnecting $\N_A$, so $\N_A$ satisfies invariant (I3) as well. 

\medskip
\noindent
{\bf Case 2: $T_N$ open.} Consider the unfolding depicted in~\autoref{fig:NEdegree3b}b, which handles the more general case where $NE$ and $NJ$ exist (handling cases when one or both of these boxes are missing requires only minor modifications). Note that $\N_A$ provides type-1 entry and exit connections $e'_A \in F_N$ and $x'_A \in B_A$, therefore $\N_A$ satisfies invariant (I2). 
The following observations support our claim that $\N_A$ satisfies invariant (I1):
\begin{itemize}
\item The entry and exit ring faces for $NE$, $NJ$ and $W$ are as follows: 
$e_{NE} \in T_N$ and $x_{NE} \in R_A$; 
$e_{NJ} \in K_A$ and $x_{NJ} \in T_N$; 
and $e_{W} \in L_N$ and $x_{W} \in B_A$. 
\item $\N_{NE}$, $\N_{NJ}$ and $\N_W$ provide type-1 entry connections. This is because
$\xrightarrow{e_{NE}} \in K_N$ and $\xrightarrow{e_{NJ}} \in L_A$ are closed, and 
$\xrightarrow{e_{W}} \in F_N$ is not adjacent to $\T_W$. 
\item $\N_{NE}$ and $\N_{NJ}$ provide type-1 exit connections, since 
$\xleftarrow{x_{NE}} \in F_A$ and $\xleftarrow{x_{NJ}} \in R_N$ are closed.
\item Since $\xleftarrow{x_W} \in K_A$ is open, the unit square $\XX_W$ (occupied by $\xleftarrow{x_W}$ in~\autoref{fig:NEdegree3b}b) does not belong to the inductive region for $W$. 
\end{itemize}
Arguments similar to the ones above show that $\N_A$ satisfies invariant (I3) as well.

\begin{figure*}[htbp]
\centering
\includegraphics[page=2,width=0.9\linewidth]{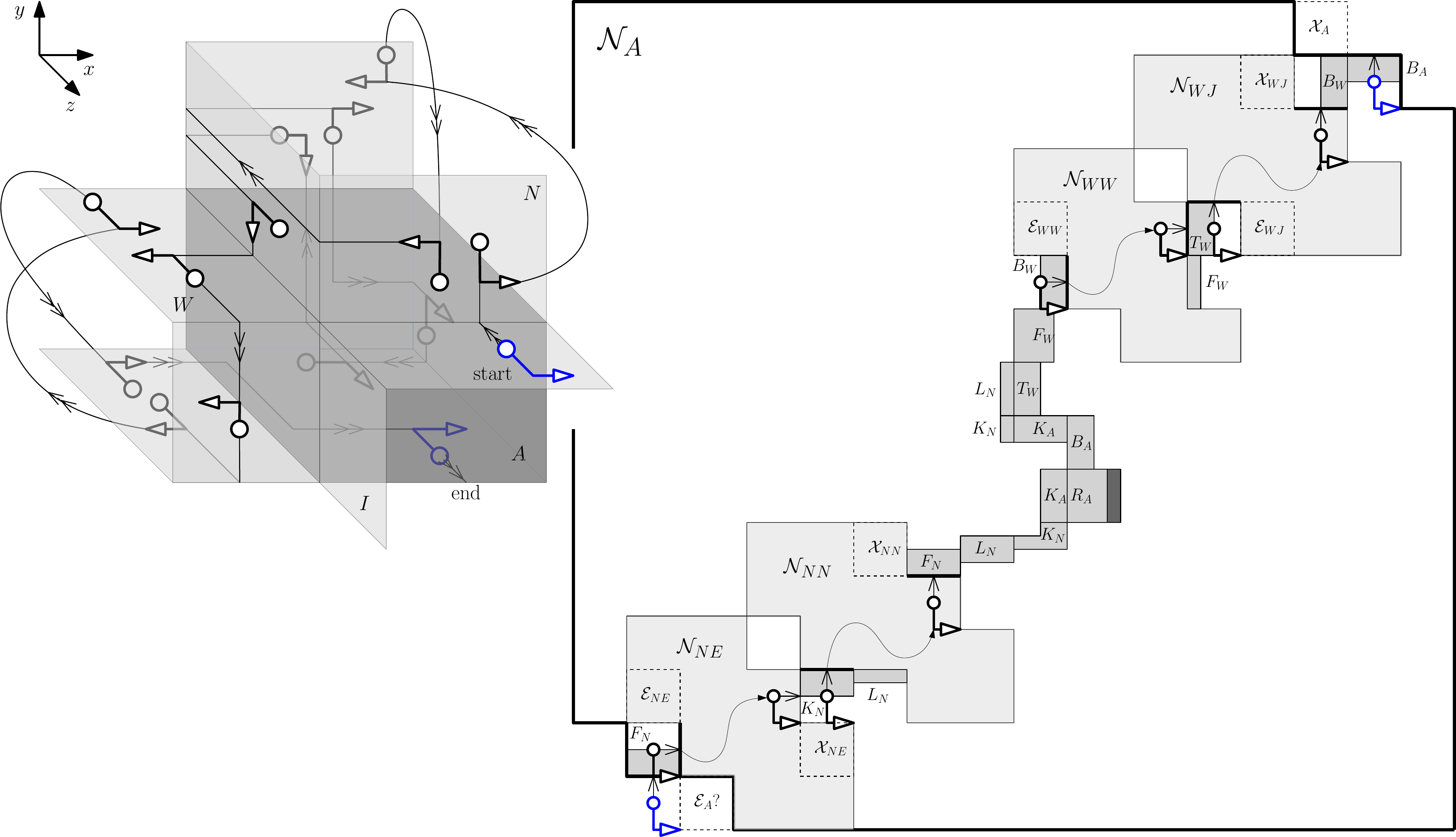}
\caption{Unfolding of degree-3 box $A$ with parent $I$ and children $N$ and $W$, case 
$B_I$ closed (so $B_W$ open), $T_N$ closed and $K_N$ open.
}
\label{fig:NEdegree3c}
\end{figure*}
%

\medskip
\noindent
{\bf Case 3: $B_I$, $T_N$ closed and $K_N$ open.} Note that in this case $B_W$ is open. 
Consider the unfolding depicted in~\autoref{fig:NEdegree3c}, which handles the more general case where $NE$, $WW$ and $WJ$ exist (handling cases when one or more of these boxes do not exist requires only minor modifications). 
Arguments similar to the ones above show that $\N_A$ satisfies invariants (I2) and (I3).
The following observations support our claim that $\N_A$ satisfies invariant (I1):
\begin{itemize}
\item The entry and exit ring faces for $NE$, $NN$, $WW$ and $WJ$ are as follows: 
$e_{NE} \in F_N$ and $x_{NE} \in K_N$; 
$e_{NN} \in K_N$ and $x_{NN} \in F_N$; 
$e_{WW} \in B_W$ and $x_{WW} \in T_W$; 
and $e_{WJ} \in T_W$ and $x_{WJ} \in B_W$. 
\item $\N_{NE}$, $\N_{WW}$ and $\N_{WJ}$ provide type-1 entry connections, since 
$\xrightarrow{e_{NE}} \in T_N$, $\xrightarrow{e_{WW}} \in K_W$ and $\xrightarrow{e_{WJ}} \in R_W$ are closed.
Since $\xrightarrow{e_{NN}} \in L_N$ is open, 
the unit square $\XE_{NN}$ (occupied 
by $\xrightarrow{e_{NN}}$ in~\autoref{fig:NEdegree3c}) does not belong to the inductive region for $NN$.
%
\item $\N_{NE}$, $\N_{NN}$ and $\N_{WJ}$ provide type-1 exit connections, 
since $\xleftarrow{x_{NE}} \in B_N$, $\xleftarrow{x_{NN}} \in R_N$ and $\xleftarrow{x_{WJ}} \in L_W$ are closed.
Since $\xleftarrow{x_{WW}} \in F_W$ is open, 
the unit square $\XX_{WW}$ (occupied by $\xleftarrow{x_{WW}}$ in~\autoref{fig:NEdegree3c}) does not belong to the inductive region for $WW$.
\end{itemize}

\begin{figure*}[htbp]
\centering
\includegraphics[page=3,width=0.9\linewidth]{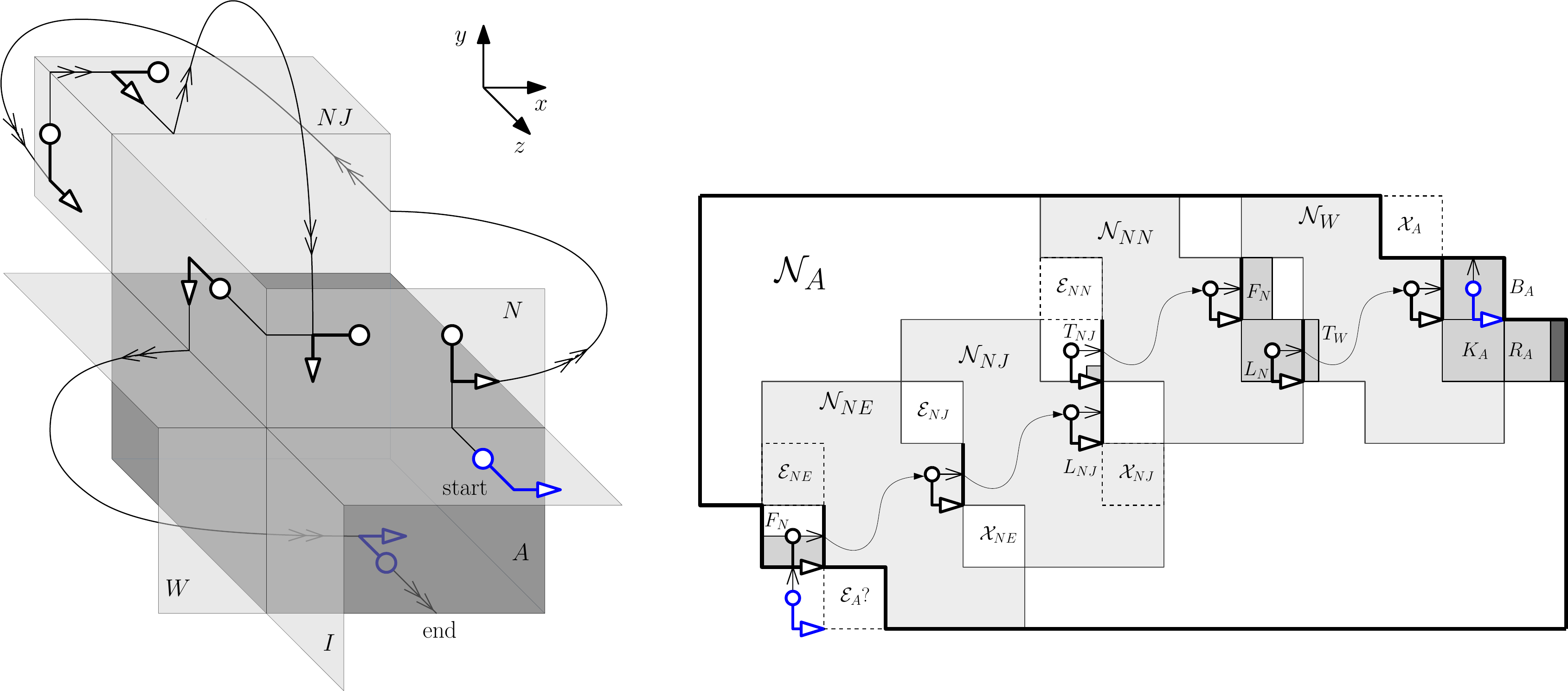}
\caption{Unfolding of degree-3 box $A$ with parent $I$ and children $N$ and $E$, case 
$B_I$ closed (so $B_W$ open), $T_N$ closed and $K_N$ closed (and so $T_{NJ}$ and $L_{NJ}$ open). 
}
\label{fig:NEdegree3d}
\end{figure*}
%

\medskip
\noindent
{\bf Case 4: $B_I$, $T_N$ and $K_N$ closed.} Note that in this case $NJ$ exists, and $T_{NJ}$ and $L_{NJ}$ are open. 
Consider the unfolding depicted in~\autoref{fig:NEdegree3d}, which handles the more general case where $NE$ exists (handling the case when $NE$ does not exist requires only minor modifications). 
Arguments similar to the ones above show that $\N_A$ satisfies invariants (I2) and (I3).
The following observations support our claim that $\N_A$ satisfies invariant (I1):
\begin{itemize}
\item The entry and exit ring faces for $NE$, $NJ$, $NN$ and $W$ are as follows: 
$e_{NE} \in F_N$ and $x_{NE} \in R_{NJ}$; 
$e_{NJ} \in K_{NE}$ and $x_{NJ} \in L_N$; 
$e_{NN} \in T_{NJ}$ and $x_{NN} \in F_N$; 
and $e_{W} \in L_N$ and $x_{W} \in B_A$. 
\item $\N_{NE}$, $\N_{NJ}$, $\N_{NN}$ and $\N_{W}$ provide type-1 entry connections. This is because 
$\xrightarrow{e_{NE}} \in T_N$ is closed, 
$\xrightarrow{e_{NJ}} \in T_{NE}$ is not adjacent to $\T_{NJ}$, 
$\xrightarrow{e_{NN}} \in R_{NJ}$ is not adjacent to $\T_{NN}$, and 
$\xrightarrow{e_{W}} \in F_N$ is not adjacent to $\T_W$. 
\item $\N_{NE}$ and $\N_{NJ}$ provide type-1 exit connections. This is because 
 $\xleftarrow{x_{NE}} \in B_{NJ}$ is not adjacent to $\T_{NE}$, and 
 $\xleftarrow{x_{NJ}} \in B_N$ is closed. Note that the type-1 exit connection of 
 $\N_{NE}$ connects to the type-1 entry connection of $\N_{NJ}$. 
\item Since $\xleftarrow{x_{NN}} \in L_N$ is open,  
the unit square $\XX_{NN}$ (occupied by $L_N$ in~\autoref{fig:NEdegree3d}) does not belong to the inductive region 
for $NN$.
%
\item Similarly, since $\xleftarrow{x_{W}} \in K_A$ is open, 
the unit square $\XX_W$ (occupied by $K_A$ in~\autoref{fig:NEdegree3d}) does not belong to the inductive region 
for $W$.
\item By invariant (I3) applied to $NJ$, the ring face that lies on $T_{NJ}$ (not used in the 
entry or exit connections for $NJ$) can be relocated outside of $\N_{NJ}$.  
In~\autoref{fig:NEdegree3d}, we use a piece of $T_{NJ}$ to connect $\N_{NJ}$ and $\N_{NN}$ together. 
\end{itemize}
Having exhausted all possible cases, we conclude that this lemma holds. 
\end{proof}

\section{Unfolding Degree-4 Nodes (Remaining Cases)}
\label{sec:degree4-appendix}
In this section we discuss the unfoldings for cases 4.2 through 4.7 listed in Section~\ref{sec:degree4}.

\begin{lemma}
\label{lem:NEWdegree4}
Let $A \in \T$ be a degree-$4$ node with parent $I$ and children $N$, $E$, and $W$ \emph{(Case 4.2)}. 
If $A$'s children satisfy invariants (I1)-(I3), then $A$ satisfies invariants (I1)-(I3).
\end{lemma}
\begin{proof}
We discuss the following three exhaustive scenarios:
\begin{enumerate}
\item $K_N$ closed
\item $K_N$ open and $B_I$ closed 
\item $K_N$ open and $B_I$ open
\end{enumerate}
%
\begin{figure}[h]
\centering
\includegraphics[page=1,width=\textwidth]{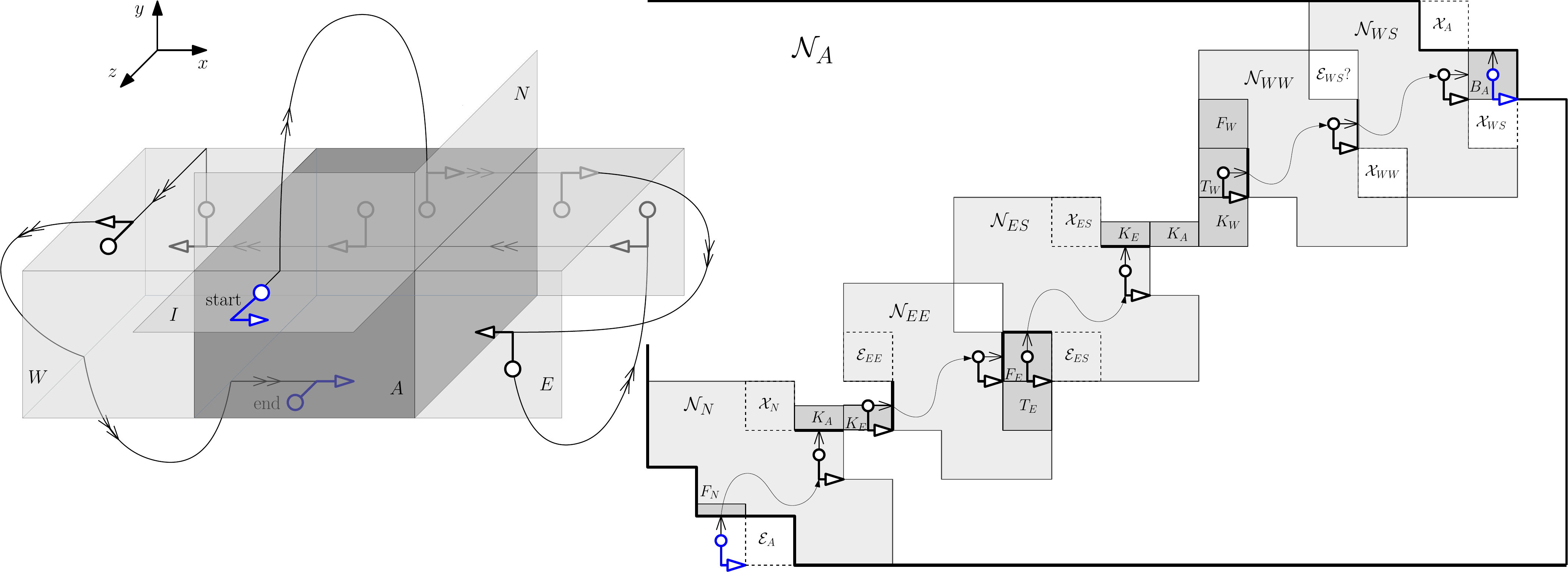}
\caption{Unfolding of degree-4 box $A$ with $N$, $E$ and $W$ children, case $K_N$ closed (so $K_E$, $K_W$ open); unfolding shown for the case when $EE$, $ES$, $WW$ and $WS$ exist.}
\label{fig:NEWdegree4a-1}
\end{figure}
%
{\bf Case 1: $K_N$ closed.} It can be easily verified that in this case $K_E$ and $K_W$ are open. 
Consider the unfolding depicted in~\autoref{fig:NEWdegree4a-1}, which handles the more general case where $EE$, $ES$, $WW$ and $WS$ exist 
(handling cases when one or more of these boxes do not exist requires only minor modifications). 
Note that $\N_A$ provides a type-1 entry connection (by arguments similar to the ones used in the proof of~\autoref{lem:degree5}) and a type-1 exit connection $x'_A \in B_A$, therefore 
$\N_A$ that satisfies invariant (I2).
Also note that the only open ring face of $A$ is the exit ring face, so $\N_A$ trivially satisfies (I3).  
The following observations support our claim that $\N_A$ is connected and satisfies invariant (I1):
\begin{itemize}
\item The entry and exit ring faces for $N$, $EE$, $ES$, $WW$ and $WS$ are as follows: 
$e_{N} \in T_I$ and $x_{N} \in K_A$; 
$e_{EE} \in K_{E}$ and $x_{EE} \in F_E$; 
$e_{ES} \in F_E$ and $x_{ES} \in K_E$; 
$e_{WW} \in T_W$ and $x_{WW} \in L_{WS}$; 
and $e_{WS} \in B_{WW}$ and $x_{WS} \in B_A$. 
\item $\N_N$, $\N_{ES}$, $\N_{WW}$ and $\N_{WS}$ provide type-1 exit connections. This is because 
$\xleftarrow{x_{N}} \in L_A$ is closed, $\xleftarrow{x_{ES}} \in R_E$ is closed, 
$\xleftarrow{x_{WW}} \in K_{WS}$ is not adjacent to $\T_{WW}$, 
and $\xleftarrow{x_{WS}} \in K_A$ is open but not adjacent to $\T_{WS}$. 
\item $\N_{EE}$, $\N_{ES}$ and $\N_{WS}$ provide type-1 entry connections. This is because $\xrightarrow{e_{EE}} \in B_E$ is closed (since $ES$ exists), $\xrightarrow{e_{ES}} \in L_E$ is closed, and $\xrightarrow{e_{WS}} \in F_{WW}$ is not adjacent to $\T_{WS}$. 
Note that the type-1 entry connection of $\N_{WS}$ connects to the type-1 exit connection of $\N_{WW}$.
\item Since $\xleftarrow{x_{EE}} \in T_E$ is open, the unit square 
$\XX_{EE}$ (occupied by $T_E$ in~\autoref{fig:NEWdegree4a-1}) does not belong to the inductive region 
for $EE$. 
\item Since $\xrightarrow{e_{WW}} \in F_W$ is open, the unit square 
$\XE_{WW}$ (occupied by $F_W$ in~\autoref{fig:NEWdegree4a-1}) does not belong to the inductive region 
for $WW$. 
\end{itemize}
Note that we split the unfolding of $W$ into two subnets ($\N_{WW}$ and $\N_{WS}$) so as to avoid sharing the ring face on $K_W$ between its current position in $\N_A$ and the type-2 exit connection that $\N_W$ would have provided (had it not been split). A similar intuition was used to split the unfolding of $E$ into $\N_{EE}$ and $\N_{ES}$. 

%
\begin{figure}[h]
\centering
\includegraphics[page=2,width=\textwidth]{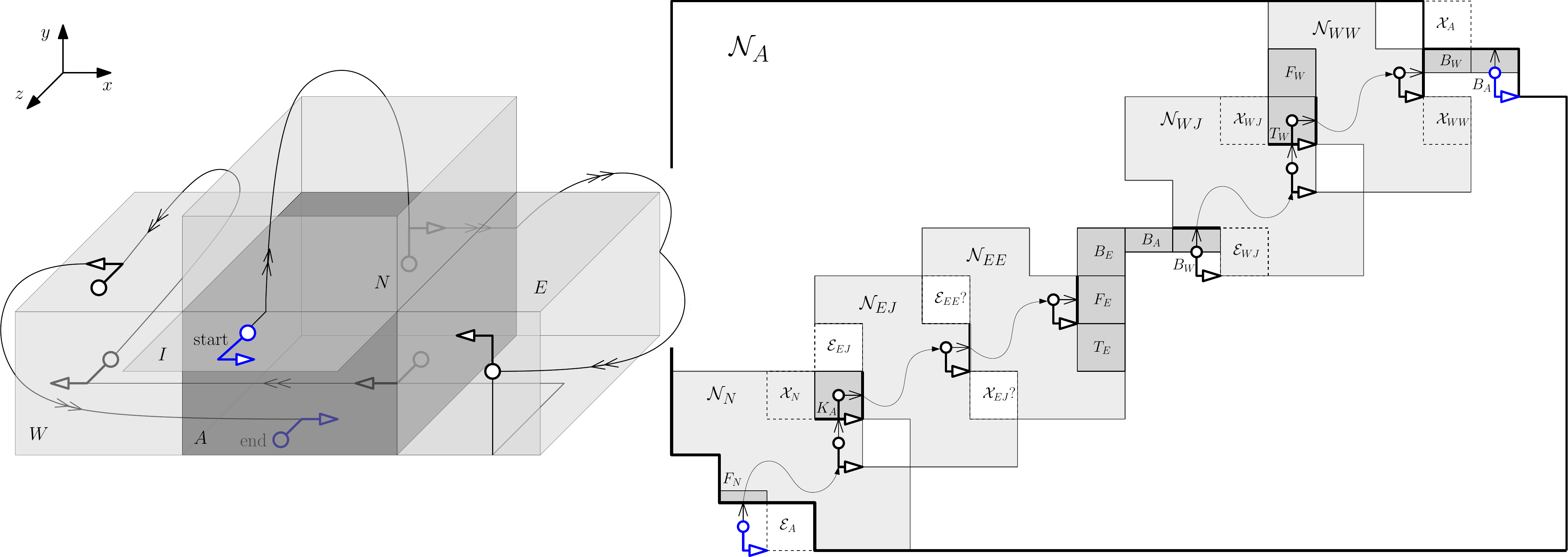}
\caption{Unfolding of degree-4 box $A$ with $N$, $E$ and $W$ children, case $K_N$ open and $B_I$ closed (so $B_E$, $B_W$ open); unfolding shown for the case when $EJ$, $EE$, $WJ$ and $WW$ exist.}
\label{fig:NEWdegree4a-2}
\end{figure}
%

\medskip
\noindent
{\bf Case 2: $K_N$ open and $B_I$ closed.} 
In this case $B_E$ and $B_W$ are open (refer to~\autoref{fig:NEWdegree4a-2}, which shows the unfolding for the case when $EJ$, $EE$, $WJ$ and $WW$ exist). 
Arguments similar to ones used in the previous case show that 
$\N_A$ satisfies invariants (I2) and (I3).
The following observations support our claim that $\N_A$ is connected and satisfies invariant (I1):
\begin{itemize}
%
\item Same arguments as in Case 1 apply to the entry and exit ports of $\N_N$.
\item The entry and exit ring faces for $EJ$, $EE$, $WJ$ and $WW$ are as follows: 
$e_{EJ} \in K_A$ and $x_{EJ} \in K_{EE}$; 
$e_{EE} \in R_{EJ}$ and $x_{EE} \in F_E$; 
$e_{WJ} \in B_W$ and $x_{WJ} \in T_W$; 
and $e_{WW} \in T_W$ and $x_{WW} \in B_W$. 
\item $\N_{EJ}$ provides a type-1 entry connection, since $\xrightarrow{e_{EJ}} \in B_A$ is not adjacent to $\T_{EJ}$. 
Also note that $\xleftarrow{x_{EJ}} \in T_{EE}$ is not adjacent to $\T_{EJ}$, therefore $\N_{EJ}$ provides a type-1 exit connection. 
\item Since  $\xrightarrow{e_{EE}} \in B_{EJ}$ is not adjacent to $\T_{EE}$, 
$\N_{EE}$ provides a type-1 entry connection (which attaches to the type-1 exit connection of $\N_{EJ}$).
Also, since $\xleftarrow{x_{EE}} \in T_E$ is open, the unit square $\XX_{EE}$ (occupied by $T_E$ in~\autoref{fig:NEWdegree4a-2}) does not belong to the inductive region for $EE$.  
\item $\N_{WJ}$ provides type-1 entry and exit connections, since 
$\xrightarrow{e_{WJ}} \in L_W$ is closed (by our assumption that $WW$ exists) and $\xleftarrow{x_{WJ}} \in R_W$ is also closed.
\item Since $\xrightarrow{e_{WW}} \in F_W$ is open, the unit square $\XE_{WW}$ (occupied by $F_W$ in~\autoref{fig:NEWdegree4a-2}) does not belong to the inductive region for $WW$. 
\item Since $\xleftarrow{x_{WW}} \in K_W$ is closed, 
$\N_{WW}$ provides a type-1 exit connection.
\end{itemize}
As in the previous case, we split the unfolding of $E$ into two subnets, $\N_{EJ}$ and $\N_{EE}$, so as to avoid sharing part of $A$'s exit ring face with the type-2 entry connection that $\N_E$ would have provided (had it not been split). 
%
\begin{figure}[h]
\centering
\includegraphics[page=3,width=.85\textwidth]{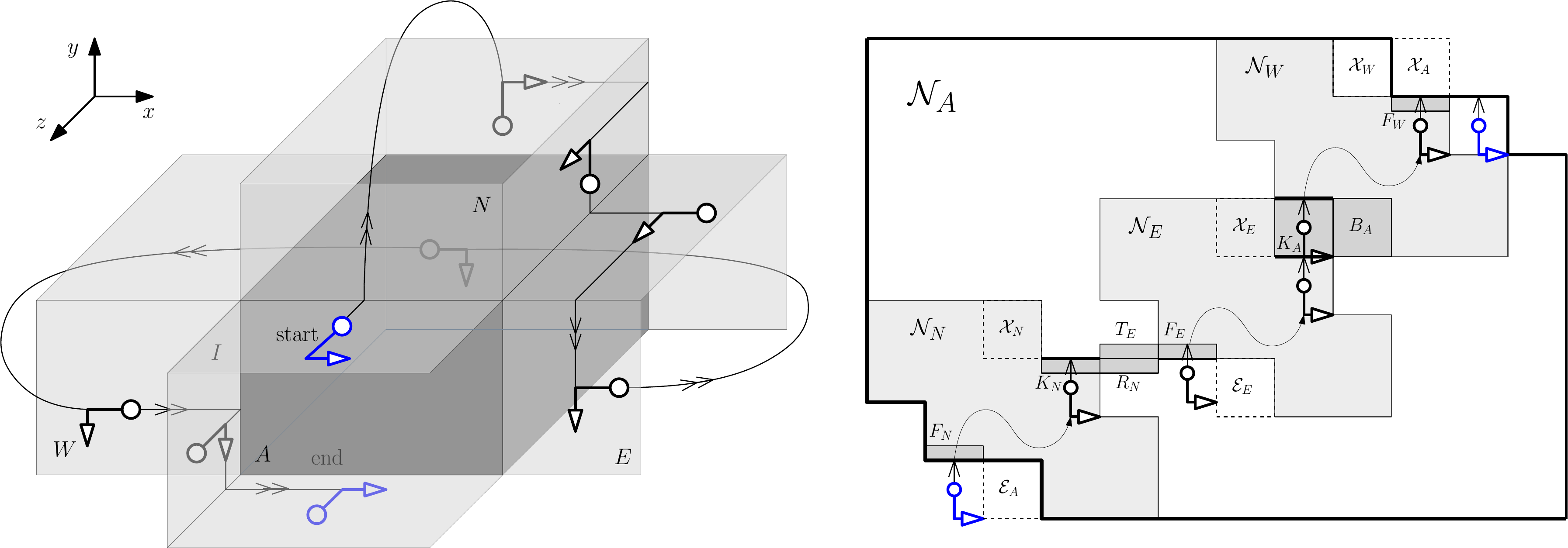}
\caption{Unfolding of degree-4 box $A$ with $N$, $E$ and $W$ children, case $K_N$ open and $B_I$ open.}
\label{fig:NEWdegree4a-3}
\end{figure}
%

\medskip
\noindent
{\bf Case 3: $K_N$ open and $B_I$ open.}  The unfolding for this case is depicted in~\autoref{fig:NEWdegree4a-3}. Note that this unfolding follows a path similar to the one depicted in~\autoref{fig:NEdegree3a}a up to the point where it reaches $K_A$, where it deviates and proceeds with the recursive unfolding of $W$. Note that the entry and exit ring faces for $W$ are $e_W \in K_A$ and $x_W \in L_I$. 

Observe that $\xleftarrow{x_{W}} \in T_I$ is open but not adjacent to $\T_W$, therefore $\N_W$ provides a type-1 exit connection $x'_W \in F_W$, which is a type-2 exit connection for $\N_A$. This along with the fact that $\xleftarrow{x_{A}} \in L_I$ is adjacent to $\T_A$, shows that the exit port of $\N_A$ satisfies invariant (I2).
Arguments similar to the ones used in Case 1 above show that the entry port of $\N_A$ also satisfies (I2), and that $\N_A$ satisfies (I3) as well. 

Turning to (I1), notice that $\xrightarrow{e_{W}} \in B_A$ is open, therefore the unit square $\XE_W$ (occupied by $B_A$ in~\autoref{fig:NEWdegree4a-3}) does not belong to the inductive region got $W$. Furthermore, since $\xrightarrow{e_{W}}$ is 
adjacent to $\T_W$, $\N_W$ may provide a type-1 or type-2 entry connection, which attaches to $K_A$ or $B_A$ placed alongside its entry port and entry port extension, respectively. This, along with the arguments used in the proof of~\autoref{lem:NEdegree3} showing that $\N_N$ and $\N_E$ connect together, shows that $\N_A$ is connected and satisfies invariant (I1).
\end{proof}

\begin{lemma}
\label{lem:NEJdegree4}
Let $A \in \T$ be a degree-$4$ node with parent $I$ and children $N$, $E$, and $J$  \emph{(Case 4.3)}. 
If $A$'s children satisfy invariants (I1)-(I3), then $A$ satisfies invariants (I1)-(I3).
\end{lemma}
\begin{proof}
%
\begin{figure}[h]
\centering
\includegraphics[page=1,width=\textwidth]{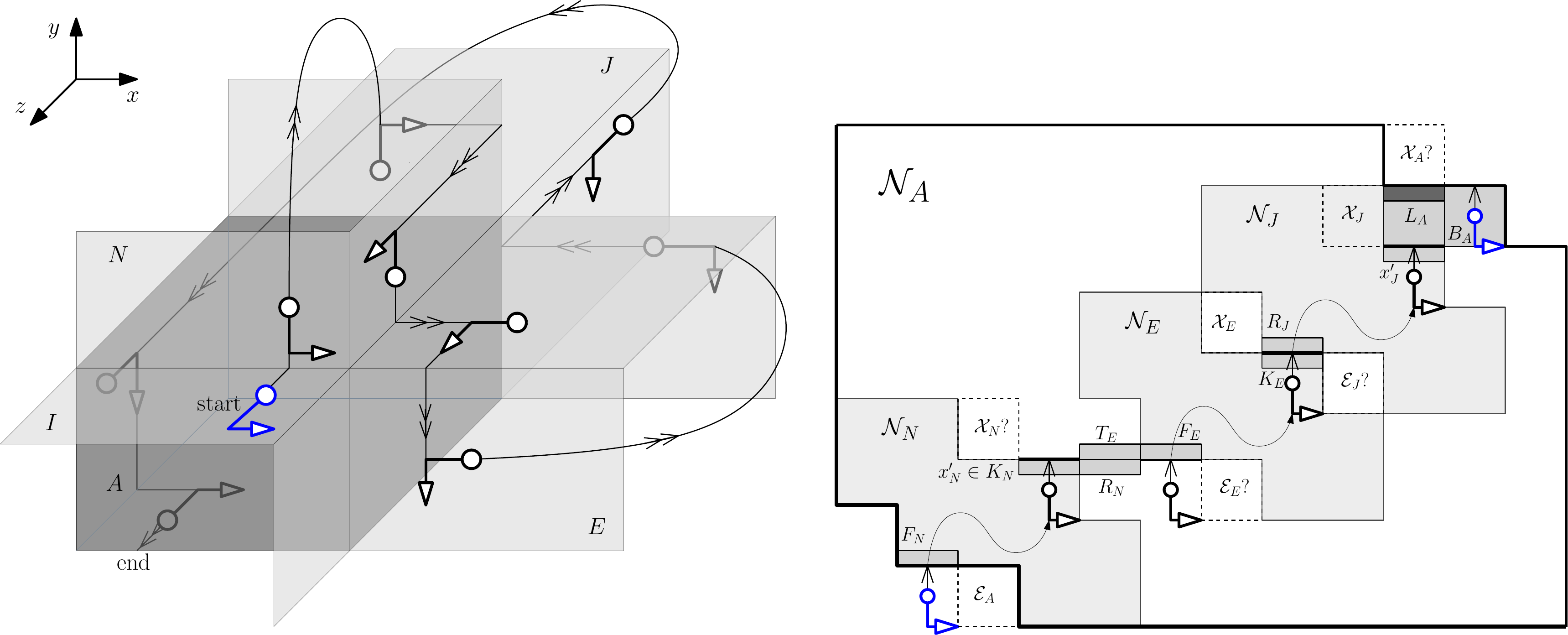}
\caption{\hand-east unfolding of degree-4 box $A$ with $N$, $E$ and $J$ children.}
\label{fig:NEJdegree4a}
\end{figure}
%
Consider the \hand-east unfolding depicted in~\autoref{fig:NEJdegree4a}, and notice that this unfolding is a 
generalization of the degree-3 unfolding from~\autoref{fig:NEdegree3a}a, where the unfolded face $K_A$ is replaced by the recursive unfolding of child $J$. Arguments similar to the ones used in the proof of~\autoref{lem:NEdegree3} show that the unfolding $\N_A$ from~\autoref{fig:NEJdegree4a} satisfies invariants (I2) and (I3). 
The following observations support our claim that $\N_A$ is connected and satisfies invariant (I1):
\begin{itemize}
\item $\N_N$ and $\N_E$ are connected (by the proof of~\autoref{lem:NEdegree3}).
\item $\N_E$ provides a type-1 exit connection, since $\xleftarrow{x_E} \in T_J$ is not adjacent to $\T_E$. This connection attaches to the type-1 entry connection provided by $\N_J$ (since $\xrightarrow{e_J} \in B_E$ is not adjacent to $\T_J$). 
Also, $\N_J$ provides type-1 exit connection (since $\xleftarrow{x_J} \in T_A$ is closed), which attaches to the exit ring face $x_J \in L_A$ placed alongside its exit port. 
\end{itemize}
This concludes the proof.
\end{proof}

\begin{lemma}
\label{lem:NWJdegree4}
Let $A \in \T$ be a degree-$4$ node with parent $I$ and children $N$, $W$, and $J$  \emph{(Case 4.4)}. 
If $A$'s children satisfy invariants (I1)-(I3), then $A$ satisfies invariants (I1)-(I3).
\end{lemma}
\begin{proof}%
The unfolding for this case is slightly more complex and involves four exhaustive scenarios:
\begin{enumerate}
\item $R_J$ open
\item $R_J$ closed and $R_I$ closed
\item $R_J$ closed, $R_I$ open and $B_J$ open
\item $R_J$ closed, $R_I$ open and $B_J$ closed
\end{enumerate}
%

\begin{figure}[h]
\centering
\includegraphics[page=1,width=\textwidth]{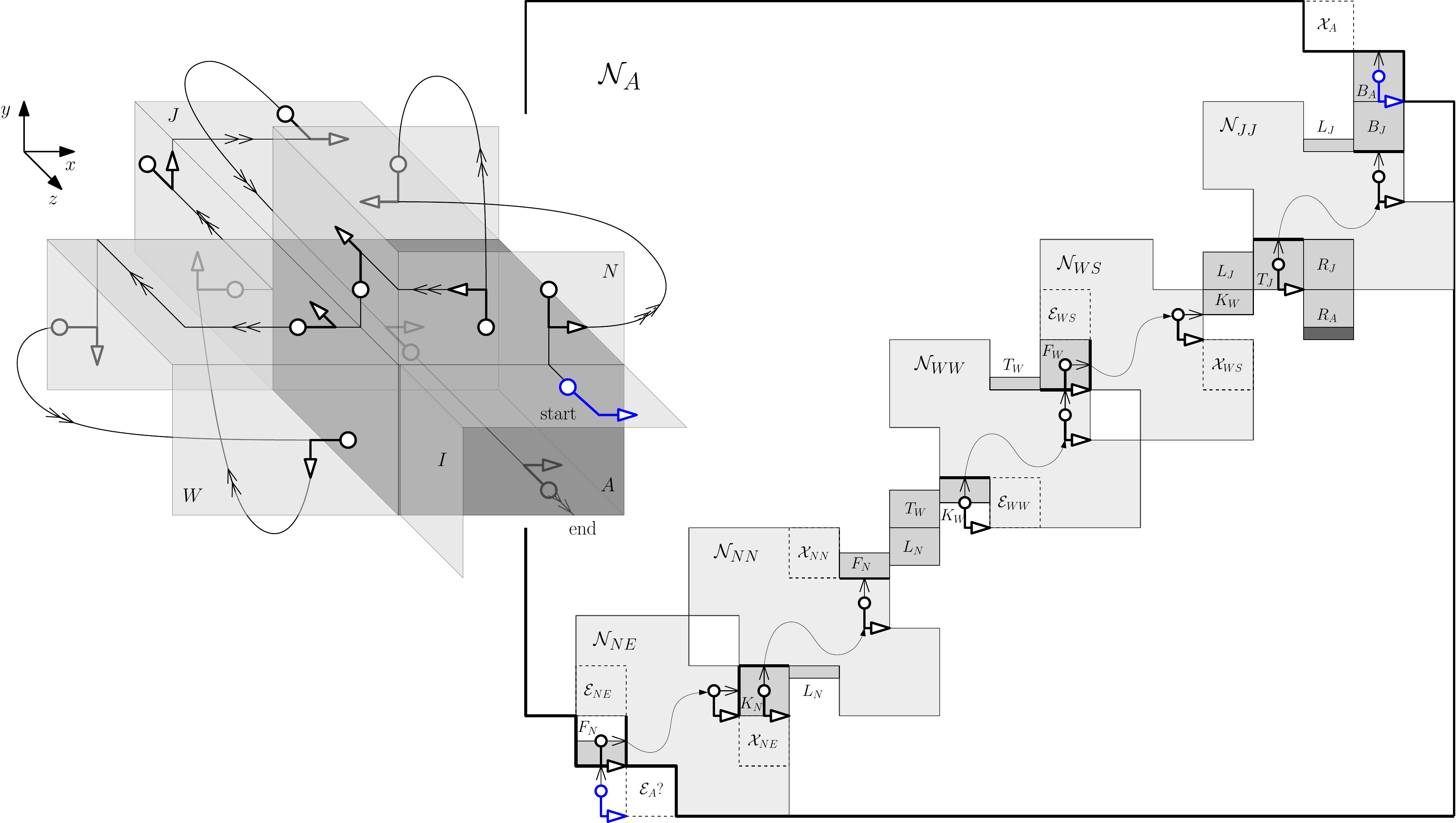}
\caption{Unfolding for box $A$ of degree 4 with $N$, $W$ and $J$ children, case $R_J$ open; unfolding shown for general case when $NE$, $NN$, $WW$ and $WS$ exist.}
\label{fig:NEJdegree4b-1}
\end{figure}
%
\medskip
\noindent
{\bf Case 1: $R_J$ open.} An unfolding for this case is depicted in~\autoref{fig:NEJdegree4b-1}, which handles the more general case where $NE$, $NN$, $WW$ and $WS$ exist (handling cases where  one or more of these boxes do not exist requires only minor modifications). First note that $\N_A$ provides a type-1 entry connection 
$e'_A \in F_N$ and a type-1 exit connection $x'_A \in B_A$, therefore it satisfies invariant (I2). The following observations support our claim that $\N_A$ satisfies invariant (I1):
\begin{itemize}
\item The entry and exit ring faces for $NE$, $NN$, $WW$, $WS$ and $JJ$ are as follows: 
$e_{NE} \in F_N$ and $x_{NE} \in K_N$; 
$e_{NN} \in K_N$ and $x_{NN} \in F_N$; 
$e_{WW} \in K_W$ and $x_{WW} \in F_W$; 
$e_{WS} \in F_W$ and $x_{WS} \in K_W$; 
and $e_{JJ} \in T_J$ and $x_{JJ} \in B_J$ (which is open, since we assume that $B_W$ is closed). 
\item $\N_{NE}$, $\N_{WW}$ and $\N_{WS}$ provide type-1 entry connections. This is because 
$\xrightarrow{e_{NE}} \in T_N$ is closed (since $NN$ exists), $\xrightarrow{e_{WW}} \in B_W$ is closed, and  
$\xrightarrow{e_{WS}} \in R_W$ is closed. 
\item Since $\xrightarrow{e_{NN}} \in L_N$ is open, the unit square $\XE_{NN}$ (occupied by $\xrightarrow{e_{NN}}$ in~\autoref{fig:NEJdegree4b-1}) does not belong to the inductive region for $NN$.
\item Since $\xrightarrow{e_{JJ}} \in R_J$ is open, the unit square $\XE_{JJ}$ (occupied by $R_J$ in~\autoref{fig:NEJdegree4b-1}) does not belong to the inductive region for $JJ$. Notice that we place $R_A$ right underneath it. 
\item $\N_{NE}$, $\N_{NN}$ and $\N_{WS}$ provide type-1 exit connections. This is because $\xleftarrow{x_{NE}} \in B_N$ is closed, $\xleftarrow{x_{NN}} \in R_N$ is closed (since $NE$ exists), and $\xleftarrow{x_{WS}} \in L_W$ is  closed (since $WW$ exists). 
\item Since $\xleftarrow{x_{WW}} \in T_W$ is open, the unit square $\XX_{WW}$ (occupied by $\xleftarrow{x_{WW}}$ in~\autoref{fig:NEJdegree4b-1}) does not belong to the inductive region for $WW$.
\item Since $\xleftarrow{x_{JJ}} \in L_J$ is open, the unit square $\XX_{JJ}$ (occupied by $\xleftarrow{x_{JJ}}$ in~\autoref{fig:NEJdegree4b-1}) does not belong to the inductive region for $JJ$.
\end{itemize}
Turning to invariant (I3), observe that the ring face $\xrightarrow{x'_A} \in R_A$ (dark-shaded in~\autoref{fig:NEJdegree4b-1}) can be removed from $\N_A$ without disconnecting $\N_A$, so (I3) is met. 

\begin{figure}[htpb]
\centering
\includegraphics[page=2,width=\textwidth]{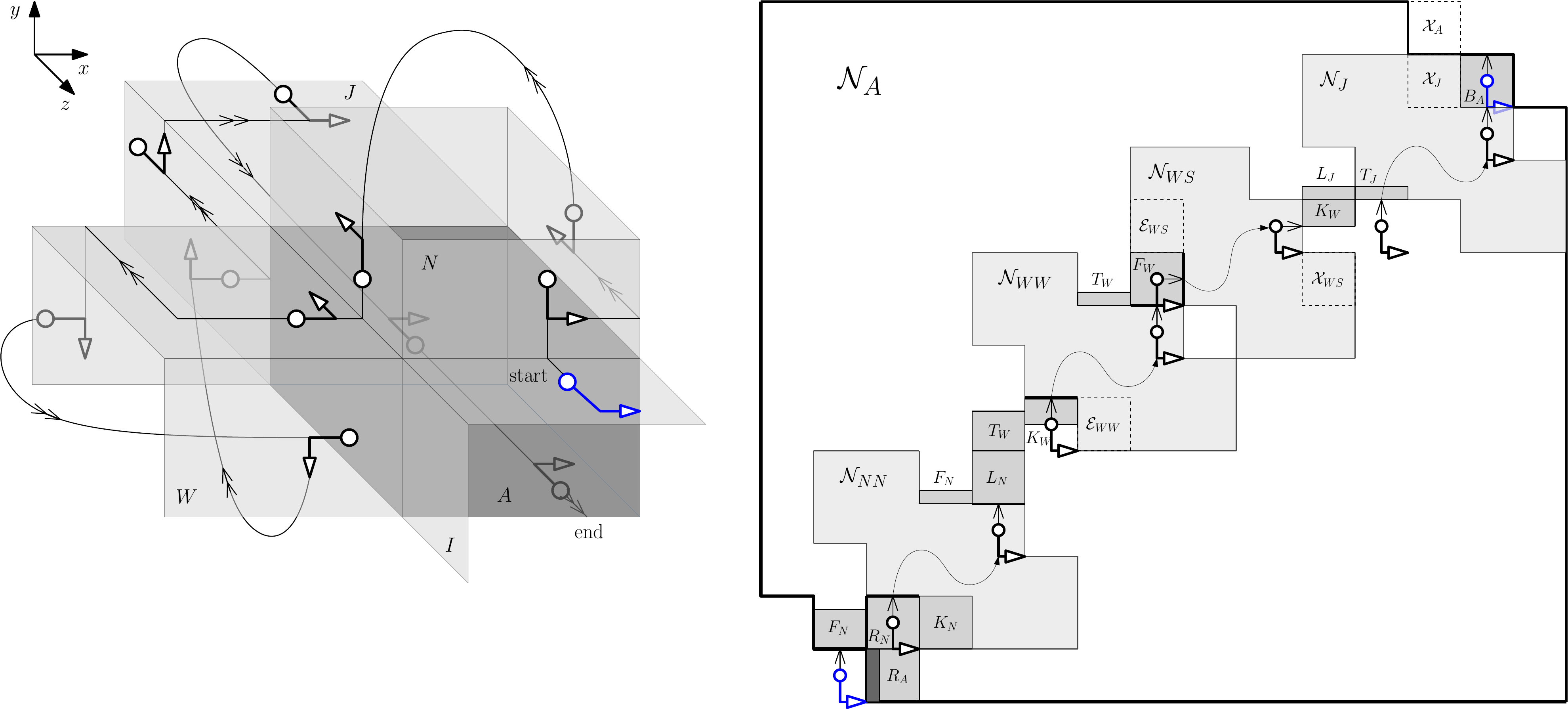}
\caption{Unfolding of degree-4 box $A$ with $N$, $W$ and $J$ children (case $R_J$, $R_I$ closed).}
\label{fig:NEJdegree4b-2}
\end{figure}

\medskip
\noindent
{\bf Case 2: $R_J$ and $R_I$ closed.} An unfolding for this case is depicted in~\autoref{fig:NEJdegree4b-2}. 
Note that this unfolding is very similar to the one shown in~\autoref{fig:NEJdegree4b-1}, with only a few minor modifications: 
\begin{itemize}
\item $R_N$ is open, so $\N_{NE}$ reduces to a single face $R_N$. Since $R_I$ is closed, $\XE_A$ belongs to the inductive region for $A$, therefore we can place $R_A$ underneath $R_N$. 
\item From $R_N$ we proceed directly to recursively unfold $NN$, and in this case $e_{NN} \in R_N$ and $x_{NN} \in L_N$. Since $\xrightarrow{e_{NN}} \in K_N$ is open, the unit square $\XE_{NN}$ (occupied by $K_N$ in~\autoref{fig:NEJdegree4b-2}) does not belong to the inductive region for $NN$. Similarly, since $\xleftarrow{x_{NN}} \in F_N$ is open, 
the unit square $\XX_{NN}$ (occupied by $\xleftarrow{x_{NN}}$ in~\autoref{fig:NEJdegree4b-2}) does not belong to the inductive region for $NN$.
\item The entry and exit ring faces for $J$ are $e_J \in K_N$ and $x_J \in B_A$. Since $\xrightarrow{e_J} \in R_N$ is not adjacent to $\T_J$, and since $\xleftarrow{x_J} \in L_A$ is closed, $\N_J$ provides type-1 entry and exit connections. By invarient (I3), ring face $L_J$ can be removed from $\N_J$ and used as bridge to connect to $K_W$. 
\end{itemize}

%
\begin{figure}[htpb]
\centering
\includegraphics[page=3,width=0.9\textwidth]{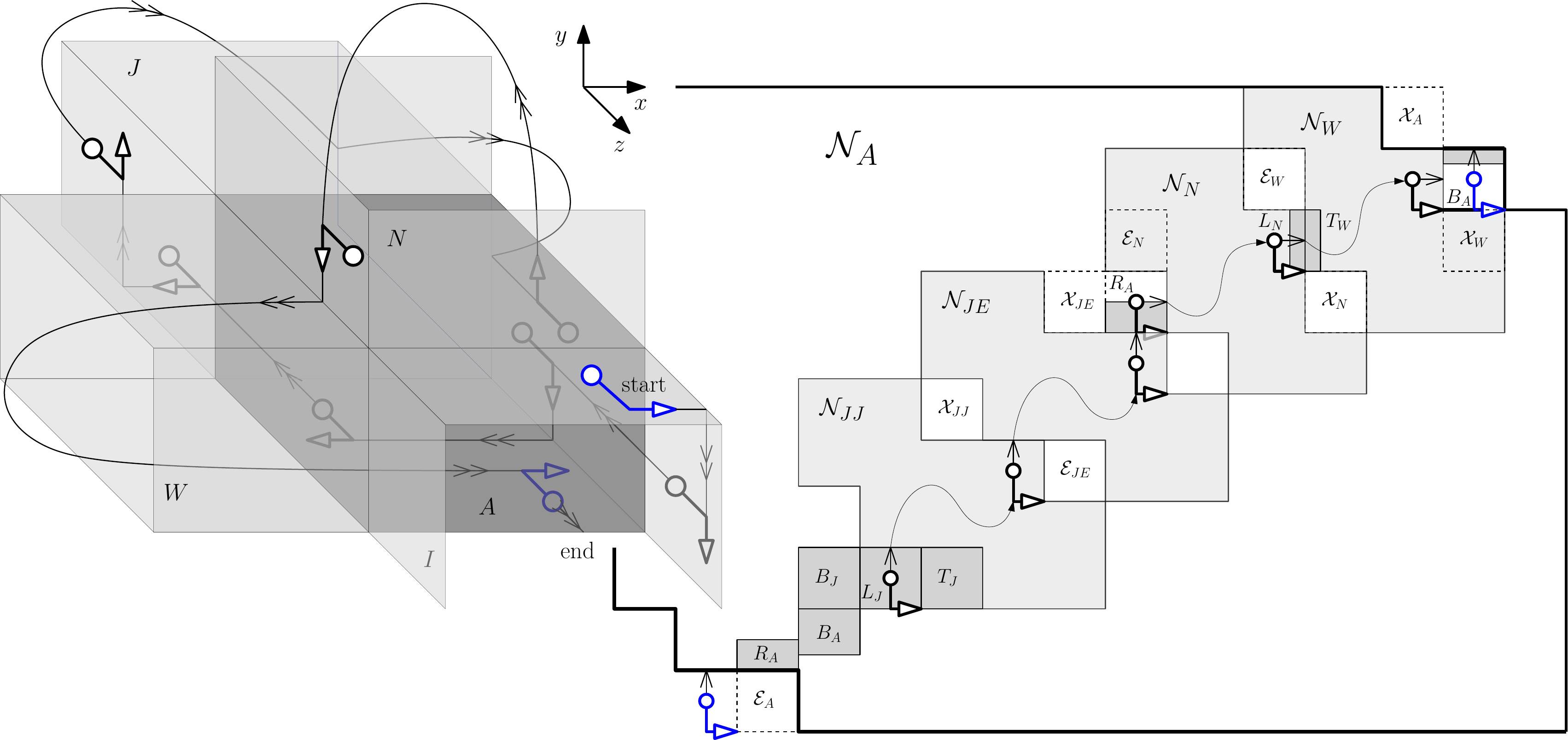}
\caption{Unfolding of degree-4 box $A$ with $N$, $W$ and $J$ children (case $R_J$ closed, $R_I$ and $B_J$ open).}
\label{fig:NEJdegree4b-3}
\end{figure}
%
\medskip
\noindent
{\bf Case 3: $R_J$ closed, $R_I$ and $B_J$ open.}  An unfolding for this case is depicted in~\autoref{fig:NEJdegree4b-3}, which handles the case where $JJ$ exists (handling the case where $JJ$ does not exist requires only minor modifications). 
First note that $\xrightarrow{e_A} \in R_I$ is open and adjacent to $\T_A$, and $\N_A$ provides a type-2 entry connection $\xrightarrow{e'_A} \in R_A$. Also note that $\xleftarrow{x'_A} \in L_A$ is closed and $\N_A$ provides a type-1 exit connection $x'_A \in B_A$. These together show that $\N_A$ satisfies invariant (I2). Since all open ring faces of $A$ are used in entry and exit connections, $\N_A$ trivially satisfies invariant (I3). 
The following observations support our claim that $\N_A$ satisfies invariant (I1):
\begin{itemize}
\item The entry and exit ring faces for $JJ$, $JE$, $N$ and $W$ are as follows: 
$e_{JJ} \in L_J$ and $x_{JJ} \in K_{JE}$; 
$e_{JE} \in R_{JJ}$ and $x_{JE} \in R_A$; 
$e_N \in R_A$ and $x_N \in T_W$; 
and $e_W \in L_N$ (which connects to $T_W$) and $x_W \in B_A$. 
\item Since $\xrightarrow{e_{JJ}} \in T_J$ is open, the unit square $\XE_{JJ}$ (occupied by $T_J$ in~\autoref{fig:NEJdegree4b-3}) does not belong to the inductive region for $JJ$.
\item $\N_{JE}$, $\N_N$ and $\N_W$ provide type-1 entry connections. This is because 
$\xrightarrow{e_{JE}} \in T_{JJ}$ is not adjacent to $\T_{JE}$, 
$\xrightarrow{e_{N}} \in F_A$ is closed, and $\xrightarrow{e_W} \in F_N$ is not adjacent to $\T_W$. 
\item $\N_{JJ}$, $\N_{JE}$, $\N_N$ and $\N_W$ provide type-1 exit connections. This is because 
$\xleftarrow{x_{JJ}} \in B_{JE}$ is not adjacent to $\T_{JJ}$, 
$\xleftarrow{x_{JE}} \in B_A$ is not adjacent to $\T_{JE}$, 
$\xleftarrow{x_N} \in K_W$ is not adjacent to $\T_N$, and $\xleftarrow{x_W} \in K_A$ is closed. Note that 
the type-1 exit connection of $\N_{JJ}$ attaches to the type-1 entry connection of $\N_{JE}$, and 
the type-1 exit connection of $\N_N$ attaches to the type-1 entry connection of $\N_W$. 
\end{itemize}

Note that we split the unfolding of $J$ into two subnets ($\N_{JJ}$ and $\N_{JE}$) so as to avoid sharing the ring face $\xleftarrow{x'_J} \in B_J$ between its current position in $\N_A$ (where it serves as a bridge between $B_A$ and $L_J$) and the type-2 exit connection that $\N_J$ would have provided (had it not been split).

%
\begin{figure}[htpb]
\centering
\includegraphics[page=4,width=0.87\textwidth]{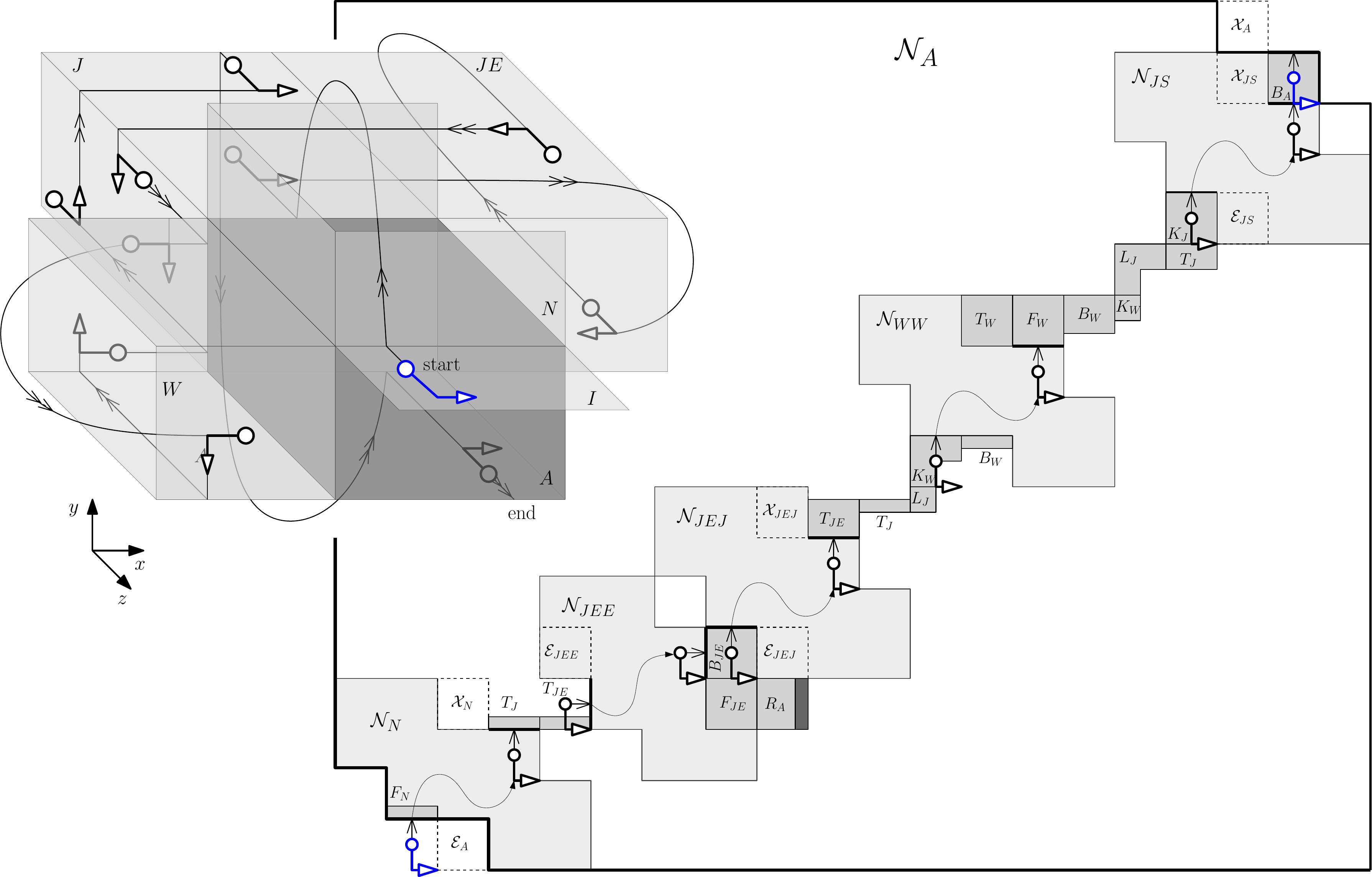}
\caption{Unfolding of degree-4 box $A$ with $N$, $W$ and $J$ children (case $R_J$ closed, $R_I$ open and $B_{J}$ closed -- so $B_{JE}$ and $B_W$ open); unfolding shown for the case when $JEJ$ and $JEE$ exist (so $K_J$ open).}
\label{fig:NEJdegree4b-4}
\end{figure}
%

\medskip
\noindent
{\bf Case 4: $R_J$ closed, $R_I$ open and $B_J$ closed.}  An unfolding for this case is depicted in~\autoref{fig:NEJdegree4b-4}, which handles the more general case when $JEJ$ and $JEE$ exist. 
Arguments similar to the ones used in the proof of Case 1 of~\autoref{lem:NEWdegree4} show that the unfolding $\N_A$ from~\autoref{fig:NEJdegree4b-4} satisfies invariants (I2) and (I3). The following observations support our claim that $\N_A$ satisfies invariant (I1):
\begin{itemize}
\item The entry and exit ring faces for $N$, $JEE$, $JEJ$, $WW$ and $JS$ are as follows: 
$e_N \in T_I$ and $x_N \in T_J$; 
$e_{JEE} \in T_{JE}$ and $x_{JEE} \in B_{JE}$ (which is open, since $B_J$ is closed); 
$e_{JEJ} \in B_{JE}$ and $x_{JEJ} \in T_{JE}$; 
$e_{WW} \in K_W$ and $x_{WW} \in F_W$; 
and $e_{JS} \in K_J$ (which is open, by our assumption that $K_{JE}$ is closed) and $x_{JS} \in B_A$.
\item $\N_N$, $\N_{JEJ}$ and $\N_{JS}$ provide type-1 exit connections. This is because 
$\xleftarrow{x_N} \in L_J$ is not adjacent to $\T_N$,
$\xleftarrow{x_{JEJ}} \in R_{JE}$ is closed (since $JEE$ exists),
and $\xleftarrow{x_{JS}} \in L_A$ is closed.
\item Since $\xleftarrow{x_{JEE}} \in F_{JE}$ is open, the unit square $\XX_{JEE}$ (occupied by $F_{JE}$ in~\autoref{fig:NEJdegree4b-4}) does not belong to the inductive region for $JEE$.
\item Since $\xrightarrow{e_{WW}} \in B_W$ and $\xleftarrow{x_{WW}} \in T_W$ are open, the unit squares $\XE_{WW}$ and $\XX_{WW}$ (occupied in~\autoref{fig:NEJdegree4b-4} by $\xrightarrow{e_{WW}}$ and $T_W$, respectively) do not belong to the inductive region for $WW$.
\item $\N_{JEE}$, $\N_{JEJ}$ and $\N_{JS}$ provide type-1 entry connections. This is because
$\xrightarrow{e_{JEE}} \in K_{JE}$ is closed (since $JEJ$ exists), 
$\xrightarrow{e_{JEJ}} \in L_{JE}$ is closed, and 
$\xrightarrow{e_{JS}} \in R_{J}$ is closed. (Observe that the existence of $JEJ$ implies that $K_J$ is open.)
\end{itemize}
\end{proof}

\begin{lemma}
\label{lem:NESdegree4}
Let $A \in \T$ be a degree-$4$ node with parent $I$ and children $N$, $E$, and $S$  \emph{(Case 4.5)}. 
If $A$'s children satisfy invariants (I1)-(I3), then $A$ satisfies invariants (I1)-(I3).
\end{lemma}
\begin{proof}
The unfolding for this case involves three different case scenarios:
\begin{enumerate}
\item $K_E$ open
\item $K_E$ closed and $L_I$ closed
\item $K_E$ closed and $L_I$ open 
\end{enumerate}
%
%
\begin{figure}[htpb]
\centering
\includegraphics[page=1,width=.9\textwidth]{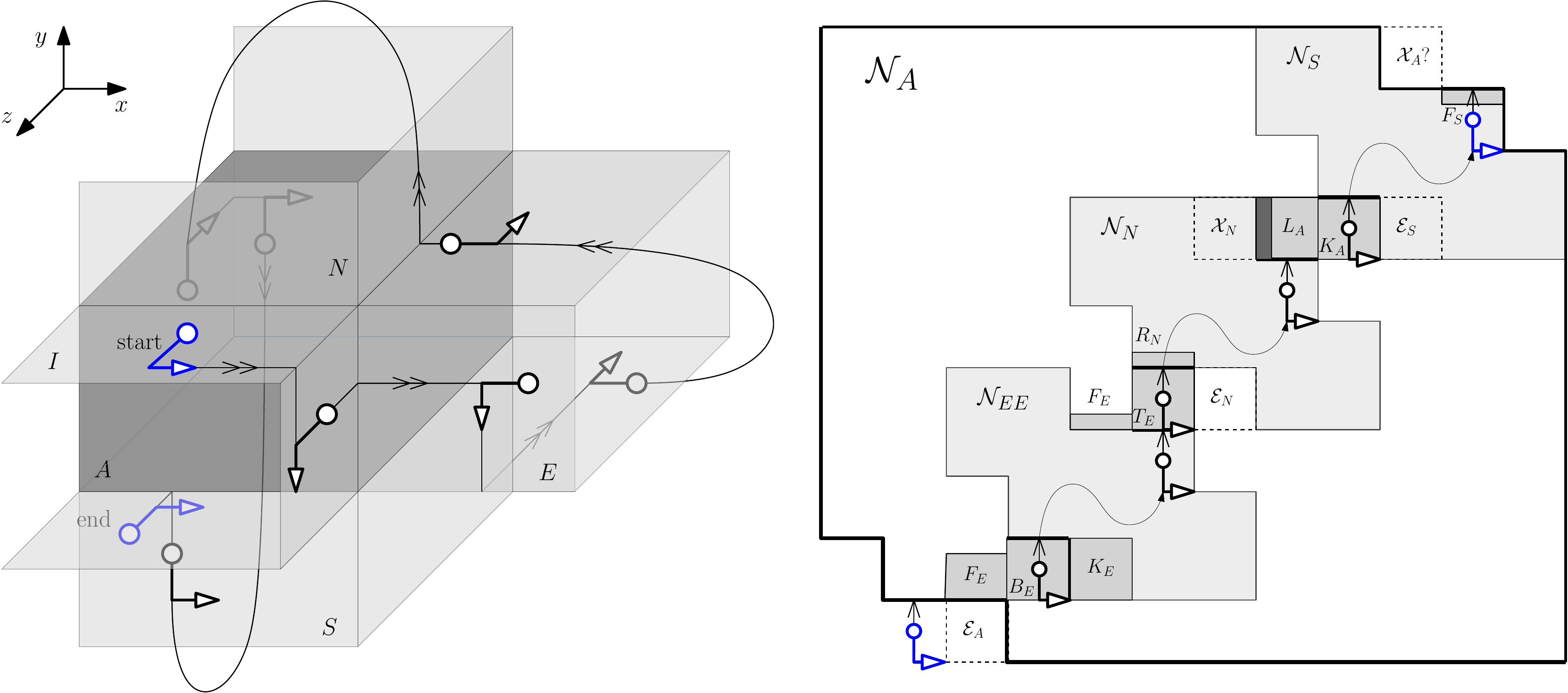}
\caption{\hand-east unfolding of degree-4 box $A$ with $N$, $E$ and $S$ children (case $K_E$ open).}
\label{fig:NESdegree4a-1}
\end{figure}
%
\medskip
\noindent
{\bf Case 1: $K_E$ open.} Note that in this case $E$ is either a leaf or a connector. 
Consider the \hand-east unfolding depicted in~\autoref{fig:NESdegree4a-1}. Note that $\xrightarrow{e_A} \in R_I$ is open and adjacent to $\T_A$, and $\N_A$ provides a type-2 entry connection $\xrightarrow{e'_A} \in F_E$. Also, since $\xleftarrow{x_S} \in L_I$ is not adjacent to $\T_S$, invariant (I2) applied to $S$ tells us that $\N_S$ provides a type-1 exit connection, which is also a type-1 exit connection for $A$ (since $e'_A = e'_S \in F_S$). These together show that $\N_A$ satisfies invariant (I2). The following observations support our claim that $\N_A$ satisfies invariant (I1):
\begin{itemize}
%
\item The entry and exit ring faces for $EE$, $N$ and $S$ are as follows: 
$e_{EE} \in B_E$ and $x_{EE} \in T_E$; 
$e_N \in T_E$ and $x_N \in L_A$; 
and $e_S \in K_A$ and $x_S \in B_I$.
\item Since $\xrightarrow{e_{EE}} \in K_E$ and $\xleftarrow{x_{EE}} \in F_E$ are open, the unit squares $\XE_{EE}$ and $\XX_{EE}$  (occupied in~\autoref{fig:NESdegree4a-1} by $K_E$ and $\xleftarrow{x_{EE}}$, respectively) do not belong to the inductive region for $EE$.
\item $\N_N$ provides type-1 entry and exit connections. This is because 
$\xrightarrow{e_N} \in K_E$ is not adjacent to $\T_N$ (note however that it is open, so $\XE_N$ does not belong to the $N$'s inductive region), 
and $\xleftarrow{x_N} \in F_A$ is closed. Also  
$\N_S$ provides a type-1 entry connection, since
$\xrightarrow{e_{S}} \in R_A$ is closed.
\end{itemize}
Finally, observe that the ring face $\xrightarrow{x'_A} \in L_A$ (dark-shaded in~\autoref{fig:NESdegree4a-1}) can be removed from $\N_A$ without disconnecting $\N_A$, so $\N_A$ satisfies invariant (I3).

%
\begin{figure}[htpb]
\centering
\includegraphics[page=2,width=.9\textwidth]{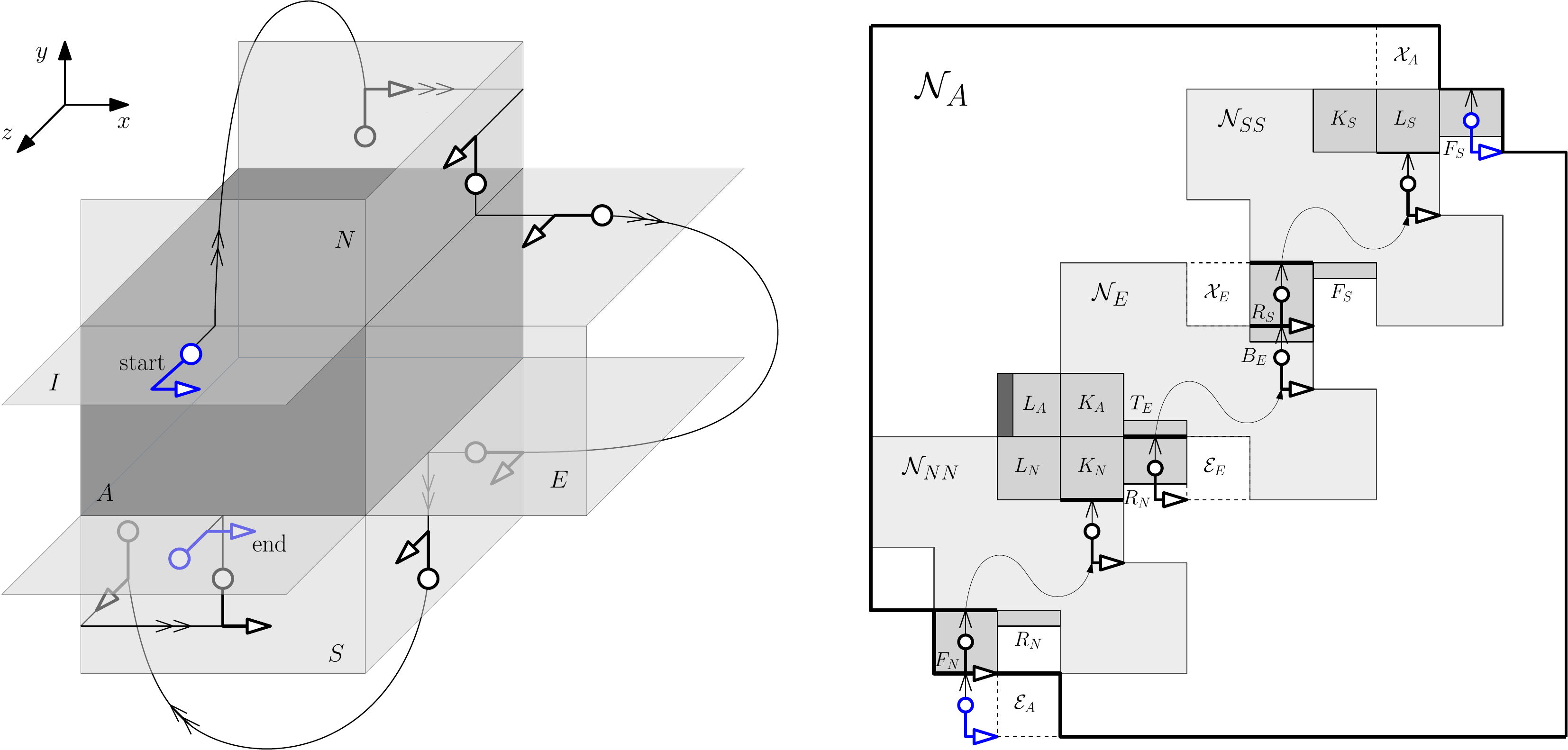}
\caption{\hand-east unfolding of degree-4 box $A$ with $N$, $E$ and $S$ children, case $K_E$ closed (so $K_N$, $K_S$ open) and $L_I$ closed (so $L_N$, $L_S$ open).}
\label{fig:NESdegree4a-2}
\end{figure}
%
\medskip
\noindent
{\bf Case 2: $K_E$ closed and $L_I$ closed.} Note that in this case $K_N$ and $K_S$ are open
(since $K_E$ is closed) and $L_N$ and $L_S$ are also open (since $L_I$ is closed). 
Consider the \hand-east unfolding from~\autoref{fig:NESdegree4a-2}, and note that 
$\N_A$ provides a type-1 entry connection $e'_A \in F_N$ and a type-1 exit connection $x'_A \in F_S$. 
Thus $\N_A$ satisfies invariant (I2). 
The following observations support our claim that $\N_A$ satisfies invariant (I1):
\begin{itemize}
%
\item The entry and exit ring faces for $NN$, $E$ and $SS$ are as follows: 
$e_{NN} \in F_N$ and $x_{NN} \in K_N$; 
$e_E \in R_N$ and $x_E \in R_S$; 
and $e_{SS} \in R_S$ and $x_{SS} \in L_S$. 
\item Since $\xrightarrow{e_{NN}} \in R_N$ and $\xleftarrow{x_{NN}} \in L_N$ are open, 
the unit squares $\XE_{NN}$ and $\XX_{NN}$ (occupied in~\autoref{fig:NESdegree4a-2} by 
$\xrightarrow{e_{NN}}$ and $L_N$, respectively) do not belong to the inductive region for $NN$.
\item Similarly, since $\xrightarrow{e_{SS}} \in F_S$ and $\xleftarrow{x_{SS}} \in K_S$ are open, 
the unit squares $\XE_{SS}$ and $\XX_{SS}$ (occupied in~\autoref{fig:NESdegree4a-2} by 
$\xrightarrow{e_{SS}}$ and $K_S$, respectively) do not belong to the inductive region for $SS$.
\item $\N_E$ provides type-1 entry and exit connections, 
since $\xrightarrow{e_{E}} \in F_N$ and $\xleftarrow{x_{E}} \in K_S$ are not adjacent to $\T_E$. 
\end{itemize}
Arguments similar to the ones above show that $\N_A$ satisfies invariant (I3).

\medskip
\noindent
{\bf Case 3: $K_E$ closed and $L_I$ open.}  
In this case we use $\xrightarrow{e_A} \in R_I$ and $\xleftarrow{x_A} \in L_I$ as entry and exit ring faces for the unfolding case when $A$ has $N$, $E$ and $W$ children and $K_E$ is closed. This approach is depicted in~\autoref{fig:reduce}, with the understanding that $\N'_A$ is the net from~\autoref{fig:NEWdegree4a-1}. Because $\xrightarrow{e_A} \in R_I$ and $\xleftarrow{x_A} \in L_I$ are both open and adjacent to $\T_A$, $A$ may provide type-1 or type-2 entry and exit connections. Note that the unfolding net from~\autoref{fig:NEWdegree4a-1} provides a type-1 entry/exit connection, which is a type-2 entry/exit connection for $\N_A$. Since $\N'_A$ satisfies invariants (I1)-(I3) (by~\autoref{lem:NEWdegree4}), we conclude that $\N_A$ satisfies invariants (I1)-(I3).

\end{proof}

\begin{lemma}
\label{lem:NWSdegree4}
Let $A \in \T$ be a degree-$4$ node with parent $I$ and children $N$, $W$, and $S$  \emph{(Case 4.6)}. 
If $A$'s children satisfy invariants (I1)-(I3), then $A$ satisfies invariants (I1)-(I3).
\end{lemma}
\begin{proof}
%
\begin{figure}[htpb]
\centering
\includegraphics[page=3,width=.9\textwidth]{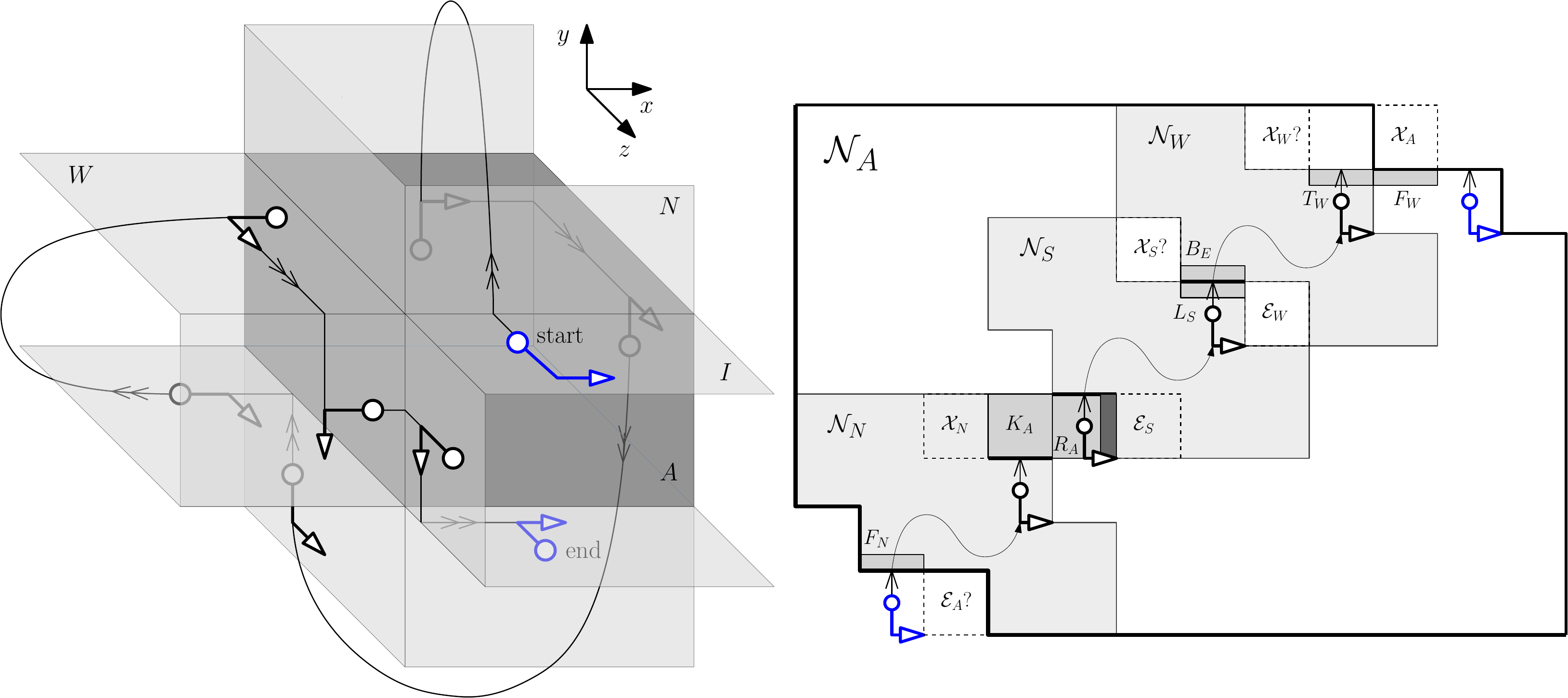}
\caption{Unfolding of degree-4 box $A$ with $N$, $W$ and $S$ children.}
\label{fig:NESdegree4b-1}
\end{figure}
%
Consider the unfolding from~\autoref{fig:NESdegree4b-1}, and notice that this is a general case of the degree-3 unfolding from~\autoref{fig:NEdegree3b}, where the unfolded face $B_A$ is replaced by the recursive unfolding of child $S$. 
Arguments similar to the ones used in the proof of Case 1 of~\autoref{lem:NEdegree3} show that the unfolding $\N_A$ from~\autoref{fig:NESdegree4b-1} satisfies invariants (I2) and (I3). Turning to (I1), note that $\N_N$, 
$\N_S$ and $\N_W$ all provide type-1 entry and exit connections. This is because 
$\xrightarrow{e_N} \in R_I$ is not adjacent to $\T_N$,
$\xleftarrow{x_N} \in L_A$ and $\xrightarrow{e_S} \in F_A$ are closed, 
$\xleftarrow{x_S} \in K_W$ is not adjacent to $\T_S$, and 
$\xrightarrow{e_W} \in F_S$ and $\xleftarrow{x_{W}} \in K_N$ are not adjacent to $\T_W$.
These together show that $\N_A$ satisfies invariant (I1).
\end{proof}

%
\begin{figure}[htpb]
\centering
\includegraphics[page=1,width=.9\textwidth]{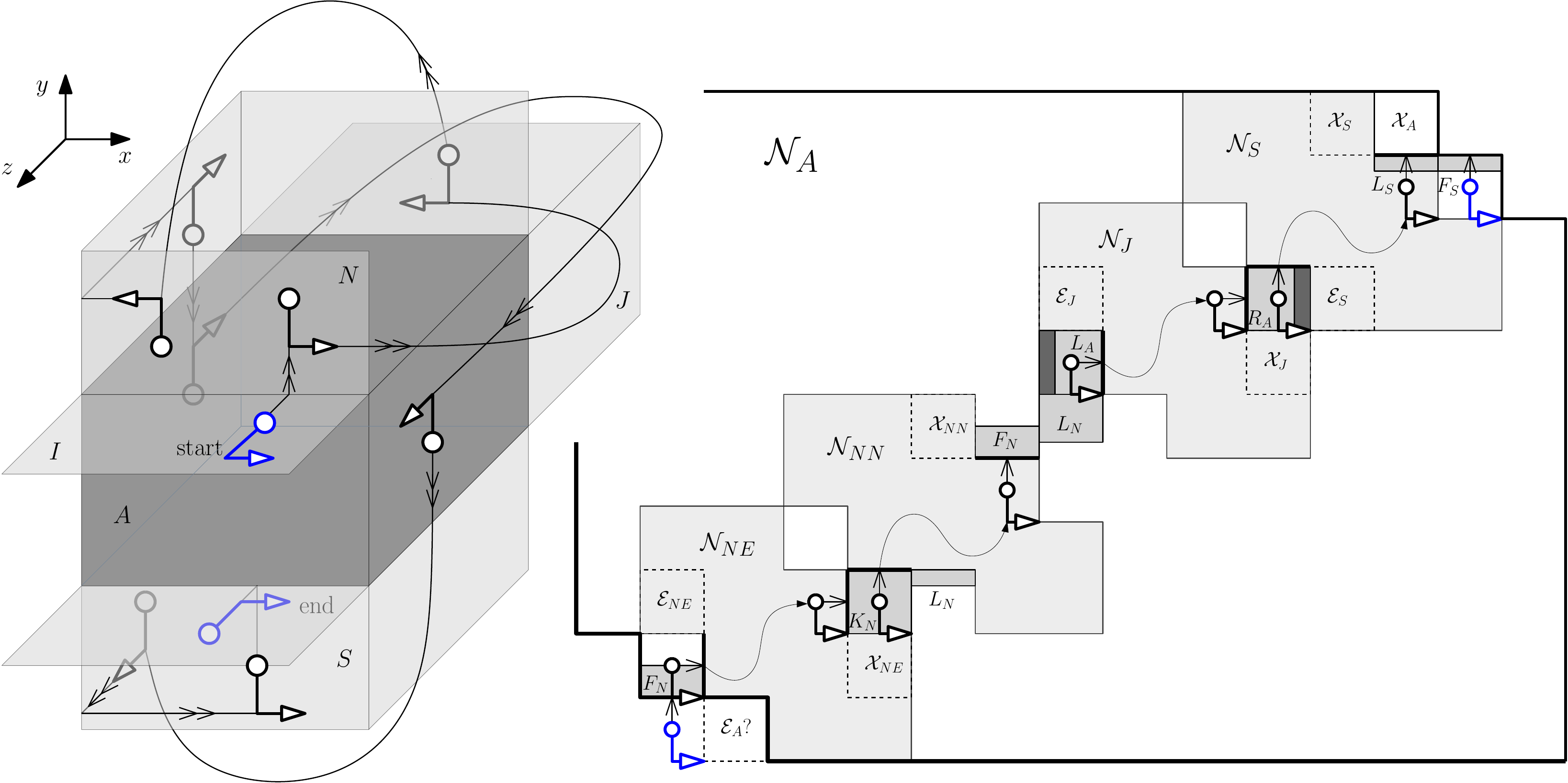}
\caption{Unfolding of degree-$4$ box $A$ with children $N$, $J$ and $S$, case $L_I$ closed (so $L_N$, $L_S$ open); 
unfolding shown for the case when $NN$ and $NE$ exist.}
\label{fig:NJSdegree4a-1}
\end{figure}
%

\begin{lemma}
\label{lem:NJSdegree4}
Let $A \in \T$ be a degree-$4$ node with parent $I$ and children $N$, $J$, and $S$  \emph{(Case 4.7)}. 
If $A$'s children satisfy invariants (I1)-(I3), then $A$ satisfies invariants (I1)-(I3).
\end{lemma}
\begin{proof}
We discuss the following four exhaustive scenarios:
\begin{enumerate}
\item $L_I$ and $R_I$ open
\item $L_I$ closed
\item $R_I$ closed and $L_J$ closed
\item $R_I$ closed and $L_J$ open
\end{enumerate}

\noindent
{\bf Case 1: $L_I$ and $R_I$ open.}  
In this case $I$ is a non-junction and we can use the unfolding from~\autoref{fig:reduce}, where we substitute $\N'_A$ with the net from~\autoref{fig:EJWdegree4}. Because $\xrightarrow{e_A} \in R_I$ and $\xleftarrow{x_A} \in L_I$ are both open and adjacent to $\T_A$, $\N_A$ may provide type-1 or type-2 entry and exit connections. Note that the unfolding net from~\autoref{fig:EJWdegree4} provides a type-1 entry/exit connection, which is a type-2 entry/exit connection for $\N_A$. 
Since $\N'_A$ satisfies invariants (I1)-(I3) (by~\autoref{lem:EJWdegree4}), we conclude that $\N_A$ satisfies invariants (I1)-(I3). 

\medskip
\noindent
{\bf Case 2: $L_I$ closed.}  Note that in this case $L_N$ and $L_S$ are open. Consider the unfolding from~\autoref{fig:NJSdegree4a-1}, which handles the more general case when $NN$ and $NE$ exist (handling cases where one or both are missing requires only minor modifications). Note that $\N_A$ provides a type-1 entry connection 
$e'_A \in F_N$ and a type-1 exit connection $x'_A \in F_S$, therefore $\N_A$ satisfies invariant (I2).
The following observations support our claim that $\N_A$ satisfies invariant (I1):
\begin{itemize}
%
\item The entry and exit ring faces for $NE$, $NN$, $J$ and $S$ are as follows: 
$e_{NE} \in F_N$ and $x_{NE} \in K_N$; 
$e_{NN} \in K_N$ and $x_{NN} \in F_N$; 
$e_J \in L_A$ and $x_J \in R_A$; 
and $e_S \in R_A$ and $x_S \in L_A$.
\item $\N_{NE}$, $N_J$ and $\N_S$ provide type-1 entry and exit connections. This is because 
$\xrightarrow{e_{NE}} \in T_N$ and $\xleftarrow{x_{NE}} \in B_N$ are closed (recall our assumption that $NN$ exists),
$\xrightarrow{e_J} \in B_A$ and $\xleftarrow{x_J} \in T_A$ are closed, and
$\xrightarrow{e_S} \in F_A$ and $\xleftarrow{x_S} \in K_A$ are also closed.
\item Since $\xrightarrow{e_{NN}} \in L_N$ is open, the unit square $\XE_{NN}$ (occupied by $\xrightarrow{e_{NN}}$ in~\autoref{fig:NJSdegree4a-1}) does not belong to the inductive region for $NN$. Since $\xleftarrow{x_{NN}} \in R_N$ is closed, 
$\N_{NN}$ provides a type-1 exit connection.
\end{itemize}
The two open ring faces of $A$ not used in entry and exit connections are on $L_A$ and $R_A$ (dark-shaded in~\autoref{fig:NJSdegree4a-1}), and they can be removed without disconnecting $\N_A$. It follows that $\N_A$ satisfies invariant (I3).
\end{proof}

\medskip
\noindent
{\bf Case 3: $R_I$ and $L_J$ closed.}  Note that in this case $L_N$ and $L_S$ are open, and the unfolding for this case is identical to the one shown in~\autoref{fig:NJSdegree4a-1}.

%
\begin{figure}[htpb]
\centering
\includegraphics[page=2,width=\textwidth]{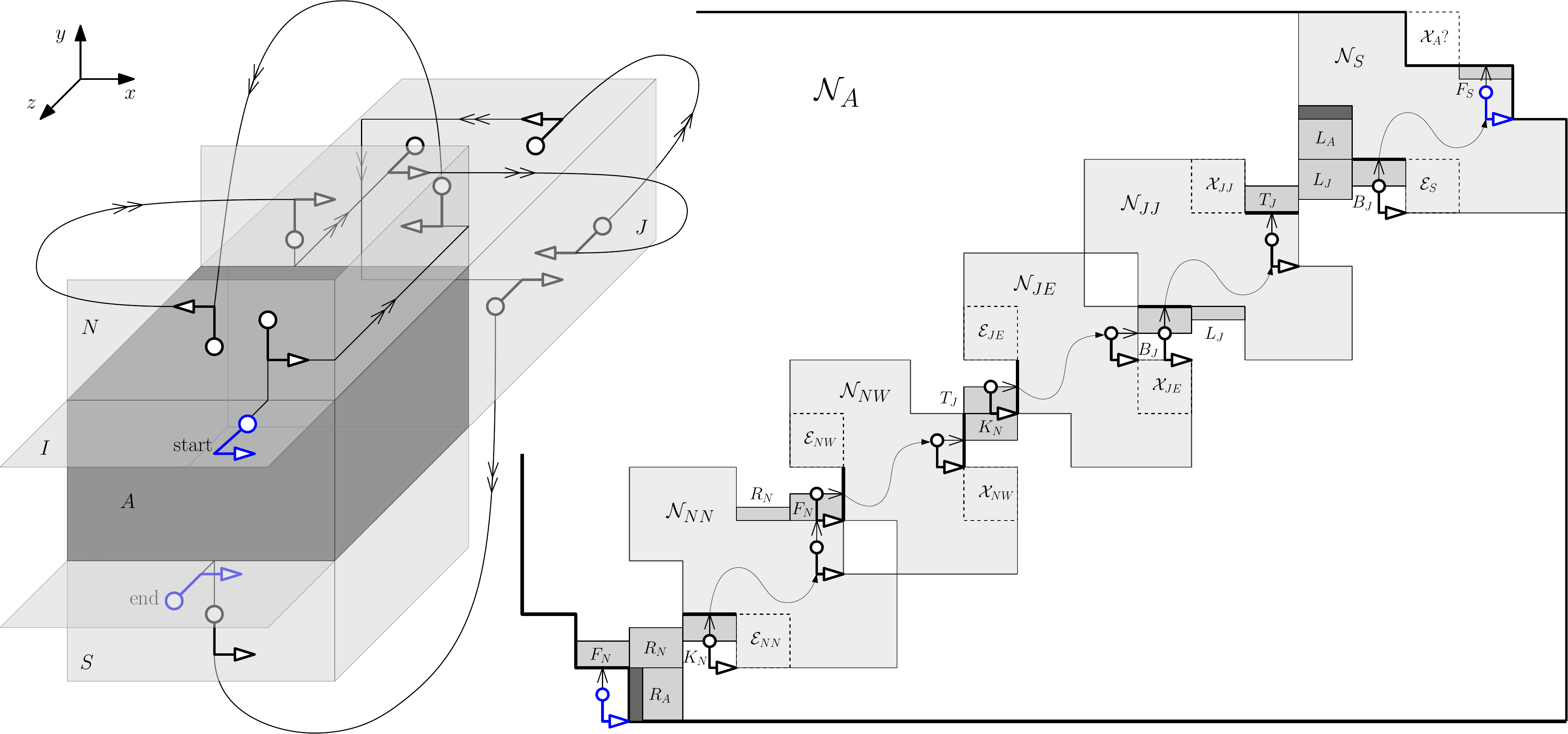}
\caption{Unfolding of degree-$4$ box $A$ with children $N$, $J$ and $S$, case $R_I$ closed (so $R_N$, $R_S$ open) and $L_J$ open; unfolding shown for the case when $NN$, $NW$, $JE$ and $JJ$ exist.}
\label{fig:NJSdegree4a-2}
\end{figure}
%

\medskip
\noindent
{\bf Case 4: $R_I$ closed and $L_J$ open.}  
Note that in this case $R_N$ and $R_S$ are open.  Consider the unfolding from~\autoref{fig:NJSdegree4a-2}, which handles the more general case when $NN$, $NW$, $JE$ and $JJ$ exist (handling cases when one or more of these boxes do not exist requires only minor modifications). Note that $\N_A$ provides a type-1 entry connection $e'_A \in F_N$. 
Also note that $\xleftarrow{x_S} \in L_I$ is not adjacent to $\T_S$, therefore $\N_S$ provides a type-1 exit connection, which is also a type-1 exit connection for $\N_A$ (since $x_A = x_S$). These together show that $\N_A$ satisfies invariant (I2).
The following observations support our claim that $\N_A$ satisfies invariant (I1):
\begin{itemize}
%
\item Since $\xrightarrow{e_A} \in R_I$ is closed, the unit square $\XE_A$ belongs to the inductive region of $A$.
\item The entry and exit ring faces for $NN$, $NW$, $JE$, $JJ$ and $S$ are as follows: 
$e_{NN} \in K_N$ and $x_{NN} \in F_N$; 
$e_{NW} \in F_N$ and $x_{NW} \in K_N$; 
$e_{JE} \in T_J$ and $x_{JE} \in B_J$; 
$e_{JJ} \in B_J$ and $x_{JJ} \in T_J$; 
and $e_S \in B_J$ and $x_S \in B_I$.
\item $\N_{NN}$, $\N_{NW}$, $\N_{JE}$ and $\N_S$ provide type-1 entry connections. This is because 
$\xrightarrow{e_{NN}} \in L_N$ is closed (since $NW$ exists),  
$\xrightarrow{e_{NW}} \in B_N$ is closed,
$\xrightarrow{e_{JE}} \in K_J$ is closed (since $JJ$ exists), and
$\xrightarrow{e_S} \in R_J$ is closed (since $JE$ exists). 
\item Since $\xrightarrow{e_{JJ}} \in L_J$ is open, the unit square $\XE_{JJ}$ (occupied by $\xrightarrow{e_{JJ}}$ in~\autoref{fig:NJSdegree4a-2}) does not belong to the inductive region for $JJ$.  
\item Since $\xleftarrow{x_{NN}} \in R_N$ is open, the unit square $\XX_{NN}$ (occupied by $\xleftarrow{x_{NN}}$ in~\autoref{fig:NJSdegree4a-2}) does not belong to the inductive region for $NN$.  
\item $\N_{NW}$, $\N_{JE}$ and $\N_{JJ}$ provide type-1 exit connections. This is because 
$\xleftarrow{x_{NW}} \in T_N$ is closed (since $NN$ exists), 
$\xleftarrow{x_{JE}} \in F_J$ is closed, and 
$\xleftarrow{x_{JJ}} \in R_J$ is closed (since $JE$ exists).
\end{itemize}
Arguments similar to the ones used in the previous case show that $\N_A$ satisfies invariant (I3) as well.

\section{Unfolding Degree-5 Nodes (Remaining Cases)}
\label{sec:degree5-appendix}
In this section we discuss the unfoldings for cases 5.2 through 5.4 listed in Section~\ref{sec:degree5}.

\begin{lemma}
\label{lem:NESJdegree5}
Let $A \in \T$ be a degree-$5$ node with parent $I$ and children $N$, $E$, $J$ and $S$ \emph{(Case $5.2$)}. 
If $A$'s children satisfy invariants (I1)-(I3), then $A$ satisfies invariants (I1)-(I3).
\end{lemma}
\begin{proof}
%
\begin{figure}[htpb]
\centering
\includegraphics[page=1,width=\textwidth]{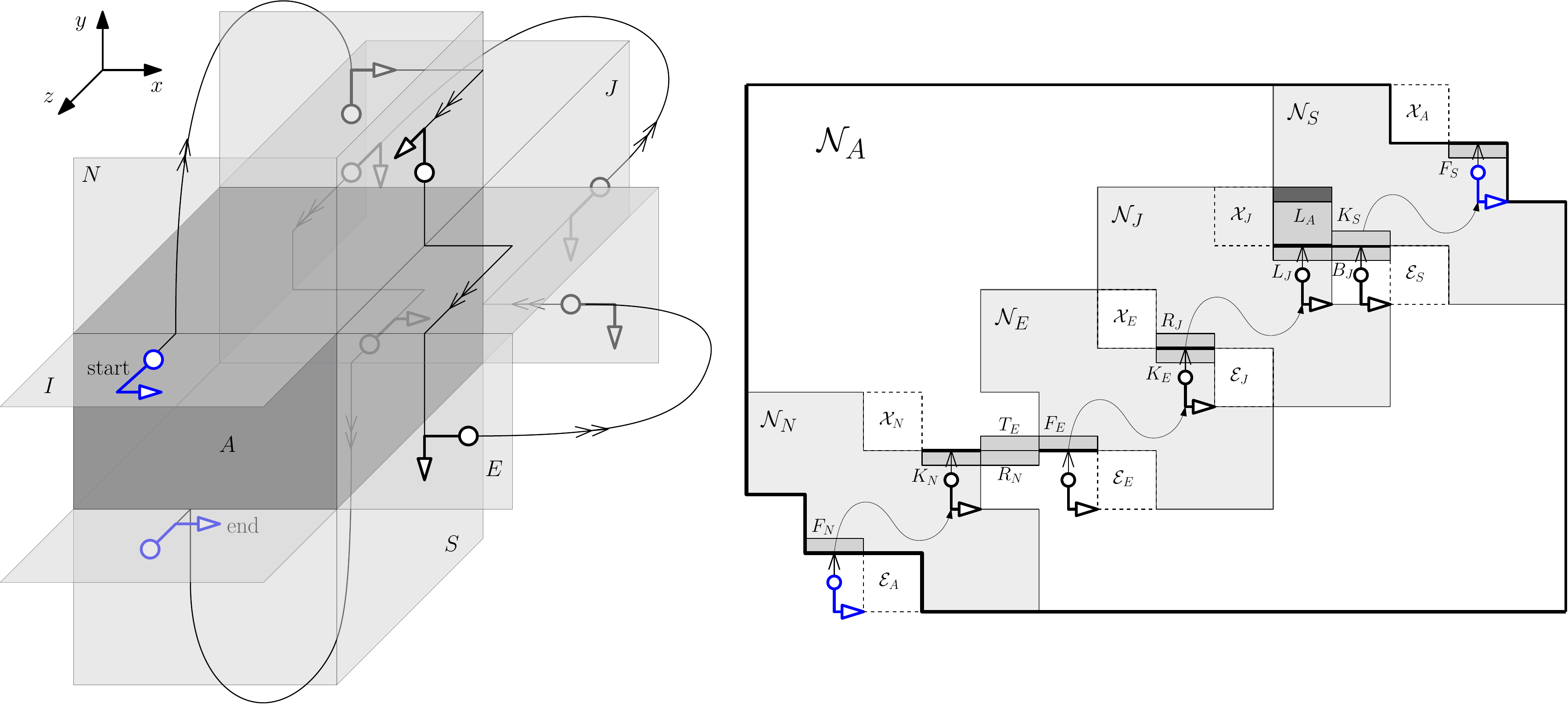}
\caption{\hand-east unfolding of degree-5 box $A$ with $N$, $E$, $J$ and $S$ children, case when $I$ and $J$ are both non-junctions.}
\label{fig:NESJdegree5}
\end{figure}
%
Arguments similar to the ones used in the proof of~\autoref{lem:degree5-connector} show that either $I$ and $J$ are both non-junctions, or else $N$ and $S$ are both non-junctions.

The unfolding  when $I$ and $J$ are both non-junctions is depicted in~\autoref{fig:NESJdegree5}. Note that this unfolding follows the same path as the one for the degree-4 case depicted in~\autoref{fig:NEJdegree4a}, up to the point where it reaches $L_J$, where it slides to $B_J$ to begin the recursive unfolding of $S$. Since these unfoldings and their correctness proofs are very similar, we only point out the differences here:
\begin{itemize}
\item Since $\xleftarrow{x_S} \in L_I$ is not adjacent to $\T_S$, $\N_S$ provides a type-1 exit connection $x'_S \in F_S$, which is also a type-1 exit connection for $A$.
\item The ring face of $J$ that lies on $B_J$ is not used in $\N_J$'s entry and exit connection, therefore it can be relocated outside of $\N_J$ (by invariant (I3) applied to $J$). We place it to the right of $x'_J \in L_J$ to serve as entry ring face for $\N_S$.
\item Since $\xrightarrow{e_S} \in R_J$ is not adjacent to $\T_S$, $\N_S$ provides a type-1 entry connection $e'_S \in K_S$.
\end{itemize}

%
\begin{figure}[htpb]
\centering
\includegraphics[page=2,width=0.9\textwidth]{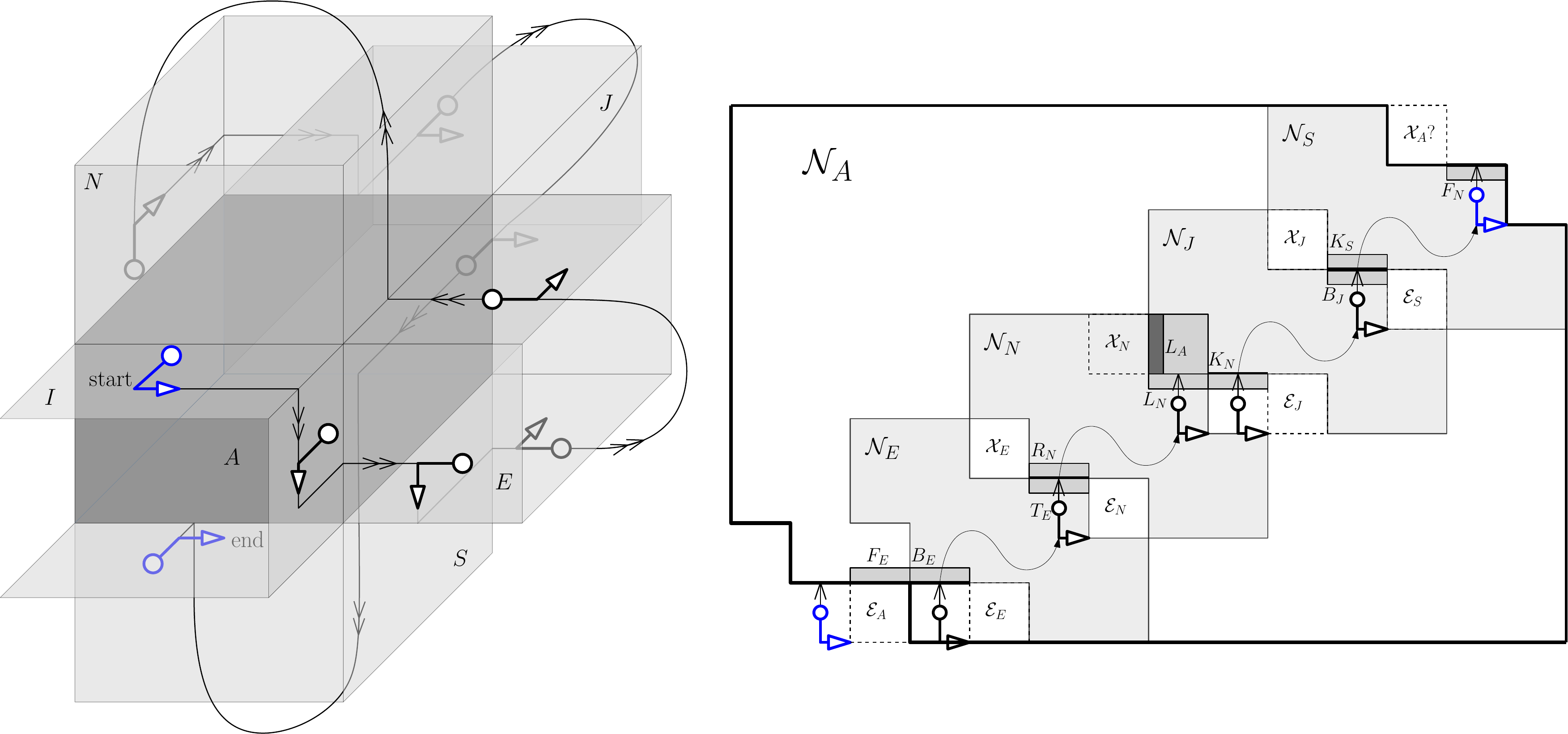}
\caption{\hand-east unfolding of degree-5 box $A$ with $N$, $E$, $J$ and $S$ children, case when $N$ and $S$ are both non-junctions.}
\label{fig:NESJdegree5-3}
\end{figure}
%
Assume now that $N$ and $S$ are both non-junctions. The unfolding for this case is depicted in~\autoref{fig:NESJdegree5-3}.  Because
$\xrightarrow{e_A} \in R_I$ is open and adjacent to
$\T_A$, $\N_A$ has the option of providing a type-2 entry connection, which it does by placing $\xrightarrow{e'_A} \in F_E$
adjacent to the entry port extension.  It also provides a type-1
exit connection $x'_A = x'_S \in F_S$, and therefore $\N_A$ satisfies invariant (I2).  Also note that (I3) is satisfied, since $A$'s ring face on $L_A$ (darkened in~\autoref{fig:NESJdegree5-3}) is not used in entry or exit connections and can be removed without disconnecting $\N_A$. The following observations support our claim that $\N_A$ is connected and satisfies invariant (I1):
\begin{itemize}
\item The entry and exit ring faces for $E$, $N$, $J$ and $S$ are as follows: 
$e_E \in R_S$ and $x_E \in R_N$; 
$e_N \in T_E$ and $x_N \in L_A$; 
$e_J \in K_N$ and $x_J \in K_S$; 
and $e_S \in B_J$ and $x_S \in B_I$.
\item All children nets provide type-1 entry connections (by invariant (I2)). This is because 
$\xrightarrow{e_E} \in K_S$ is not adjacent to $\T_E$, 
$\xrightarrow{e_N} \in K_E$ is not adjacent to $\T_N$, 
$\xrightarrow{e_J} \in R_N$ is not adjacent to $\T_J$,  and
$\xrightarrow{e_S} \in R_J$ is not adjacent to $\T_S$.
In addition, all children provide type-1 exit connections because
$\xleftarrow{x_E} \in F_N$ is not adjacent to $\T_E$,
 $\xleftarrow{x_N} \in F_A$ is closed,
 $\xleftarrow{x_J} \in L_S$ is not adjacent to $\T_J$, and 
 $\xleftarrow{x_S} \in L_I$ is 
 not adjacent to $\T_S$.
 \item Ring faces that lie on $F_E$ and $K_N$ can be relocated anywhere outside of $\N_E$ and $\N_N$ respectively, by invariant (I3), noting that none of these ring faces are used in any entry or exit connections.
 \item The exit connection $x'_N \in L_N$, because $N$ is a non-junction and $L_N$ is open. Thus $L_A$ is attached to $x'_N \in L_N$ in the unfolding.
\end{itemize}
This concludes the proof. 
\end{proof}

\begin{lemma}
\label{lem:NWSJdegree5}
Let $A \in \T$ be a degree-$5$ node with parent $I$ and children $N$, $W$, $J$ and $S$ \emph{(Case $5.3$)}. 
If $A$'s children satisfy invariants (I1)-(I3), then $A$ satisfies invariants (I1)-(I3). 
\end{lemma}
\begin{proof}
Arguments similar to the ones used in the proof of~\autoref{lem:degree5-connector} show that either $I$ and $J$ are both non-junctions, or else $N$ and $S$ are both non-junctions.

\begin{figure}[htpb]
\centering
\includegraphics[page=1,width=0.8\textwidth]{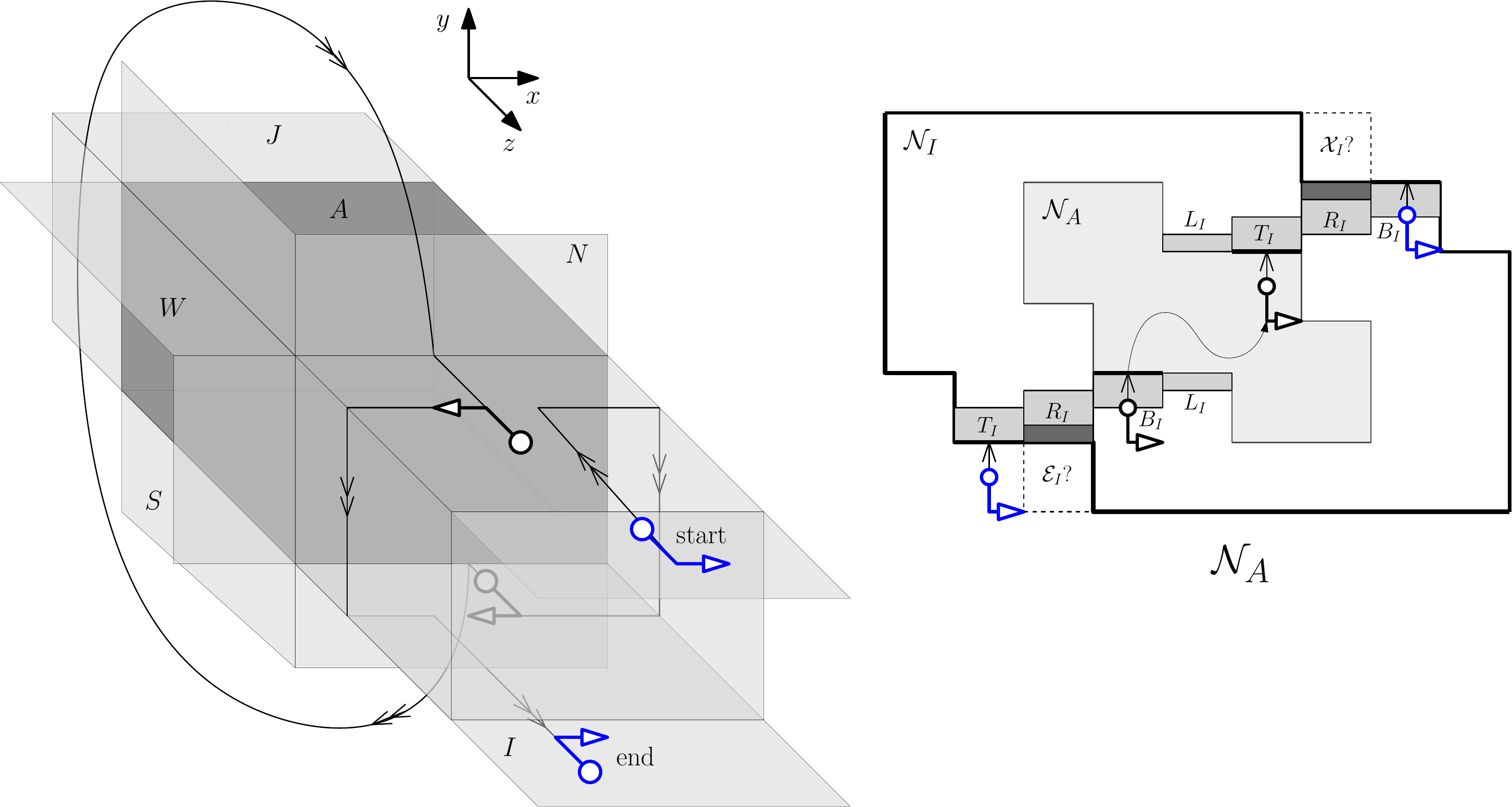}
\caption{Unfolding of degree-5 box $A$ with $N$, $W$, $J$ and $S$ children, case when $I$ and $J$ are both non-junctions.}
\label{fig:NESJdegree5-2}
\end{figure}
%

The unfolding for the case when $I$ and $J$ are both non-junctions can be reduced to the case from~\autoref{fig:NESJdegree5} using the method depicted in~\autoref{fig:NESJdegree5-2}. In this case, $I$'s unfolding is handled specially, so we describe the recursive unfolding of $I$ assuming $I$ is in standard position (with $A$ in the back).  The unfolding path starts at the top front edge of $I$ and cycles clockwise  to $I$'s bottom back edge, which is $A$'s entry port.   By using this bottom entry port, $A$'s unfolding is a horizontal reflection of that in~\autoref{fig:NESJdegree5}.  After unfolding $A$, the unfolding path cycles counter-clockwise from the top back edge of $I$ (which is $A$'s exit port) to $I$'s bottom front edge.

Observe in the unfolding shown in~\autoref{fig:NESJdegree5-2} that $\N_I$ provides type-1 entry and exit connections, $I$'s ring faces not used in entry or exit connections (shown darkened) can be removed without disconnecting $\N_I$, and all open faces of $I$ are unfolded. We have already established that $\N_A$ satisfies invariants (I1)-(I3) 
and provides type-1 entry and exit connections, which connect to    
the pieces $B_I$ and $T_I$ placed adjacent to them. Thus $\N_I$ satisfies invariants (I1)-(I3).

%
\begin{figure}[h]
\centering
\includegraphics[page=2,width=0.9\textwidth]{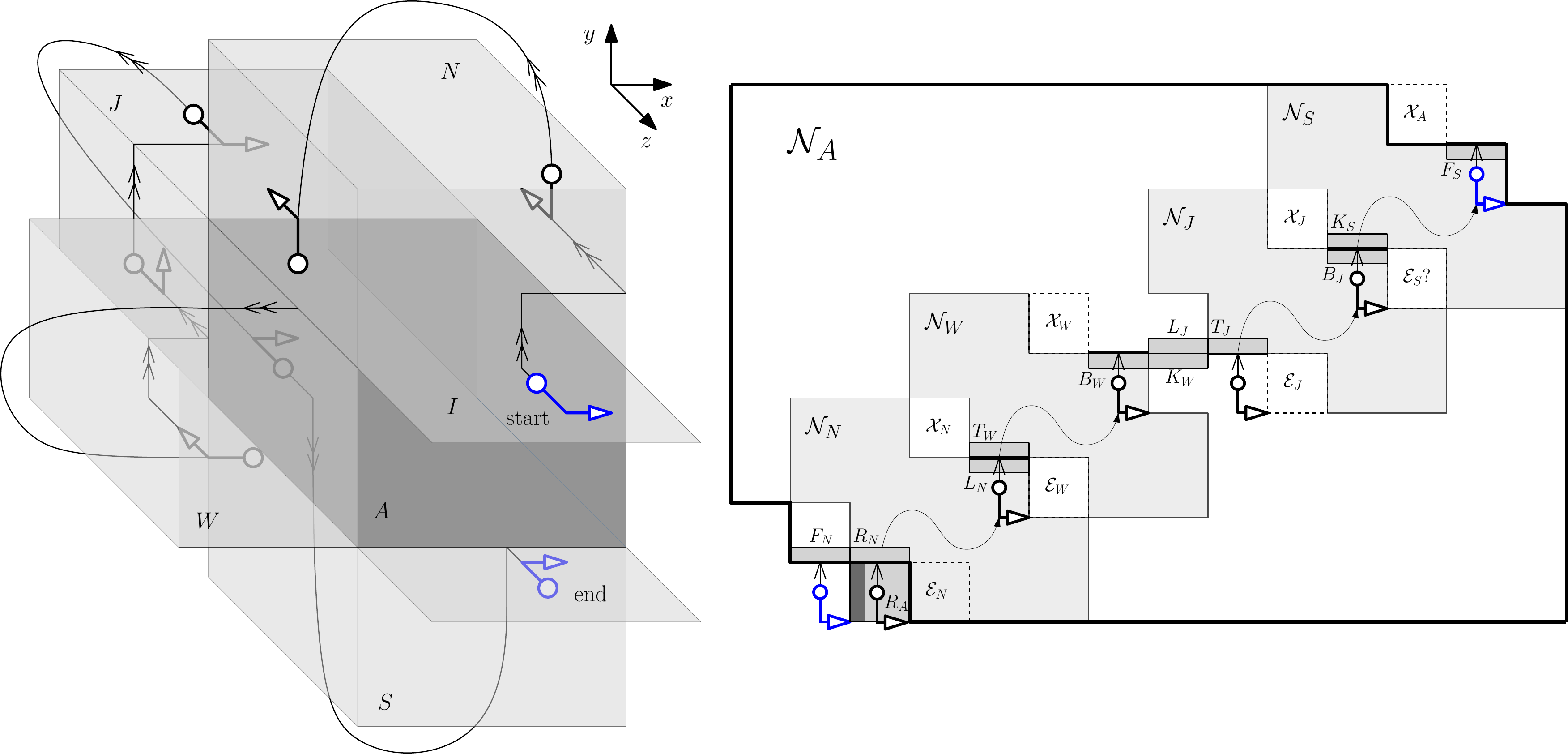}
\caption{Unfolding of degree-5 box $A$ with $N$, $W$, $J$ and $S$ children, case when $N$ and $S$ are both non-junctions and $I$ has an east neighbor.}
\label{fig:NESJdegree5-4}
\end{figure}
%
Consider now the case when $N$ and $S$ are both non-junctions.  If $I$ has a neighbor to its east, then the
unfolding is as depicted in~\autoref{fig:NESJdegree5-4}. The net 
$\N_A$ provides a type-1 entry connection $e'_A \in F_N$ and a type-1
exit connection $x'_A = x'_S \in F_S$. Therefore $\N_A$ satisfies invariant (I2).  Also note that (I3) is satisfied, since $A$'s ring face on $R_A$ (darkened in~\autoref{fig:NESJdegree5-4}) is not used in entry or exit connections and can be removed without disconnecting $\N_A$.   Because $R_I$ is closed, $\XE_A$ is part of $\N_A$'s inductive region and therefore the face $R_A$ can be placed there.
The following observations support our claim that $\N_A$ is connected and satisfies invariant (I1):
\begin{itemize}
\item The entry and exit ring faces for $N$, $W$, $J$ and $S$ are as follows: 
$e_N \in R_A$ and $x_N \in T_W$; 
$e_W \in L_N$ and $x_W \in L_S$; 
$e_J \in K_N$ and $x_J \in K_S$; 
and $e_S \in B_J$ and $x_S \in B_I$.
\item All children nets provide type-1 entry connections (by invariant (I2)). This is because 
$\xrightarrow{e_N} \in K_A$ is closed, 
$\xrightarrow{e_W} \in K_N$ is not adjacent to $\T_W$, 
$\xrightarrow{e_J} \in R_N$ is not adjacent to $\T_J$,  and
$\xrightarrow{e_S} \in R_J$ is not adjacent to $\T_S$.
In addition, all children provide type-1 exit connections because
$\xleftarrow{x_N} \in F_W$ is not adjacent to $\T_N$,
 $\xleftarrow{x_W} \in F_S$ is not adjacent to $\T_W$,
 $\xleftarrow{x_J} \in L_S$ is not adjacent to $\T_J$, and  
 $\xleftarrow{x_S} \in L_I$ is not adjacent to $\T_S$.
 \item Ring faces that lie on $F_N$, $L_J$ and $K_W$ can be relocated anywhere outside of $\N_N$, $\N_J$ and $\N_W$ respectively, by invariant (I3), noting that none of these ring faces are used in any entry or exit connections.
\end{itemize}

If $I$ does not have a neighbor to its west, then $I$ is a non-junction and the unfolding can be reduced to the case depicted in~\autoref{fig:NESJdegree5-3} using the technique outlined in~\autoref{fig:NESJdegree5-2}:  the path cycles around $I$ to $B_I$ and $A$ is unfolded using 
a horizontal reflection of \autoref{fig:NESJdegree5-3}.  The proof that this satisfies invariants (I1)-(I3) is similar to the one used for the unfolding in ~\autoref{fig:NESJdegree5-2}, noting that here $\N_A$ has a type-2 entry connection that attaches to
$L_I$ in~\autoref{fig:NESJdegree5-2}.
\end{proof}

\begin{lemma}
Let $A \in \T$ be a degree-$5$ node with parent $I$ and children $E$, $W$, $J$ and $S$ \emph{(Case 5.4)}. 
If $A$'s children satisfy invariants (I1)-(I3), then $A$ satisfies invariants (I1)-(I3).
\end{lemma}
\begin{proof}
Arguments similar to the ones used in the proof of~\autoref{lem:degree5-connector} show that either $E$ and $W$ are both non-junctions, or else $I$ and $J$ are both non-junctions.

%
\begin{figure}[h]
\centering
\includegraphics[page=1,width=0.95\textwidth]{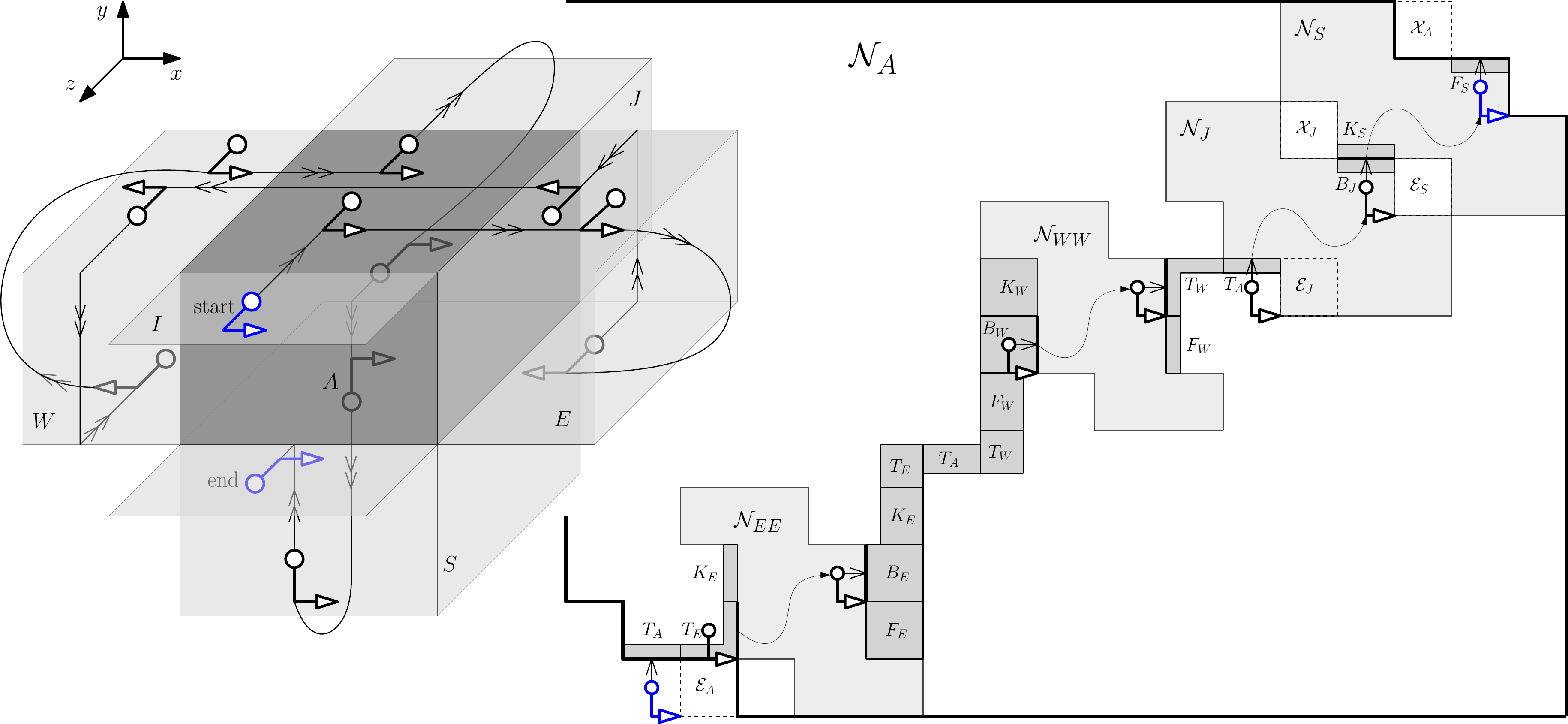}
\caption{Unfolding for box $A$ of degree $5$ with $E$, $W$, $J$ and $S$ children, case when $E$ and $W$ are both non-junctions.}
\label{fig:JEWSdegree5}
\end{figure}
The unfolding for the case when $E$ and $W$ are both non-junctions is depicted in~\autoref{fig:JEWSdegree5}. Note that $\N_A$ provides a type-1 entry connection $e'_A \in T_A$ and a type-1 exit connection $x'_A \in F_S$, therefore $\N_A$ satisfies invariant (I2). Also note that (I3) is trivially satisfied, since $A$ has a single open ring face $e'_A \in T_A$ that  is an entry connection. 
The following observations support our claim that $\N_A$ is connected and satisfies invariant (I1):
\begin{itemize}
\item The entry and exit ring faces for $EE$, $WW$, $J$ and $S$ are as follows: 
$e_{EE} \in T_E$ and $x_{EE} \in B_E$; 
$e_{WW} \in B_W$ and $x_{WW} \in T_W$; 
$e_J \in T_A$ and $x_J \in K_S$; 
and $e_S \in B_J$ and $x_S \in B_I$.
\item $\N_J$ and $\N_S$ provide type-1 entry and exit connections. This is because $\xrightarrow{e_J} \in R_A$ is closed, 
$\xleftarrow{x_J} \in L_S$ is not adjacent to $\T_J$, and $\xrightarrow{e_S} \in R_J$ and $\xleftarrow{x_S} \in L_I$ are open but not adjacent to $\T_S$. 
\item Since $\xrightarrow{e_{EE}} \in K_E$ and $\xleftarrow{x_{EE}} \in F_E$ are open, the unit squares $\XE_{EE}$ and $\XX_{EE}$ (occupied in~\autoref{fig:JEWSdegree5} by $\xrightarrow{e_{EE}}$ and $F_E$, respectively) do not belong to the inductive region for $EE$. 
\item Since $\xrightarrow{e_{WW}} \in K_W$ and $\xleftarrow{x_{WW}} \in F_W$ are open, the unit squares $\XE_{WW}$ and $\XX_{WW}$ (occupied in~\autoref{fig:JEWSdegree5} by $K_W$ and  $\xleftarrow{x_{WW}}$, respectively) do not belong to the inductive region for $WW$. 
\end{itemize}
If $I$ are $J$ are non-junctions, then we use the unfolding from~\autoref{fig:reduce}, with the understanding that $\N'_A$ is the net from~\autoref{fig:NESJdegree5}. Note that $\N'_A$ provides type-1 entry and exit connections, which implies that the net $\N_A$ from~\autoref{fig:reduce} provides type-2 entry and exit connections. Since $\xrightarrow{e_A} \in R_I$ and $\xleftarrow{x_A} \in L_I$ are open and adjacent to $\T_A$, and $\N'_A$ satisfies invariants (I1)-(I3), we conclude that $\N_A$ satisfies invariants (I1)-(I3). 
\end{proof}

\section{Another Complete Unfolding Example}
\label{sec:example2}
We conclude this paper with another complete unfolding example that incorporates some of the cases presented in the appendices (which could not be included in the first example from Section~\ref{sec:example1}).
%
\begin{figure*}[h]
\centering
\includegraphics[page=2,width=\linewidth]{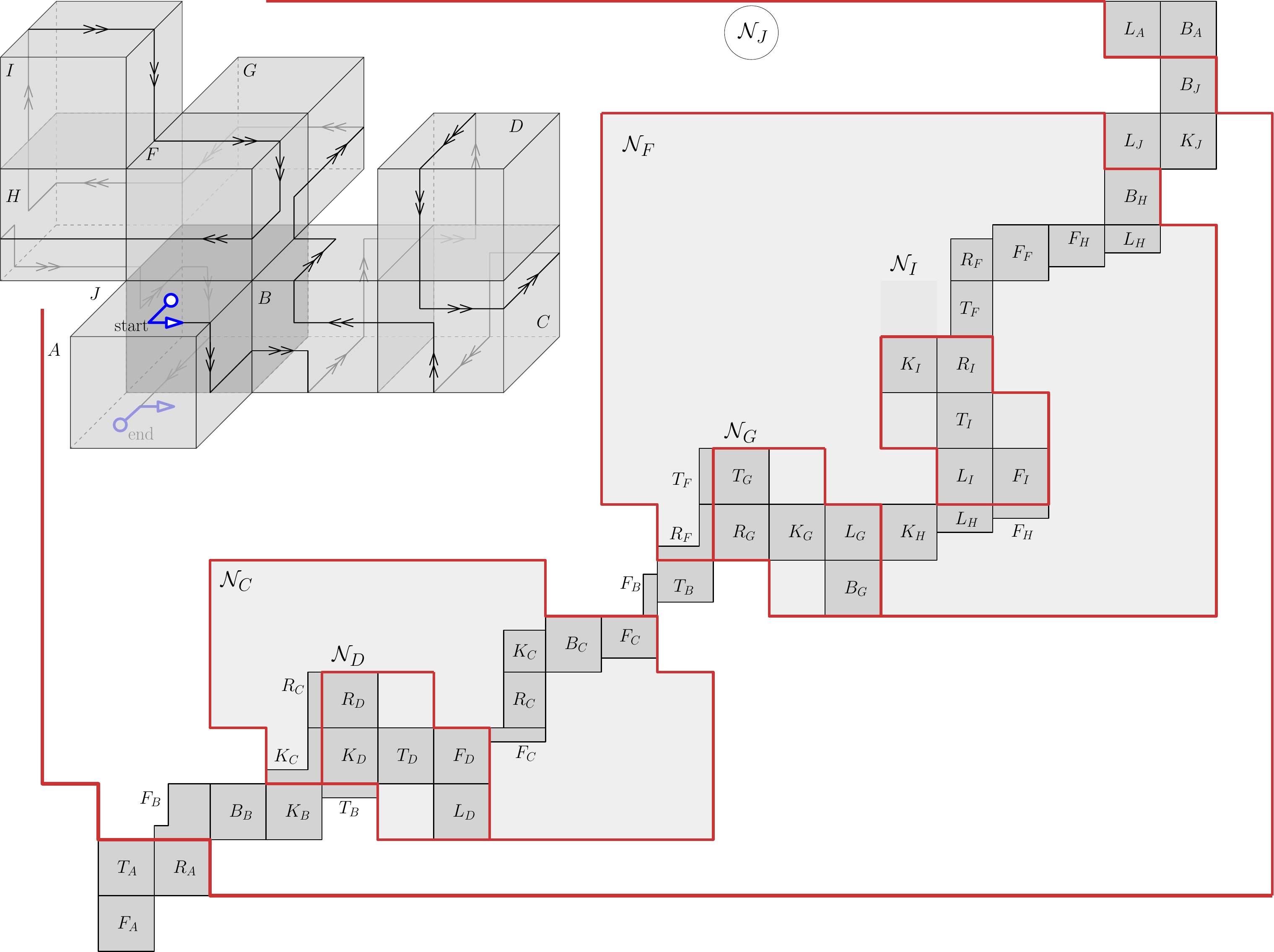}
\caption{Unfolding of polycube tree with root $A$ (back child of $A$ is $J$). 
}
\label{fig:example2}
\end{figure*}
%
We use as running example a polycube tree composed of nine boxes, depicted in~\autoref{fig:example2}.
The root $A$ of the the unfolding tree is a degree-1 box with back child $J$, which is unfolded recursively. The unfolding of $J$ follows the pattern depicted in~\autoref{fig:NEdegree3a}b, slightly adjusted to accommodate for the fact that $J$ does not have a south-east child. The east-east child of $J$ (labeled $C$ in~\autoref{fig:example2}) follows the unfolding pattern depicted in~\autoref{fig:EWdegree2}a.
The north child of $J$ (labeled $F$ in~\autoref{fig:example2}) follows the unfolding pattern from~\autoref{fig:NEdegree3b}b, 
traversed on reverse (note that the subtree rooted at $F$ is a horizontal mirror plane reflection of the case depicted in~\autoref{fig:NEdegree3b}b, after a clockwise $90^\circ$-rotation about a vertical axis followed by a clockwise $90^\circ$-rotation about a horizontal axis, to bring it in standard position).  Finally, the leaves are unfolded as in~\autoref{fig:degree1}. 

\end{appendix}

\end{document}